\documentclass{article}
\usepackage{graphicx} 
\usepackage[a4paper, margin=1in]{geometry} 
\RequirePackage{etex}
\usepackage[english]{babel}
\usepackage{subcaption}
\usepackage{mathrsfs}
\usepackage[all]{xy}
\usepackage{color}
\usepackage{url}
\usepackage{indent first}
\usepackage[labelfont=bf,labelsep=period,justification=raggedright]{caption}
\usepackage{quantikz}
\usepackage{amsthm}
\usepackage{amsmath}
\usepackage{amsfonts}
\usepackage{amssymb}
\usepackage{tikz}
\usetikzlibrary{arrows.meta, positioning, bending}
\usepackage{float}
\usepackage{xcolor}
\usepackage{subcaption}
\theoremstyle{plain}
\newtheorem{Theorem}{Theorem}[section]
\newtheorem{cor}[Theorem]{Corollary}

\newtheorem{prop}[Theorem]{Proposition}
\newtheorem{lemma}[Theorem]{Lemma}

\newtheorem{defn}[Theorem]{Definition}

\newtheorem{remark}[Theorem]{Remark}
\newtheorem{ex}[Theorem]{Example}

\usepackage{csquotes}
\makeatletter
\renewcommand{\theTheorem}{%
  \thesection
  \ifnum\value{subsection}>0
    .\arabic{subsection}%
  \fi
  .\arabic{Theorem}}
\makeatother
\usepackage{graphicx} 
\usepackage[utf8]{inputenc}
\usepackage[ backend=biber]{biblatex}
\usepackage{abstract} 
\usepackage[colorlinks=true, linkcolor=blue, citecolor=blue, urlcolor=blue]{hyperref}
\usepackage{authblk} 
\renewbibmacro{in:}{}
\title{Construction of Partial Join Graphs with Perfect State Transfer
in Shunt Decomposition-Based Quantum Walks}
                  
\author[1]{Banita Katuwal\thanks{Corresponding author. Email: banitakatuwal@sssihl.edu.in}}
\author[1]{Y~Lakshmi Naidu\thanks{Email: ylakshminaidu@sssihl.edu.in}}
\author[1]{Srinath M~S\thanks{Email: srinathms@sssihl.edu.in}}
\author[2]{Supriyo Dutta\thanks{Email:dosupriyo@gmail.com}}

\affil[1]{Department of Mathematics and Computer Science\\
Sri Sathya Sai Institute of Higher Learning\\
Andhra Pradesh, India -- 515001}

\affil[2]{Department of Mathematics\\
National Institute of Technology Agartala\\
Jirania, West Tripura, India -- 799046}
\date{}
\begin{document}

\maketitle
\begin{abstract}
In this paper, we define directed partial join graphs with signed couplings and construct
discrete-time quantum-walk transition operators for these graphs using
the shunt-decomposition framework. The resulting transition operators
apply to several important graph families, including complete graphs
with loops, circulant partial joins, complete bipartite graphs, tensor
powers of complete bipartite graphs,  all with signed couplings. For each family,
we identify the corresponding structure of the transition operator and
derive  necessary and sufficient conditions for periodicity  and perfect state transfer (PST), when one of the two directed regular graphs admits PST. Based on these results, we identify two types of state transfer; internal PST, which occurs between vertices within the same graph, and coupling PST, which occurs between two components of join graphs. We further develop a double-cover construction for directed partial join
graphs and derive conditions for periodicity and PST
when the associated transition operators do not necessarily commute.
Using this construction, we establish PST results for double covers of
complete graphs with loops. In particular, we provide an example in
which the complete graph \(K_n\) does not exhibit PST for \(n\geq4\),
whereas a suitable partial join of \(K_n\) exhibits PST when
\(n=2^m\), \(m\geq2\). Hence, these results extend the class of graph families
admitting PST in shunt-decomposition-based quantum walks
and provide a unified framework for studying quantum state transfer in
graph joins, products, and covers. \\~\\


\end{abstract}
\section{Introduction}
Quantum walks are the quantum analogues of classical random walks~\cite{portugal2013quantum}. It has 
emerged as one of the most versatile framework in quantum computing schemes~\cite{childs2009universal,lovett2010universal,childs2013universal},
quantum algorithm design~\cite{portugal2013quantum,childs2003exponential,Ambainis2003},  quantum cryptography~\cite{Vlachou2015},  quantum key
distribution~\cite{vlachou2018quantum}, and many more. A classical random walk describes the probabilistic
evolution of a particle hopping between the vertices of a graph. A quantum
walk describes the evolution of a quantum state through a unitary process
determined by the structure of the underlying graph. 
 
Quantum walks are broadly classified into two models:
continuous-time quantum walks (CTQWs)~\cite{CoutinhoGodsil2016,FarhiGutmann1998} and  discrete-time
quantum walks (DTQWs)~\cite{Ambainis2003,GodsilZhan2019}. In a CTQW, the evolution is governed by a
Hamiltonian, typically derived from the adjacency matrix or the Laplacian
matrix of the graph. In contrast, a DTQW evolves through the repeated
application of a unitary coin operator followed by a shift operator,
producing a discrete-time unitary evolution. Several  models of
DTQWs have been developed, each motivated by different physical settings
and mathematical frameworks. These include the arc-reversal, two-reflection, sedentary,  vertex-face, and 
shunt-decomposition models~\cite{GodsilZhan2019,Godsil_Zhan_2023,zhan2021quantum}.
 
Among these, we focus on the shunt-decomposition model, also known as the shunt-decomposition walk~\cite[Chapter 7]{Godsil_Zhan_2023}. This model was first introduced by Aharonov et al.~\cite{Aharonov2001} and was later reformulated in a combinatorial framework by Godsil and Zhan~\cite{GodsilZhan2019}. It provides an elegant combinatorial description of quantum walks in which the evolution alternates between a coin operator and a shift operator that moves the walker along outgoing arcs belonging to a prescribed shunt (or arc class). This model has been extensively studied for its spectral properties and mixing behavior, including characterizations of uniform average mixing when the Grover coin is employed~\cite{Godsil_Zhan_2023}. Owing to its simple yet expressive structure, this model is particularly well suited for quantum circuit implementations~\cite{WingBocanegra2023}. Furthermore, it has found applications in a variety of areas of quantum information processing, including quantum search algorithms~\cite{WingBocanegra2025}, quantum channels~\cite{Katuwal2026}, and other quantum computing tasks~\cite{sato2024circuit}.  Beyond their algorithmic applications, quantum walks provide a natural framework for modeling the transport of quantum information across networks of interacting qubits~\cite{portugal2013quantum}. Such networks can be represented by simple, undirected graphs, where vertices correspond to qubits and edges represent interactions between them. This graph-theoretic viewpoint has established quantum walks as a powerful tool for studying quantum communication protocols~\cite{panda2023quantum} and information transfer in quantum networks~\cite{mulken2011continuous}.

A fundamental concept in quantum communication assisted by quantum walks
is PST~\cite{christandl2004perfect}, which
describes the lossless transfer of a quantum state between two vertices
of a graph. A quantum walk on a graph $G$ is said to exhibit PST from
a vertex $a$ to a vertex $b$, where $a,b\in V(G)$, if there exists a
time at which a quantum state initially localized at $a$ evolves to the
state localized at $b$, up to a global phase. In the discrete-time
setting~\cite{Godsil_Zhan_2023}, this condition can be written as
\(
U_G^{k}e_a=\gamma e_b, \quad \left|\langle e_b|U^{k}|e_a\rangle\right|=1,
\quad |\gamma|=1,
\)
for some positive integer $k$, where $U_G$ is the evolution operator of
the quantum walk and $e_a,e_b$ are the basis states associated with the
vertices $a$ and $b$, respectively. 
In the continuous-time setting~\cite{godsil2012state}, the evolution is
generated by the adjacency matrix $A$ of $G$, and PST from $a$ to $b$
at time $t\in\mathbb{R}$ is characterized by
\(
\left|\left\langle e_b,e^{-iAt}e_a\right\rangle\right|=1.
\)
 PST on simple graphs is a rare phenomenon, as only special classes of graphs admit it. Nevertheless, when it occurs, PST provides an ideal mechanism for transmitting quantum information with unit fidelity and without requiring active control during the evolution.  A closely related concept is periodicity. A quantum walk is periodic at a vertex $a$ if the quantum state returns to $a$, up to a global phase, after a finite evolution time~\cite{Godsil2011}. Since a walk exhibiting PST between two vertices must eventually return to its initial state, periodicity is a necessary condition for PST~\cite{ChanZhan2023,Godsil2011}. Accordingly, the study of PST often begins with an analysis of the periodicity of the underlying quantum walk.

A substantial body of research has identified graph families admitting PST
in both continuous- and discrete-time quantum walks~\cite{angeles2009perfect,godsil2012state,CoutinhoGodsil2016,angeles2009perfect,ChanZhan2023, Katuwal2026,coutinho2016perfect,Katuwal2026pst}. 
In the continuous-time quantum walk framework, considerable attention
has been devoted to graph constructions based on joins and graph
products, as these operations provide systematic methods for generating
larger graph families. In Hamiltonian-based models, PST is governed primarily by the spectral properties of the
adjacency or Laplacian matrices, including eigenvalue spacing,
commensurability, and phase alignment.

The study of PST on join graphs began with the work of Bose
\emph{et al.}~\cite{Bose2009}, who showed that the graph obtained by
removing a single edge from the complete graph,
\(
K_n\setminus e \;\cong\; \overline{K_2}\vee K_{n-2},
\)
admits Laplacian PST between the two non-adjacent
vertices whenever $n\equiv0\pmod{4}$. This construction demonstrated
that a simple graph join can induce PST even though the complete graph
itself does not exhibit this property. Subsequently, Angeles-Canul
\emph{et al.}~\cite{angeles2009perfect} generalized this result by
establishing sufficient conditions for adjacency-based PST in the
unweighted joins $X\vee Y$, where $X\in\{\overline{K_2},K_2\}$ and $Y$
is a regular graph. They further identified several families of
join graphs admitting PST, including circulant joins $G+_{C}G$, integral
circulant graph families
\(
\mathrm{ICG}_n\big(\{2,n/2^{b}\}\cup Q\big),
\)
where $b\in\{1,2\}$, $n$ is a multiple of $16$, and $Q$ is a subset of the
odd divisors of $n$; $n$-fold self-joins $G^{+n}$; $\mathrm{ICG}_n(\{1,n/2\})^{+m}$;
joins of known integral circulant and hypercubic PST graphs,
$\bigsqcup_{k=1}^{m}G_k$.  

Since then, numerous join and product graph constructions admitting PST have been developed. To the best of our knowledge, the graph families summarized in Table~\ref{tab:known_join_pst} represent the principal join and product constructions known to admit PST in the continuous-time quantum walk framework. More recently, Kirkland and Monterde~\cite{kirkland2026quantum} demonstrated that the join operation itself can induce PST even when neither of the constituent graphs admits PST individually. In addition to these works, several other studies have investigated PST on graph joins and products from different perspectives~\cite{godsil2025perfect, coutinho2016perfect,godsil2021sedentary}, further enriching the theory and expanding the range of graph families known to exhibit PST.

\begin{table}[h]
\centering

\renewcommand{\arraystretch}{0.98}
\scalebox{0.85} {
\begin{tabular}{|c| p{11cm} |c|}
\hline
\textbf{No.} & \textbf{Graph / Construction} & \textbf{Refs} \\
\hline
3&Double cone over a $k$-regular graph with suitably weighted apex edge and/or apex self-loops; Weighted Cartesian products (Hamming graphs); PST between every pair of vertices and between uniform superpositions on arbitrary hypercube subcubes& ~\cite{angeles2009quantum}\\
\hline
4&Weighted double cones $K_2+G$ (connected graphs);
double half-cones $K_1+G\circ G+K_1$ with circulant connections;
cylindrical cones $K_1+G+\overline{K_n}+G+K_1$ (shown not to admit PST)&~\cite{ge2011perfect}\\
\hline
5&Unweighted double cone $\overline{K_2}+G$ (Laplacian); PST between the apex vertices iff $|V(G)|\equiv 2 \pmod{4}$; no PST for $K_2+G$&~\cite{alvir2016perfect}\\
\hline
6& Semi-Cayley graphs $SC(G,R,L,G)$, which is the joins of Cayley graphs $\mathrm{Cay}(G,R)$ and $\mathrm{Cay}(G,L)$ over an abelian group $G$, where $R$ and $L$ are the inverse-closed connection sets of the two Cayley graphs & ~\cite{arezoomand2023perfect}\\
\hline 
7&Signed join $(-K_2)+G$ ($G$ $(n,3)$-regular); signed complete graphs; double covers of signed graphs; and exterior powers&~\cite{2481614.2481624}\\
\hline
\end{tabular}
}
\caption{Known join graphs constructions admitting PST in the continuous-time quantum walk framework.}
\label{tab:known_join_pst}
\end{table}

Although PST on join graphs is well understood in the continuous-time setting, the corresponding discrete-time case remains largely unexplored. Unlike continuous-time quantum walks, where PST is determined by the spectral properties of adjacency or Laplacian matrices, discrete-time PST depends on the eigenphase structure of the unitary evolution operator. In particular, transfer occurs only when the relevant eigenphases synchronize while transitions between different join components are suppressed. Moreover, standard graph joins often introduce dimensional inconsistencies in coin-and-shift quantum walk models. The partial join construction considered here overcomes this limitation by preserving degree compatibility through controlled interconnections between component graphs. As a result, it substantially expands the class of graphs admitting PST in discrete-time quantum walks and provides new insights into the interplay between graph structure and quantum information transport.

The remainder of this paper is organized as follows. Section~\ref{sec2}
presents the preliminaries of the shunt-decomposition-based
discrete-time quantum-walk framework and reviews the notions of
periodicity and PST. Section~\ref{sec3}
introduces several classes of partial join graphs with signed coupling based on their
block structure. Section~\ref{sec4} constructs circulant partial joins
and discusses directed partial joins and tensor product graphs, whose
block-form adjacency matrices fall into different cases.
Section~\ref{sec5} constructs the transition operator for directed
partial join graphs and derives conditions under which the transition
matrix $U_G$ associated with a directed partial join of two
$d$-regular directed graphs is unitary. Section~\ref{sec6} derives
spectral conditions for periodicity and PST in directed partial joins
when all block transition matrices are identical. These results are
then applied to several graph families with even $n$, including the
complete graph with loops $K_n^{\circlearrowleft}$, the complete bipartite graph
$K_{n,n}$, the tensor product graphs $(K_{n,n})^{\otimes n}$ and
$K_2^{\otimes n}$, and the circulant join graph $C_n(S)$ with
$S=\{1,\ldots,n/2\}$, all with signed coupling.
Section~\ref{sec7} considers the third case, in which the block
transition matrices of the directed partial join do not necessarily
commute. We investigate the corresponding double-cover construction
and derive conditions for periodicity and PST in this setting. By
combining these results with those obtained in Section~\ref{sec6}, we
establish conditions for periodicity and PST in double covers of
complete graphs with loops. Section~\ref{sec8} presents an example in which the complete graph
\(K_n\), \(n\geq4\), does not exhibit PST, whereas a suitable partial
join of \(K_n\) exhibits PST. Finally, Section~\ref{sec9} concludes the
paper with a summary of the main results.

\section{Shunt Decomposition Walks and Perfect State Transfer(PST)}\label{sec2}

A \emph{shunt} on a directed graph $G$ is a permutation on its vertex set in which every vertex is mapped to one of its out-neighbors. A \emph{shunt decomposition} is a collection of such shunts whose corresponding arc sets partition the arc set of $G$. If $G$ is $d$-regular directed graph (that is, every vertex has exactly $d$ in-neighbors and $d$ out-neighbors), then a shunt decomposition consists of exactly $d$ shunts.

Shunt decompositions provide a natural framework for constructing discrete-time quantum walks, since they decompose the adjacency matrix into a sum of permutation matrices that define the shift operator. Specifically, if $A$ is the adjacency matrix of a $d$-regular directed graph $G$, then there exist permutation matrices $P_1, \ldots, P_d \in \mathbb{C}^{n \times n}$ satisfying
\begin{equation}
A = \sum_{j=1}^{d} P_j,
\end{equation}
where $n = |V(G)|$. The existence of such a decomposition follows from the below lemma.
\begin{lemma}~\cite[Lemma~7.1.1]{Godsil_Zhan_2023}\label{lemma 2.1}
Let \(G\) be a \(d\)-regular directed graph. Then \(G\) admits a shunt decompositions.  
\end{lemma}
In a discrete-time quantum walk on a $d$-regular directed graph $G$ with
$n=|V(G)|$ vertices, the state space is the Hilbert space
\(
\mathcal{H}=\mathbb{C}^{n}\otimes\mathbb{C}^{d},
\)
where the first factor represents the vertex space and the second
factor represents the $d$-dimensional coin space.
Using the natural identification
\(
\mathbb{C}^{d}\otimes\mathbb{C}^{n}
\cong
\mathbb{C}^{n}\otimes\mathbb{C}^{d},
\)
the evolution operator can be written as
\begin{equation}\label{transition_shift}
U=SC,
\end{equation}
where the shift and coin operators defined as
\begin{equation}\label{shift_matrix}
S=
\sum_{j=1}^{d}P_j\otimes E_{jj}, \quad
C=
\sum_{v\in V(G)}E_{vv}\otimes C_v.
\end{equation}
Here, $P_j$ is the permutation matrix associated with the $j$-th
shunt in the shunt decomposition of $G$. The matrices $E_{jj}$ and
$E_{vv}$ denote the corresponding matrix units. In particular,
$E_{jj}$ is the $d\times d$ matrix with a $1$ in the $(j,j)$-entry and
zeros elsewhere, while $E_{vv}$ is the $n\times n$ matrix with a $1$ in
the $(v,v)$-entry and zeros elsewhere. The matrix $C_v$ is the coin
operator acting on the $d$-dimensional coin space associated with
vertex $v$. Thus, the $j$-th coin state is shifted according to the permutation
$P_j$, resulting in a deterministic shift along the corresponding
shunt, as discussed in~\cite[\textit{Sec.}~2.2]{GodsilZhan2019}.
In the case of a uniform coin, we take
\(
C_v=C_0
\quad\text{for all }v\in V(G).
\)
Consequently,
\begin{equation}
C
=
\sum_{v\in V(G)}E_{vv}\otimes C_0
=
I_n\otimes C_0,
\end{equation}
and hence
\begin{equation}
U
=
\left(
\sum_{j=1}^{d}P_j\otimes E_{jj}
\right)
(I_n\otimes C_0).
\end{equation}
The state of the walk evolves in discrete time according to
\(
\psi_k=U^k x,
\)
where $x\in\mathcal{H}$ is the initial state. Following~\cite{ChanZhan2023}, we say that there is PST from a state $x\in\mathcal{H}$ to a state
$y\in\mathcal{H}$ if there exists a positive integer $k$ such that
\begin{equation}
U^k x=\gamma y,
\end{equation}
for some $\gamma\in\mathbb{C}$ with $|\gamma|=1$. Equivalently, PST
from $x$ to $y$ occurs at time $k$ if
\begin{equation}
\left|\langle U^k x,y\rangle\right|=1.
\end{equation}
We say that the  graph  $G$
is said to be \emph{periodic} if there exists a positive integer $\tau$
such that
\(
U^\tau=I.
\)
The smallest such positive integer $\tau$ is called the
\emph{period of the graph} $G$. For a state $x\in\mathcal{H}$, we say that $G$ is \emph{periodic at $x$} if
there exist a positive integer $\tau$ and a scaler
$\gamma\in\mathbb{C}$ with $|\gamma|=1$ such that
\begin{equation}
U^\tau x=\gamma x.
\end{equation} The smallest such positive integer $\tau$ is called the
\emph{period of $G$ at $x$}~\cite{Godsil_Zhan_2023}.

\section{Directed Partial Join Graphs with Signed Coupling}\label{sec3}

\begin{defn}
Let $G_1$ and $G_2$ be two vertex-disjoint $d$-regular directed graphs. The \emph{join} $G_1 \vee G_2$ is the graph with vertex set $V(G_1) \cup V(G_2)$ and edge set
\[
E(G_1) \cup E(G_2) \cup \{(u,v) : u \in V(G_1),\ v \in V(G_2)\}.
\]
In this configuration, every vertex in $G_1$ is adjacent to every vertex in $G_2$.
\end{defn}

Let $G = G_1 \vee G_2$ denote the standard join of graphs $G_1$ and $G_2$ with $|V(G_1)| = n_1$, $|V(G_2)| = n_2$. If
\(
A_{G_1} \in \mathbb{R}^{n_1 \times n_1},
\quad
A_{G_2} \in \mathbb{R}^{n_2 \times n_2}
\)
are the adjacency matrices of $G_1$ and $G_2$, respectively, then the adjacency matrix of the join graph $G$ is
\begin{equation}
A_G =
\begin{pmatrix}
A_{G_1} & J_{n_1 \times n_2} \\[4pt]
J_{n_2 \times n_1} & A_{G_2}
\end{pmatrix},
\end{equation}
where $J_{n_1 \times n_2}$ denotes the all-ones matrix of order $n_1 \times n_2$, and $J_{n_2 \times n_1} = J_{n_1\times n_2}^{\top}$ is the all-ones matrix of order $n_2 \times n_1$ \cite{GodsilRoyle2001}. When $n_1 = n_2 = n$, we write simply $J_n$ for the $n \times n$ all-ones matrix.

For a self-join graph, consider as an example the \emph{looped} complete graph $K_{2n}^{\circlearrowleft}$, that is, $K_{2n}$ together with a loop added at every vertex. Then
\begin{equation}
K_{2n}^{\circlearrowleft} \;\cong\; K_n^{\circlearrowleft} \vee K_n^{\circlearrowleft},
\end{equation}
where $K_n^{\circlearrowleft}$ denotes the complete graph on $n$ vertices with a loop added at every vertex (see Fig.~\ref{Fig 2}). In other words, the self-join of $K_n^{\circlearrowleft}$ with itself is isomorphic to the complete graph on $2n$ vertices with loops at every vertex. The adjacency matrix is given by
\begin{equation}\label{eq12}
A\big(K_{2n}^{\circlearrowleft}\big) =
\begin{pmatrix}
J_n & J_n \\
J_n & J_n
\end{pmatrix}.
\end{equation}
Here each block has size $n \times n$, and \emph{all four blocks are equal}, each being the all-ones matrix $J_n$. In this setting, the adjacency matrix of the join reduces to the all-ones matrix $J_{2n}$.

Next, consider the loopless complete bipartite graph $K_{n,n}$, where
\begin{equation}\label{loopless complete bipartite graph $K_{n,n}$}
A\big(K_{n,n}\big) =
\begin{pmatrix}
0 & J_n \\
J_n & 0
\end{pmatrix}.
\end{equation}
Here the two diagonal blocks of the adjacency matrix are zero, while the two off-diagonal blocks are both equal to the all-ones matrix $J_n$. Equivalently, $K_{n,n}$ is the full join of two independent sets of $n$ vertices each.

\medskip

Now we adopt the definition of partial joins of graphs \cite{stiebitz1993colouring}. Instead of joining every vertex of $G_1$ to every vertex of $G_2$, we restrict the interconnections to a prescribed edge set chosen so that each vertex is connected to a fixed number of vertices in the other graph.

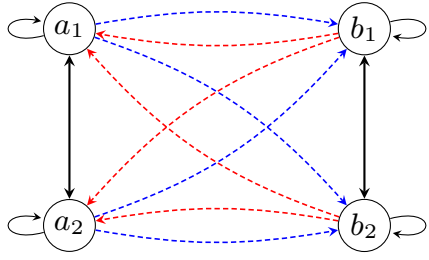
\begin{figure}[t]
\centering
\begin{tikzpicture}[
scale=1.3, transform shape,
every node/.style={circle, draw, fill=white, inner sep=1.5pt, minimum size=15pt, font=\small},
AtoB/.style={blue, semithick, ->, >=stealth, dash pattern=on 2pt off 1.2pt},
BtoA/.style={red, semithick, ->, >=stealth, dash pattern=on 2pt off 1.2pt},
Gedge/.style={thick, <->, >=stealth}
]

\node (a1) at (0, 1) {$a_1$};
\node (a2) at (0, -1) {$a_2$};

\draw[Gedge] (a1) -- (a2);
\draw[>=stealth, <->] (a1) edge[loop left] ();
\draw[>=stealth, <->] (a2) edge[loop left] ();

\node (b1) at (3, 1) {$b_1$};
\node (b2) at (3, -1) {$b_2$};

\draw[Gedge] (b1) -- (b2);
\draw[>=stealth, <->] (b1) edge[loop right] ();
\draw[>=stealth, <->] (b2) edge[loop right] ();


\draw[AtoB] (a1) to[bend left=10] (b1);
\draw[BtoA] (b1) to[bend left=10] (a1);

\draw[AtoB] (a2) to[bend right=10] (b2);
\draw[BtoA] (b2) to[bend right=10] (a2);

\draw[AtoB] (a1) to[bend left=15] (b2);
\draw[BtoA] (b2) to[bend left=15] (a1);

\draw[AtoB] (a2) to[bend right=15] (b1);
\draw[BtoA] (b1) to[bend right=15] (a2);

\end{tikzpicture}
\caption{The full join $K_2^\circlearrowleft \vee K_2^\circlearrowleft$ , shown with bidirectional internal edges and loops, representing a complete graph $K_4^\circlearrowleft$ with loops.}
\label{Fig 2}
\end{figure}
\begin{defn}\label{def:dpj-reg}
Let $G_1=(V_1,E_1)$ and $G_2=(V_2,E_2)$ be vertex-disjoint
$d$-regular directed graphs, each on $n$ vertices. Let
$d_1,d_2\geq 0$ with $d_1,d_2\leq n$. The
\emph{directed partial join}
\(
G=G_1\overset{(d_1,d_2)}{\vec{\vee}}G_2
\)
is the directed graph with vertex set
\(
V(G)=V_1\cup V_2
\)
and edge set
\[
E(G)=E_1\cup E_2\cup E(\mathcal{J}_1)\cup E(\mathcal{J}_2),
\]
where $\mathcal{J}_1 $ is a directed bipartite coupling  graph from \(G_1\) to \(G_2\) on $V_1\cup V_2$ that adds \emph{exactly} $d_1$ arcs from each $a_i\in V_1$ to distinct vertices in $V_2$, and $\mathcal{J}_2 $ from \( G_2\) to \(G_1\) that adds \emph{exactly} $d_2$ arcs from each $b_i \in V_2$ to distinct vertices in $V_1$. 
With the vertices ordered as
\(
V_1=\{a_1,\ldots,a_n\},
\quad
V_2=\{b_1,\ldots,b_n\},
\)
the adjacency matrix of $G$ has the block form
\begin{equation}
A(G)
=
\begin{pmatrix}
A(G_1) & A(\mathcal{J}_1)\\
A(\mathcal{J}_2) & A(G_2)
\end{pmatrix},
\end{equation}
where
\(
A(\mathcal{J}_1)_{ij}=1
\iff
(a_i,b_j)\in E(\mathcal{J}_1),
\)
and
\(
A(\mathcal{J}_2)_{ij}=1
\iff
(b_i,a_j)\in E(\mathcal{J}_2).
\)
Equivalently,
\(
A(G)=A_{\mathcal G}+A_{\mathcal J},
\)
where
\begin{equation}
A_{\mathcal G}
=
\begin{pmatrix}
A(G_1)&0\\
0&A(G_2)
\end{pmatrix},
\qquad
A_{\mathcal J}
=
\begin{pmatrix}
0&A(\mathcal{J}_1)\\
A(\mathcal{J}_2)&0
\end{pmatrix}.
\end{equation}
In particular, if  $d_1=d_2=d$, we call $G$ a \emph{$2d$-regular directed partial join graph}, and $\mathcal{J}$ the \emph{$d$-regular directed coupling graph} of $G$.
\end{defn}
A {signed graph} is a graph in which every edge is assigned one
of two signs
\(
+, -.
\)
Formally, a signed graph is usually written as
\(
\Sigma=(G,\sigma),
\)
where $G=(V,E)$ is the underlying graph and
\(
\sigma:E\to\{+1,-1\}
\)
is the {sign function} that assigns a sign to each edge.
If
\(
\sigma(e)=+1
\quad\text{for every }e\in E,
\)
then $\Sigma$ is called an {all-positive signed graph}. If
\(
\sigma(e)=-1
\quad\text{for every }e\in E,
\)
then $\Sigma$ is called an {all-negative signed graph}.
A signed graph is called {homogeneous} if it is either
all-positive or all-negative. Otherwise, it is called
{heterogeneous}, meaning that it contains both positive and
negative edges~\cite{acharya2016lict}.

\begin{defn}\cite[Definition 1]{zhang2022polarity}\label{def:signed-bipartite}
A bipartite graph has two separate vertex sets, each of which only has
connections with the vertices in the other vertex set. A signed
bipartite graph can be denoted as
\(
G=(V_1,V_2,E^+,E^-),
\)
where
\(
V_1=\{a_1,a_2,\ldots,a_n\},
\quad
V_2=\{b_1,b_2,\ldots,b_n\}
\)
are the mutually exclusive vertex sets. The sets
\(
E^+\subseteq V_1\times V_2
\quad\text{and}\quad
E^-\subseteq V_1\times V_2
\)
are the positive and negative edges that connect vertices between the
two vertex sets, where
\(
E^+\cap E^-=\varnothing.
\)
\end{defn}
Based on Definition~\ref{def:signed-bipartite}, we define a
\emph{signed directed coupling graph} of \(G\).
\begin{defn}\label{def:signed-coupling}
Let
\(
\mathcal{J}=\mathcal{J}_1\cup\mathcal{J}_2
\)
be the directed bipartite coupling graph of a directed partial join
$G=G_1\overset{(d_1,d_2)}{\vec{\vee}}G_2$, where
$\mathcal{J}_1$ consists of arcs from $G_1$ to $G_2$ and
$\mathcal{J}_2$ consists of arcs from $G_2$ to $G_1$.
Assign the sign $+$ to every arc in $\mathcal{J}_1$ and the sign $-$
to every arc in $\mathcal{J}_2$. We denote the resulting signed
directed coupling graph by $\mathcal{J}_{\pm}$.
Its signed adjacency matrix is
\begin{equation}\label{signed_adjacency}
A_{\mathcal{J_{\pm}}}=
\begin{pmatrix}
0 & A(\mathcal{J}_1)\\
-A(\mathcal{J}_2) & 0
\end{pmatrix}.
\end{equation}
Thus, $\mathcal{J}_{\pm}$ is a heterogeneous signed directed coupling
graph, since it contains both positive and negative arcs. If \(d_1=d_2=d\), then \(\mathcal{J}\) is called the \(d\)-regular signed directed coupling graph of \(G\).
\end{defn}

The corresponding signed adjacency matrices of the $2d$-regular directed
partial join graphs with signed coupling, given by
\eqref{eq12} and
\eqref{loopless complete bipartite graph $K_{n,n}$}, are
\begin{equation}\label{signed adjacency matrices}
A_{\pm}\big(K_{2n}^{\circlearrowleft}\big)
=
\begin{pmatrix}
J_n & J_n\\
-J_n & J_n
\end{pmatrix},
\qquad
A_{\pm}\big(K_{n,n}\big)
=
\begin{pmatrix}
0 & J_n\\
-J_n & 0
\end{pmatrix}.
\end{equation}
Based on this framework, we distinguish the following coupling cases.
\begin{enumerate}\label{reciprocal_coupling}
    \item \textit{Identical graphs with reciprocal coupling:}
    Suppose that $G_1$ and $G_2$ are identical under the correspondence $a_i\leftrightarrow b_i$, so that
    \(
    A(G_1)=A(G_2),
    \)
    and suppose that the coupling is reciprocal, that is,\ $a_i\to b_j\in E(\mathcal{J}_1) \iff b_j\to a_i\in E(\mathcal{J}_2)$. Then
    \(
    A(\mathcal{J}_2)=A(\mathcal{J}_1)^{\top}.
    \)
    In particular, reciprocity forces $d_1=d_2$, since transposition preserves the number of arcs, so $nd_1=|E(\mathcal{J}_1)|=|E(\mathcal{J}_2)|=nd_2$. If, moreover, $A(\mathcal{J}_1)$ is symmetric, then $A(\mathcal{J}_2)=A(\mathcal{J}_1)^\top=A(\mathcal{J}_1)$, so
    \(
    A(\mathcal{J}_1)=A(\mathcal{J}_2).
    \)
    In particular, if the coupling matrix coincides with the adjacency matrix of the common graph and this matrix is symmetric, then
    \begin{equation}\label{block_equal}
    A(\mathcal{J}_1)=A(\mathcal{J}_2)=A(G_1)=A(G_2).
    \end{equation}
    This is the case, for example, for the reciprocal partial join $C_4\overset{(2,2)}{\vec{\vee}}C_4$ (See Fig.~\ref{fig:C4-partial-join-case-a} ). Taking $d=n$ and $A(\mathcal{J}_1)=A(G_1)=J_n$ recovers the looped complete graph construction $K_{2n}^{\circlearrowleft}\cong K_n^{\circlearrowleft}\vee K_n^{\circlearrowleft}$.

    More generally, the common coupling matrix may be symmetric but different from the adjacency matrix of  \(G_1\) or \(G_2\), giving
    \begin{equation}\label{block_unequal}
    A(\mathcal{J}_1)=A(\mathcal{J}_2)\;\neq\;A(G_1)=A(G_2),
    \end{equation}
    corresponding to a reciprocal coupling that reuses a non-trivial interconnection pattern distinct from the graph \(G_1\) or \(G_2\)    (see Fig.~\ref{fig:C4-partial-join-case-b}). If $A(\mathcal{J}_1)$ is not symmetric, then reciprocity instead gives
    \(
    A(\mathcal{J}_2)=A(\mathcal{J}_1)^{\top}\;\neq\;A(\mathcal{J}_1),
    \)
    even though $A(G_1)=A(G_2)$; this is the generic asymmetric case.

    \item \textit{Non-reciprocal coupling:}
    If reciprocity is not assumed, an arc $a_i\to b_j\in E(\mathcal{J}_1)$ does not necessarily imply the corresponding arc $b_j\to a_i\in E(\mathcal{J}_2)$. Hence, in general,
    \(
    A(\mathcal{J}_2)\neq A(\mathcal{J}_1)^{\top},
    \)
    and there is no necessary equality or inequality between $A(\mathcal{J}_1)$ and $A(\mathcal{J}_2)$, regardless of whether $G_1\cong G_2$ or \(G_1\) and \(G_2\) are identical.
\end{enumerate}
 
 \begin{lemma}\label{lem:2d-regular}
Let $G=G_1\overset{(d_1,d_2)}{\vec\vee}G_2$ be the directed partial join of $G_1$ and $G_2$ with bipartite coupling graphs $\mathcal{J}_1$  with degree \(d_1\) and $\mathcal{J}_2$  with degree \(d_2\), where $|V(G_1)|=|V(G_2)|=n$. If $G_1,G_2$ are each $d$-regular, then
\(
G \text{ is } 2d\text{-regular graph}
\quad\Longleftrightarrow\quad
\mathcal{J}  \text{ is }  d\text{-regular directed graph.}
\)
\end{lemma}

\begin{proof}For $v\in V(G_1)$,
\begin{equation}\label{out_degree_1}
\mathrm{outdeg}_G(v)=\sum_{u\in V(G_1)} A(G_1)_{vu} \;+\; \sum_{b\in V(G_2)} A(\mathcal{J}_1)_{vb},
\qquad
\mathrm{indeg}_G(v)=\sum_{u\in V(G_1)} A(G_1)_{uv} \;+\; \sum_{b\in V(G_2)} A(\mathcal{J}_2)_{bv}.
\end{equation}
Similarly, for $w\in V(G_2)$,
\begin{equation}\label{out_degree_2}
\mathrm{outdeg}_G(w)=\sum_{a\in V(G_1)} A(\mathcal{J}_2)_{wa} \;+\; \sum_{u\in V(G_2)} A(G_2)_{wu},
\qquad
\mathrm{indeg}_G(w)=\sum_{a\in V(G_1)} A(\mathcal{J}_1)_{aw} \;+\; \sum_{u\in V(G_2)} A(G_2)_{uw}.
\end{equation}
Since $G_1$ and $G_2$ are $d$-regular,
\begin{equation}\label{G_1_degree}
\sum_{u\in V(G_1)} A(G_1)_{vu} = \sum_{u\in V(G_1)} A(G_1)_{uv} = d
\qquad \forall\, v\in V(G_1),
\end{equation}
\begin{equation}\label{G_2_degree}
\sum_{u\in V(G_2)} A(G_2)_{wu} = \sum_{u\in V(G_2)} A(G_2)_{uw} = d
\qquad \forall\, w\in V(G_2).
\end{equation}

\noindent\textbf{($\Leftarrow$)} If $\mathcal{J}$ is $d$-regular, then for every $v\in V(G_1)$ and $w\in V(G_2)$,
\begin{equation}\label{J_1_degree}
\sum_{b\in V(G_2)} A(\mathcal{J}_1)_{vb} = \sum_{b\in V(G_2)} A(\mathcal{J}_2)_{bv} = \sum_{a\in V(G_1)} A(\mathcal{J}_2)_{wa} = \sum_{a\in V(G_1)} A(\mathcal{J}_1)_{aw} = d.
\end{equation}
Substituting~\eqref{G_1_degree}, \eqref{G_2_degree}, and \eqref{J_1_degree} in ~\eqref{out_degree_1} and~\eqref{out_degree_2}, we have
\begin{equation}
\mathrm{outdeg}_G(v)=\mathrm{indeg}_G(v)=\mathrm{outdeg}_G(w)=\mathrm{indeg}_G(w)=2d
\qquad \forall\, v\in V(G_1),\ w\in V(G_2),
\end{equation}
so every vertex of $G$ has in-degree and out-degree $2d$, that is, \ $G$ is $2d$-regular directed partial join graph.

\noindent\textbf{($\Rightarrow$)} Suppose $G$ is $2d$-regular. Since $G_1$ and $G_2$ are $d$-regular, equations~\eqref{out_degree_1}, \eqref{out_degree_2}, \eqref{G_1_degree}, and \eqref{G_2_degree} give
\begin{equation}\label{in_degree_1}
\mathrm{outdeg}_G(v)=2d \;\Rightarrow\; \sum_{b\in V(G_2)} A(\mathcal{J}_1)_{vb}=d,
\qquad
\mathrm{indeg}_G(w)=2d \;\Rightarrow\; \sum_{a\in V(G_1)} A(\mathcal{J}_1)_{aw}=d,
\end{equation}
\begin{equation}\label{in_degree_2}
\mathrm{outdeg}_G(w)=2d \;\Rightarrow\; \sum_{a\in V(G_1)} A(\mathcal{J}_2)_{wa}=d,
\qquad
\mathrm{indeg}_G(v)=2d \;\Rightarrow\; \sum_{b\in V(G_2)} A(\mathcal{J}_2)_{bv}=d.
\end{equation}
By~\eqref{in_degree_1}, at every vertex $A(\mathcal{J}_1)$ has row sum (out-degree, over $V(G_1)$) equal to $d$ and column sum (in-degree, over $V(G_2)$) equal to $d$. Similarly, by~\eqref{in_degree_2}, $\mathcal{J}_2$ has row sum (out-degree, over $V(G_2)$) equal to $d$ and column sum (in-degree, over $V(G_1)$) equal to $d$. Hence $\mathcal{J}$ is $d$-regular directed graph.
\end{proof}

\begin{figure}[t]
\centering
\begin{tikzpicture}[
scale=1.1, transform shape,
every node/.style={circle, draw, fill=white, inner sep=1.2pt, minimum size=15pt, font=\small},
AtoB/.style={blue, semithick, ->, >=stealth, dash pattern=on 2pt off 1.2pt},
BtoA/.style={red,  semithick, ->, >=stealth, dash pattern=on 2pt off 1.2pt},
Gedge/.style={thick, <->, >=stealth}
]
\node (a1) at (0, 1.5)  {$a_1$};
\node (a2) at (1, 0)    {$a_2$};
\node (a3) at (0, -1.5) {$a_3$};
\node (a4) at (-1, 0)   {$a_4$};
\draw[Gedge] (a1)--(a2);
\draw[Gedge] (a2)--(a3);
\draw[Gedge] (a3)--(a4);
\draw[Gedge] (a4)--(a1);
\node (b1) at (4, 1.5)  {$b_1$};
\node (b2) at (5, 0)    {$b_2$};
\node (b3) at (4, -1.5) {$b_3$};
\node (b4) at (3, 0)    {$b_4$};
\draw[Gedge] (b1)--(b2);
\draw[Gedge] (b2)--(b3);
\draw[Gedge] (b3)--(b4);
\draw[Gedge] (b4)--(b1);
\draw[AtoB] (a1) -- (b2); \draw[BtoA] (b2) -- (a1);
\draw[AtoB] (a1) -- (b4); \draw[BtoA] (b4) -- (a1);
\draw[AtoB] (a2) -- (b1); \draw[BtoA] (b1) -- (a2);
\draw[AtoB] (a2) -- (b3); \draw[BtoA] (b3) -- (a2);
\draw[AtoB] (a3) -- (b2); \draw[BtoA] (b2) -- (a3);
\draw[AtoB] (a3) -- (b4); \draw[BtoA] (b4) -- (a3);
\draw[AtoB] (a4) -- (b1); \draw[BtoA] (b1) -- (a4);
\draw[AtoB] (a4) -- (b3); \draw[BtoA] (b3) -- (a4);
\begin{scope}[shift={(6.2,0.5)}]
    \draw[Gedge] (0,0.8) -- (0.6,0.8) node[draw=none, fill=none, right, font=\scriptsize] {edges of $G_i$};
    \draw[blue, semithick, ->, >=stealth] (0,0.4) -- (0.6,0.4) node[draw=none, fill=none, right, font=\scriptsize] {arc of $\mathcal{J}_1$: $a_i \to b_j$};
    \draw[red, semithick, ->, >=stealth] (0,0) -- (0.6,0) node[draw=none, fill=none, right, font=\scriptsize] {arc of $\mathcal{J}_2$: $b_j \to a_i$};
\end{scope}
\end{tikzpicture}
\caption{Here $A(\mathcal{J}_1)=A(\mathcal{J}_2)=A(G_1)=A(G_2)$: the directed partial join $C_4 \overset{(2,2)}{\vec{\vee}} C_4$, with each vertex of $G_1$ connecting to exactly two vertices of $G_2$ via $\mathcal{J}_1$, and reciprocally via $\mathcal{J}_2$, forming a $2$-regular bipartite subgraph of $K_{4,4}$ whose adjacency matrix coincides with that of the internal $4$-cycles.}
\label{fig:C4-partial-join-case-a}

\end{figure}
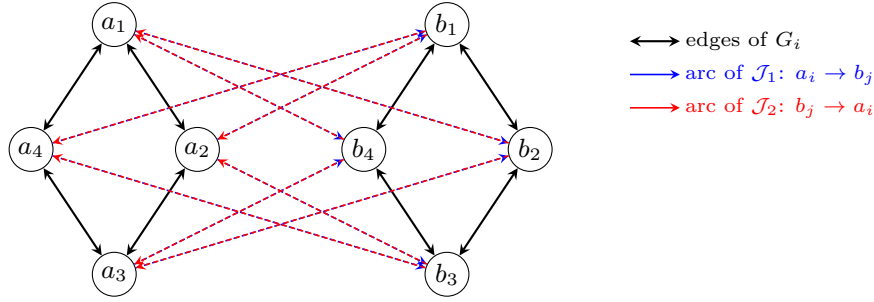

\begin{ex}\label{ex:C4-C4-case-a}
Figure~\ref{fig:C4-partial-join-case-a} illustrates the directed partial join $C_4\overset{(2,2)}{\vec{\vee}} C_4$, where $G_1\cong G_2\cong C_4$. Labelling vertices so that $a_i\leftrightarrow b_i$, one checks directly that
\(
A(\mathcal{J}_1) = A(G_1),
\)
which is symmetric, so $A(\mathcal{J}_2)=A(\mathcal{J}_1)^{\top}=A(\mathcal{J}_1)$. Combined with $A(G_1)=A(G_2)$, this gives
\(
A(\mathcal{J}_1) = A(\mathcal{J}_2) = A(G_1) = A(G_2).
\)
This is the reciprocal, identical-graph case in which the coupling matrices coincide with the adjacency matrix of $C_4$. Here $d_1=d_2=2$, so $G_1,G_2$, and $\mathcal{J}$ are all $2$-regular, and Lemma~\ref{lem:2d-regular} gives that $G$ is $4$-regular.
\end{ex}
\begin{figure}[t]
\centering
\begin{tikzpicture}[
scale=1.3, transform shape,
every node/.style={circle, draw, fill=white, inner sep=1.2pt, minimum size=15pt, font=\small},
AtoB/.style={blue, semithick, ->, >=stealth, dash pattern=on 2pt off 1.2pt},
BtoA/.style={red,  semithick, ->, >=stealth, dash pattern=on 2pt off 1.2pt},
Gedge/.style={thick, <->, >=stealth}
]
\node (a1) at (0, 1.5)  {$a_1$};
\node (a2) at (1, 0)    {$a_2$};
\node (a3) at (0, -1.5) {$a_3$};
\node (a4) at (-1, 0)   {$a_4$};
\draw[Gedge] (a1)--(a2);
\draw[Gedge] (a2)--(a3);
\draw[Gedge] (a3)--(a4);
\draw[Gedge] (a4)--(a1);
\node (b1) at (4, 1.5)  {$b_1$};
\node (b2) at (5, 0)    {$b_2$};
\node (b3) at (4, -1.5) {$b_3$};
\node (b4) at (3, 0)    {$b_4$};
\draw[Gedge] (b1)--(b2);
\draw[Gedge] (b2)--(b3);
\draw[Gedge] (b3)--(b4);
\draw[Gedge] (b4)--(b1);
\draw[AtoB] (a1) to[bend left=15]  (b1);
\draw[BtoA] (b1) to[bend left=15] (a1);
\draw[AtoB] (a1) to[bend left=15]  (b3);
\draw[BtoA] (b3) to[bend left=15] (a1);
\draw[AtoB] (a2) to[bend left=15]  (b2);
\draw[BtoA] (b2) to[bend left=15] (a2);
\draw[AtoB] (a2) to[bend left=15]  (b4);
\draw[BtoA] (b4) to[bend left=15] (a2);
\draw[AtoB] (a3) to[bend left=15]  (b1);
\draw[BtoA] (b1) to[bend left=15] (a3);
\draw[AtoB] (a3) to[bend left=15]  (b3);
\draw[BtoA] (b3) to[bend left=15] (a3);
\draw[AtoB] (a4) to[bend left=15]  (b2);
\draw[BtoA] (b2) to[bend left=15] (a4);
\draw[AtoB] (a4) to[bend left=15]  (b4);
\draw[BtoA] (b4) to[bend left=15] (a4);
\begin{scope}[shift={(6.2,0.5)}]
    \draw[Gedge] (0,0.8) -- (0.6,0.8) node[draw=none, fill=none, right, font=\scriptsize] {edges of $G_i$};
    \draw[blue, semithick, ->, >=stealth] (0,0.4) -- (0.6,0.4) node[draw=none, fill=none, right, font=\scriptsize] {$a_i \to b_j$};
    \draw[red, semithick, ->, >=stealth] (0,0) -- (0.6,0) node[draw=none, fill=none, right, font=\scriptsize] {$b_j \to a_i$};
\end{scope}
\end{tikzpicture}
\caption{Here $A(G_1)=A(G_2)\neq A(\mathcal{J}_1)=A(\mathcal{J}_2)$: the directed partial join $C_4\overset{(2,2)}{\vec\vee}C_4$
with coupling $a_1,a_3\to\{b_1,b_3\}$ and
$a_2,a_4\to\{b_2,b_4\}$ (and symmetrically $b_j\to a_i$). Each vertex
of $G_1$ sends and receives exactly two arcs, so the
 graph $\mathcal{J}=\mathcal{J}_1\cup\mathcal{J}_2$ is $2$-regular bipartite subgraphs of
$K_{4,4}$, matching the $(2,2)$ degree label.}
\label{fig:C4-partial-join-case-b}
\end{figure}
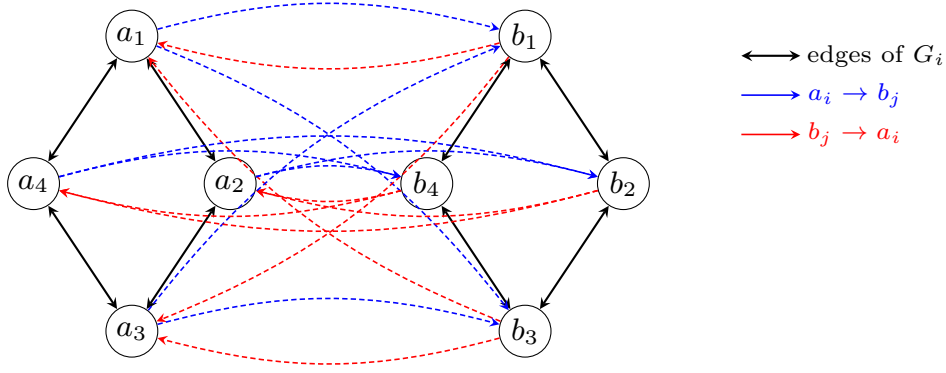
\section{Circulant Partial Joins and Tensor Products of $K_{n,n}$ and $K_2$ Graphs}
\label{sec4}
A \emph{circulant graph} $C_n(S)$ is an undirected graph on the
vertex set
\(
\mathbb{Z}_n=\{0,1,\ldots,n-1\},
\)
where each vertex $v$ is adjacent to the vertices
\(
v+s\pmod n,\quad s\in S,
\)
for a symmetric connection set
\(
S\subseteq\mathbb{Z}_n\setminus\{0\},
\quad
s\in S\Rightarrow -s\in S.
\)
Thus, the graph is determined by the set of allowed connections.
In~\cite{angeles2009perfect}, the circulant join was constructed by taking the off-diagonal block to be an arbitrary Boolean matrix $C$; that is, two copies of $G$ are connected by joining corresponding vertices across the two copies according to $C$. This construction, however, is not directly feasible for our model of the discrete-time quantum walk. An arbitrary choice of $C$ generally causes a {degree mismatch}. Since $C$ need not be $d$-regular, the local operators $U_{\mathcal{J}_1}, U_{\mathcal{J}_2}$ built from the coupling graphs $\mathcal{J}_1, \mathcal{J}_2$ act on spaces of inconsistent dimension at each vertex, incompatible with the uniform $d$-dimensional local operators $U_{G_1}, U_{G_2}$ coming from the (regular) internal graphs $G_1, G_2$. These pieces therefore cannot be assembled into a single well-defined operator on $G$, and unitarity of the transition operator $U_G$ fails, as illustrated in Example~\ref{eq.2.2} and Proposition~\ref{prop4}. We therefore construct the coupling to keep $U_G$ unitary while also preserving both the coupling degree and the regularity of the graph, as discussed for $U_G$ in Section~\ref{sec5}.

In this section we accordingly examine $d$-regular directed partial join and tensor product graphs, that is, graphs whose block-form adjacency matrix falls into one of three special cases:

\begin{enumerate}
    \item All four blocks coincide:
    \(
    A(G_1)=A(G_2)=A(\mathcal{J}_1)=A(\mathcal{J}_2),
    \)
    as given in~\eqref{block_equal};

    \item Diagonal blocks vanish, off-diagonal blocks agree:
    \(
    A(G_1)=A(G_2)=0
    \quad\text{and}\quad
    A(\mathcal{J}_1)=A(\mathcal{J}_2),
    \)
    exemplified by the loopless complete bipartite graph $K_{n,n}$, given in~\eqref{loopless complete bipartite graph $K_{n,n}$};

    \item {Diagonal and off-diagonal blocks pair up separately, but differ from each other:}
    \(
    A(G_1)=A(G_2)\\
    \text{and}\quad
    A(\mathcal{J}_1)=A(\mathcal{J}_2),
    \quad\text{with } A(G_1)\neq A(\mathcal{J}_1),
    \)
    as given in~\eqref{block_unequal}.
\end{enumerate}

We use the following equivalent distance-based notation for circulant
graphs

\begin{defn}\label{def:circulant-graph}
Let $n\geq 3$ and let
\(
S\subseteq\{1,\ldots,\lfloor n/2\rfloor\}
\)
be a set of \emph{connection distances}. The \emph{circulant graph}
$C_n(S)$ has vertex set
\(
\mathbb{Z}_n=\{0,1,\ldots,n-1\},
\)
where two vertices $i,j\in\mathbb{Z}_n$ are adjacent if and only if
\(
d(i,j)\in S,
\)
where
\(
d(i,j)=\min\{|i-j|,\,n-|i-j|\}
\)
is the \emph{circular distance} between $i$ and $j$. Thus,
$C_n(S)$ is undirected, and its adjacency matrix
\(
B=A(C_n(S))
\)
is symmetric.
\end{defn}
\begin{defn}\label{def:antipodal}
Let $C_n$ be a cycle with vertices labelled by $\mathbb{Z}_n$. If $n$ is
even, the \emph{antipodal vertex} of a vertex $i\in\mathbb{Z}_n$ is the
unique vertex
\(
i+\frac{n}{2}\pmod n.
\)
It is the vertex at maximum distance from $i$ on the cycle, with
circular distance
\(
d\left(i,i+\frac{n}{2}\right)=\frac{n}{2}.
\)
Since
\(
i+\frac{n}{2}\equiv i-\frac{n}{2}\pmod n,
\)
the vertices $i+\frac{n}{2}$ and $i-\frac{n}{2}$ are the same vertex
in $\mathbb{Z}_n$.
\end{defn}

Every vertex has two neighbors at each distance $s\in S$, except when $n$ is even and $s=n/2$. In this case, the two vertices
\(
i+\frac{n}{2} \quad\text{and}\quad i-\frac{n}{2} \pmod n
\)
coincide. This unique vertex is called the \emph{antipodal vertex} of $i$. Consequently, $C_n(S)$ is regular with degree
\begin{equation}
\deg(C_n(S)) = 2\left|S\setminus\left\{\frac{n}{2}\right\}\right| + \begin{cases} 1,&\frac{n}{2}\in S,\\ 0,&\frac{n}{2}\notin S. \end{cases}
\end{equation}
\begin{remark}
If $S=\{1,\dots,n/2\}$ for even $n$, then every pair of distinct vertices of $\mathbb{Z}_n$ has circular distance belonging to $S$. Hence
\(
C_n(S)=K_n.
\)
In this case, the distance $n/2$ corresponds to the unique antipodal vertex and contributes only one neighbor, rather than two. For example, when $n=4$, $S=\{1,2\}$, we have
\(
C_4(\{1,2\})=K_4,
\)
which is $3$-regular.
\end{remark}

\begin{prop}\label{prop:general-even-n}
Let $n$ be even and let $G_1=G_2=C_n(S)$ be two labeled copies of the circulant graph $C_n(S)$ on $n$ vertices,  with
\(
V(G_1)=\{a_0,\dots,a_{n-1}\}, \quad V(G_2)=\{b_0,\dots,b_{n-1}\}.
\)
Suppose that the two copies are coupled according to
\(
a_i\sim b_j \iff d(i,j)\in S.
\)
If $B=A(C_n(S))$, then
\[
A(G) = \begin{pmatrix} B&B\\ B&B \end{pmatrix} = J_{2}\otimes B.
\]
Moreover, $G$ is regular of degree
\(
\deg(G)=2\deg(C_n(S)).
\)
\end{prop}

\begin{proof}
Order the vertices of $G$ as $a_0,\dots,a_{n-1},b_0,\dots,b_{n-1}$. With respect to this ordering, we have the adjacency matrix of $G$ as
\[
A(G)= \begin{pmatrix} A(G_1)&A(\mathcal{J}_1)\\ A(\mathcal{J}_2)&A(G_2) \end{pmatrix}.
\]
Since $G_1$ and $G_2$ are labelled copies of $C_n(S)$, $A(G_1)=A(G_2)=B$. Since $C_n(S)$ are circulant graphs defined via the circular distance $d(i,j)=\min\{|i-j|,n-|i-j|\}$, and since $d(i,j)=d(j,i)$, the coupling is reciprocal 
\(
a_i\sim b_j \iff b_j\sim a_i.
\)
Hence, by \eqref{block_equal}, we have
\(
A(\mathcal{J}_2)=A(\mathcal{J}_1)^\top=B^\top=B,
\)
where the last equality follows from the symmetry of $B$. Thus, all four blocks of $A(G)$ are equal to $B$, and so
\[
A(G) = \begin{pmatrix} B&B\\ B&B \end{pmatrix}.
\]
Since
\(
J_{2} = \begin{pmatrix} 1&1\\ 1&1 \end{pmatrix},
\)
we obtain
\(
A(G)=J_{2}\otimes B.
\)
Finally, let $r=\deg(C_n(S))$. Every vertex of $G_1$ has $r$ neighbors within $G_1$ and, by the coupling rule, exactly $r$ neighbors in $G_2$. Hence, every vertex of $G_1$ has degree $r+r=2r$ in $G$. The same argument applies to every vertex of $G_2$. Therefore,
\(
\deg(G)=2r=2\deg(C_n(S)),
\)
 hence $G$ is \(2\deg(C_n(S))\) regular.

Now assign a positive sign to every coupling arc from $G_1$ to $G_2$
and a negative sign to every coupling arc from $G_2$ to $G_1$. Since
\(
A(\mathcal{J}_1)=A(\mathcal{J}_2)=B,
\)
the corresponding signed adjacency matrix is
\[
A_{\pm}(G)=
\begin{pmatrix}
B&A(\mathcal{J}_1)\\
-A(\mathcal{J}_2)&B
\end{pmatrix}
=
\begin{pmatrix}
B&B\\
-B&B
\end{pmatrix}.
\]

\end{proof}

\begin{prop}\label{cor:identical-circulant-coupling}
Let $n$ be even, and let $S,T\subseteq\{1,\dots,n/2\}$ be connection sets such that
\(
\deg(C_n(S))=\deg(C_n(T))=r.
\)
Let $G_1=G_2=C_n(S)$ be two copies of the same circulant graph on n vertices, and let the coupling graphs be defined by
\(
a_i\to b_j\iff d(i,j)\in T, \quad b_i\to a_j\iff d(i,j)\in T,
\)
where $d(i,j)=\min\{|i-j|,n-|i-j|\}$. Assume that $A(C_n(S))\neq A(C_n(T))$. Writing $B=A(C_n(S))$, $C=A(C_n(T))$, we have
\[
A(G)= \begin{pmatrix} B&C\\ C&B \end{pmatrix}
\]
Moreover, $G$ is regular of both in-degree and out-degree $2r$.
\end{prop}

\begin{proof}
Since $G_1=G_2=C_n(S)$, their adjacency matrices are the same $A(G_1)=A(G_2)=B$.
The coupling in both directions is defined using the same connection set $T$. Hence $A(\mathcal{J}_1)=A(\mathcal{J}_2)=C$, where $C=A(C_n(T))$. By assumption, $B\neq C$, so
\(
A(G_1)=A(G_2) \neq A(\mathcal{J}_1)=A(\mathcal{J}_2).
\)
With the vertex ordering $a_0,\dots,a_{n-1},b_0,\dots,b_{n-1}$, the adjacency matrix of $G$ is therefore
\[
A(G)= \begin{pmatrix} A(G_1)&A(\mathcal{J}_1)\\ A(\mathcal{J}_2)&A(G_2) \end{pmatrix} = \begin{pmatrix} B&C\\ C&B \end{pmatrix}.
\]
Since $C_n(S)$ and $C_n(T)$ are circulant graphs defined via the circular distance $d(i,j)=\min\{|i-j|,n-|i-j|\}$, and since $d(i,j)=d(j,i)$, their adjacency matrices $A(C_n(S))$ and $A(C_n(T))$ are symmetric. Given that
\(
\deg(C_n(S))=\deg(C_n(T))=r,
\)
both $G_1,G_2$ (copies of $C_n(S)$) and the coupling graphs $\mathcal{J}_1,\mathcal{J}_2$ (built from $C_n(T)$) are $r$-regular. By Lemma~\ref{lem:2d-regular}, $G$ is regular of both in-degree and out-degree $2r$. The corresponding signed adjacency matrix is \[
A_{\pm}(G)= \begin{pmatrix} A(G_1)&A(\mathcal{J}_1)\\ -A(\mathcal{J}_2)&A(G_2) \end{pmatrix} = \begin{pmatrix} B&C\\ -C&B \end{pmatrix}.
\]
\end{proof}

Given graphs $G_1$ and $G_2$ with adjacency matrices $A(G_1)$ and $A(G_2)$, their \emph{tensor product} $G_1\times G_2$ is defined as the graph whose adjacency matrix is $A(G_1)\otimes A(G_2)$. Now we will show the specific tensor product of \(K_{n,n}\) and \(K_2\).
\begin{prop}\label{prop:offdiag-perm-sum}
Let
\(
A=A(K_{n,n})
=
\begin{pmatrix}
0&J_n\\
J_n&0
\end{pmatrix}.
\)
Then
\[
A^{\otimes n}
=
\begin{pmatrix}
0&M_n\\
M_n&0
\end{pmatrix},
\qquad
M_n=J_n\otimes A^{\otimes(n-1)}.
\]
Moreover, $M_n$ is a regular zero-one matrix of degree $n^n$ and
order $2^{n-1}n^n$. Hence, there exist permutation matrices
$P_1,\ldots,P_{n^n}$ of order $2^{n-1}n^n$ such that
\(
M_n=P_1+\cdots+P_{n^n}.
\)
\end{prop}

\begin{proof}
Since $A(K_{n,n})$ is $n$-regular, $A^{\otimes(n-1)}$ is
$n^{n-1}$-regular. As $J_n$ is $n$-regular,
\(
M_n=J_n\otimes A^{\otimes(n-1)}
\)
is $n^n$-regular. Moreover, $M_n$ is a zero-one matrix of order
\(
n(2n)^{n-1}=2^{n-1}n^n.
\)
Thus, $M_n$ is the adjacency matrix of an $n^n$-regular directed graph.
By the shunt decomposition lemma~\cite[Lemma~7.1.1]{Godsil_Zhan_2023},
its arcs can be partitioned into $n^n$ arc-disjoint $1$-regular
spanning subdigraphs. The corresponding permutation matrices
$P_1,\ldots,P_{n^n}$ therefore satisfy
\(
M_n=P_1+\cdots+P_{n^n}.
\)
Finally,
\[
A^{\otimes n}
=
\begin{pmatrix}
0&J_n\otimes A^{\otimes(n-1)}\\
J_n\otimes A^{\otimes(n-1)}&0
\end{pmatrix},
\]
which gives the stated block form. Under signed coupling, the
corresponding signed adjacency matrix is
\[
A_{\pm}\big(K_{n,n}^{\otimes n}\big)
=
\begin{pmatrix}
0&M_n\\
-M_n&0
\end{pmatrix}.
\]
\end{proof}

\begin{ex}
For $n=2$, $A=A(K_{2,2})$ is $2$-regular, and hence
\(
M_2=\mathcal{J}_2\otimes A
\)
is $4$-regular and has order
\(
2^{n-1}n^n=2\cdot4=8.
\)
Thus
\[
M_2=
\begin{pmatrix}
A&A\\
A&A
\end{pmatrix},
\qquad
A=
\begin{pmatrix}
0&0&1&1\\
0&0&1&1\\
1&1&0&0\\
1&1&0&0
\end{pmatrix}.
\]
By the shunt decomposition lemma,
\(
M_2=P_1+P_2+P_3+P_4,
\)
where each $P_i$ is an $8\times8$ permutation matrix. Hence
\[
A(K_{2,2})^{\otimes2}
=
\begin{pmatrix}
0&P_1+P_2+P_3+P_4\\
P_1+P_2+P_3+P_4&0
\end{pmatrix},
\]
which is a $16\times16$ matrix, as expected.
\end{ex}

\begin{cor}\label{cor.4.8}
Let $A(K_2)=X=\begin{pmatrix}0&1\\1&0\end{pmatrix}$. Then
\(
A(K_2^{\otimes n}) \;=\; A(X^{\otimes n})\;=\; \begin{pmatrix} 0 & P \\ P & 0\end{pmatrix},
\)
where $P = X^{\otimes(n-1)}$ is a $2^{n-1}\times 2^{n-1}$ permutation matrix. Explicitly \[
P=\begin{pmatrix}
0&\cdots&0&0&1\\
0&\cdots&0&1&0\\
\vdots& &\ddots& &\vdots\\
0&1&\cdots&0&0\\
1&0&\cdots&0&0
\end{pmatrix}.
\] The corresponding signed adjacency matrix is 
\(
A_{\pm}(K_2^{\otimes n})= \begin{pmatrix} 0 & P \\ -P & 0\end{pmatrix},
\)
\end{cor}

\section{Transition Operator for Partial Join Graphs with Signed Coupling}\label{sec5}

This section constructs the transition operator for the partial join of graphs with signed coupling, which preserves PST.
Let $G_1$ and $G_2$ be two $d$-regular directed graphs on $n$
vertices, with
\(
|V(G_1)|=|V(G_2)|=n.
\)
Let $\mathcal{J}_{\pm}$ be the $d$-regular signed directed coupling
graph defined in Definition~\ref{def:signed-coupling}, where $\mathcal{J}_1$ consists of the arcs from $G_1$ to $G_2$,
each assigned the positive sign, while $\mathcal{J}_2$ consists of the
arcs from $G_2$ to $G_1$, each assigned the negative sign. Thus,
$G_1$ and $\mathcal{J}_1$ have the same source vertex set $V(G_1)$,
while $G_2$ and $\mathcal{J}_2$ have the same source vertex set
$V(G_2)$. The signed adjacency matrix of the $d$-regular directed
partial join graph
\(
G=G_1\overset{(d,d)}{\vec{\vee}}G_2
\)
is given by
\begin{equation}\label{signed_adjaceny_G}
A_{\pm}(G)=
\begin{pmatrix}
A(G_1) & A(\mathcal{J}_1)\\
-A(\mathcal{J}_2) & A(G_2)
\end{pmatrix}.
\end{equation}
As discussed in Section~\ref{sec2}, since $G_1$ and $G_2$ are $d$-regular directed graphs, each admits a d-shunt decomposition, from which we associate shift matrices
\begin{equation}\label{shift_G_1}
S_{G_1} \in \mathbb{C}^{nd \times nd},
\qquad
S_{G_2} \in \mathbb{C}^{nd \times nd}.
\end{equation}
\begin{prop}\label{prop:signed-coupling-shunt}
Let $\mathcal{J}_{\pm}$ be the signed directed coupling graph of the
$2d$-regular directed partial join
\(
G=G_1\overset{(d,d)}{\vec{\vee}}G_2,
\)
with positive signs on $\mathcal{J}_1$ and negative signs on
$\mathcal{J}_2$. Then $\mathcal{J}_{\pm}$ is a $d$-regular signed
directed bipartite graph and admits a signed $d$-shunt decomposition.
\end{prop}

\begin{proof}
By Lemma~\ref{lem:2d-regular}, the underlying directed coupling graph
$\mathcal{J}=\mathcal{J}_1\cup\mathcal{J}_2$ is $d$-regular. The sign
assignment does not change the degrees, so $\mathcal{J}_{\pm}$ is also
$d$-regular.
Hence, by Lemma~\ref{lemma 2.1} there exist
permutation matrices $Q_1,\ldots,Q_d$ such that
\(
A_{\mathcal{J}_{\pm}}=\sum_{i=1}^{d}Q_i.
\)
Since $\mathcal{J}_{\pm}$ is bipartite, each $Q_i$ has zero diagonal
blocks. Moreover, the signs of $\mathcal{J}_1$ and $\mathcal{J}_2$
give
\(
Q_i=
\begin{pmatrix}
0&\mathcal{Q}_i\\
-\mathcal{Q}'_i&0
\end{pmatrix},
\)
where $\mathcal{Q}_i$ and $\mathcal{Q}'_i$ are permutation matrices.
Therefore,
\begin{equation}\label{signed_shunt}
A_{\mathcal{J}_{\pm}}
=
\sum_{i=1}^{d}
\begin{pmatrix}
0&\mathcal{Q}_i\\
-\mathcal{Q}'_i&0
\end{pmatrix}
=
\begin{pmatrix}
0&\displaystyle\sum_{i=1}^{d}\mathcal{Q}_i\\[2mm]
-\displaystyle\sum_{i=1}^{d}\mathcal{Q}'_i&0
\end{pmatrix}.
\end{equation}
Thus by~\eqref{signed_adjacency} and~\eqref{signed_shunt}, we have
\begin{equation}
A(\mathcal{J}_1)
=
\sum_{i=1}^{d}\mathcal{Q}_i,
\qquad
A(\mathcal{J}_2)
=
\sum_{i=1}^{d}\mathcal{Q}'_i.
\end{equation}
Hence, $\mathcal{J}_{\pm}$ admits a signed $d$-shunt decomposition. The corresponding signed shift operator is
\begin{equation}\label{shift_pm}
S_{\mathcal{J}_{\pm}}
=
\begin{pmatrix}
0&S_{\mathcal{J}_1}\\
-S_{\mathcal{J}_2}&0
\end{pmatrix},
\qquad
S_{\mathcal{J}_1}
=
\sum_{i=1}^{d}\mathcal{Q}_i\otimes E_{ii},
\qquad
S_{\mathcal{J}_2}
=
\sum_{i=1}^{d}\mathcal{Q}'_i\otimes E_{ii}.
\end{equation}

\end{proof}

By~\eqref{signed_adjaceny_G}, the signed adjacency matrix of $G$ is
given by
\(
A_{\pm}(G)
=
A_{\mathcal{G}}+A_{\mathcal{J}_{\pm}},
\)
where
\begin{equation}
A_{\mathcal{G}}
=
\begin{pmatrix}
A(G_1)&0\\
0&A(G_2)
\end{pmatrix},
\qquad
A_{\mathcal{J}_{\pm}}
=
\begin{pmatrix}
0&A(\mathcal{J}_1)\\
-A(\mathcal{J}_2)&0
\end{pmatrix}.
\end{equation}
We define the corresponding shift operator by
\(
S_G=S_{\mathcal{G}}+S_{\mathcal{J}_{\pm}}.
\)
By~\eqref{shift_G_1} and~\eqref{shift_pm},
\begin{equation}
S_G
=
\begin{pmatrix}
S_{G_1}&S_{\mathcal{J}_1}\\
-S_{\mathcal{J}_2}&S_{G_2}
\end{pmatrix}.
\end{equation}
Let the coin operator be
\begin{equation}
C_G=
\begin{pmatrix}
C&0\\
0&C
\end{pmatrix}.
\end{equation}
Therefore, after normalization, the transition operator associated
with the $d$-regular directed partial join with signed coupling is
\begin{equation}
U_G
=
\frac{1}{\sqrt{2}}S_GC_G
=
\frac{1}{\sqrt{2}}
\begin{pmatrix}
S_{G_1}C&S_{\mathcal{J}_1}C\\
-S_{\mathcal{J}_2}C&S_{G_2}C
\end{pmatrix}.
\end{equation}
Writing
\(
U_{G_1}=S_{G_1}C,\quad
U_{G_2}=S_{G_2}C,
\)
and
\(
U_{\mathcal{J}_1}=S_{\mathcal{J}_1}C,\quad
U_{\mathcal{J}_2}=S_{\mathcal{J}_2}C,
\)
we obtain
\begin{equation}\label{transition_operator_join _Graphs}
{
U_G
=
\frac{1}{\sqrt{2}}
\begin{pmatrix}
U_{G_1}&U_{\mathcal{J}_1}\\
-U_{\mathcal{J}_2}&U_{G_2}
\end{pmatrix}
\in\mathbb{C}^{2nd\times 2nd}.
}
\end{equation}
Here, we call $U_{G_1}$,  $U_{G_2} \in\mathbb{C}^{nd\times nd}$ the \emph{internal transition
operators}, while
\(
U_{\mathcal{J}_1},U_{\mathcal{J}_2}\in\mathbb{C}^{nd\times nd}
\)
are the \emph{coupling transition operators} associated with the
corresponding signed shunt decompositions.

In our setting, we distinguish two types of PST.
We refer to PST between two vertices within the same graph, either
$G_1$ or $G_2$, as \emph{internal PST}. In contrast, we refer to PST
between vertices belonging to different components, from $G_1$ to
$G_2$ or from $G_2$ to $G_1$, as \emph{coupling PST}.

Let $a,b\in V(G_1)$ and $c,d\in V(G_2)$. We are interested in
determining when PST occurs either between two vertices $a$ and $b$
within the same graph $G_1$ or \(G_2\), or between vertices in the two different components of the join 
graphs, such as from $a\in V(G_1)$ to $c,d\in V(G_2)$. Let $G = G_1 \vee G_2$ denote the standard join of graphs, in the
continuous-time setting, Godsil~\cite[Sec.~5]{godsil2021sedentary}
observes that if PST from $a$ to $b$ occurs at time $t$, then the
transition amplitudes from $a$ to every vertex of $G_2$, as well as the
transition amplitude from $a$ to itself, must be zero. In other words,
\begin{equation}
\left|\left\langle U_G(t)e_a,e_y\right\rangle\right|=0
\quad\text{for all }y\in V(G_2),\quad  \text{and}\quad \left|\left\langle U_G(t)e_a,e_a\right\rangle\right|=0.
\end{equation}
Here, the value $0$ denotes a zero transition amplitude; equivalently,
the corresponding transition probability is zero.
In our discrete-time framework, the block structure of the transition
operator $U_G$ allows both internal and coupling state transfer to be
described within the same framework.  Thus, our framework does not impose any a priori restriction on the location of the receiving vertex, and both internal and coupling PST can be detected directly from the powers of $U_G$, as discussed in detail in Sections~\ref{sec6} and~\ref{sec7}.

\begin{prop}\label{prop4}
Let $U_G$ be the transition matrix of Eq.~\eqref{transition_operator_join _Graphs}
of a graph $G$ with $n\geq 2$ vertices, where each block
$U_{G_1},U_{G_2},U_{\mathcal{J}_1},U_{\mathcal{J}_2}$ is a transition matrix, hence unitary.
Then $U_G$ is unitary if and only if
\[
U_{G_1}^{\dagger}U_{\mathcal{J}_1}
=
U_{\mathcal{J}_2}^{\dagger}U_{G_2},
\qquad
U_{\mathcal{J}_1}^{\dagger}U_{G_1}
=
U_{G_2}^{\dagger}U_{\mathcal{J}_2}.
\]
\end{prop}

\begin{proof}
\noindent\textbf{($\Rightarrow$).}
Suppose $U_G$ is unitary. From
$U_G^\dagger U_G=I_{2n}$, writing $U_G$ in block form as in
Eq.~\eqref{transition_operator_join _Graphs}, we obtain
\[
U_G^{\dagger}U_G
=
\frac12
\begin{pmatrix}
U_{G_1}^{\dagger}U_{G_1}
+
U_{\mathcal{J}_2}^{\dagger}U_{\mathcal{J}_2}
&
U_{G_1}^{\dagger}U_{\mathcal{J}_1}
-
U_{\mathcal{J}_1}^{\dagger}U_{G_2}
\\[4pt]
U_{\mathcal{J}_1}^{\dagger}U_{G_1}
-
U_{G_2}^{\dagger}U_{\mathcal{J}_2}
&
U_{\mathcal{J}_1}^{\dagger}U_{\mathcal{J}_1}
+
U_{G_2}^{\dagger}U_{G_2}
\end{pmatrix}
=
I_{2n}.
\]
Thus,
\begin{align}
U_{G_1}^{\dagger}U_{G_1}
+
U_{\mathcal{J}_2}^{\dagger}U_{\mathcal{J}_2}
&=2I_n,
\label{eq:A}\\
U_{G_1}^{\dagger}U_{\mathcal{J}_1}
-
U_{\mathcal{J}_2}^{\dagger}U_{G_2}
&=0,
\label{eq:B}\\
U_{\mathcal{J}_1}^{\dagger}U_{G_1}
-
U_{G_2}^{\dagger}U_{\mathcal{J}_2}
&=0,
\label{eq:C}\\
U_{\mathcal{J}_1}^{\dagger}U_{\mathcal{J}_1}
+
U_{G_2}^{\dagger}U_{G_2}
&=2I_n.
\label{eq:D}
\end{align}
Since each block is unitary,
\(
U_{G_1}^{\dagger}U_{G_1}
=
U_{\mathcal{J}_1}^{\dagger}U_{\mathcal{J}_1}
=
U_{\mathcal{J}_2}^{\dagger}U_{\mathcal{J}_2}
=
U_{G_2}^{\dagger}U_{G_2}
=
I_n.
\)
Hence \eqref{eq:A} and \eqref{eq:D} are automatically satisfied.
Equations \eqref{eq:B} and \eqref{eq:C} give
\begin{equation}
U_{G_1}^{\dagger}U_{\mathcal{J}_1}
=
U_{\mathcal{J}_2}^{\dagger}U_{G_2},
\qquad
U_{\mathcal{J}_1}^{\dagger}U_{G_1}
=
U_{G_2}^{\dagger}U_{\mathcal{J}_2},
\end{equation}
which are precisely the stated off-diagonal conditions.

\noindent\textbf{($\Leftarrow$).}
Conversely, suppose
\(
U_{G_1}^{\dagger}U_{\mathcal{J}_1}
=
U_{\mathcal{J}_2}^{\dagger}U_{G_2},
\quad
U_{\mathcal{J}_1}^{\dagger}U_{G_1}
=
U_{G_2}^{\dagger}U_{\mathcal{J}_2}.
\)
Then the off-diagonal blocks in $U_G^\dagger U_G$ vanish.
Moreover, since all four blocks are unitary,
\(
U_{G_1}^{\dagger}U_{G_1}
+
U_{\mathcal{J}_2}^{\dagger}U_{\mathcal{J}_2}=
2I_n,
\quad
U_{\mathcal{J}_1}^{\dagger}U_{\mathcal{J}_1}
+
U_{G_2}^{\dagger}U_{G_2}
=
2I_n.
\)
Therefore,
\(
U_G^\dagger U_G=I_{2n}.
\)
Hence $U_G$ is unitary.
\end{proof}
A key challenge in this construction is the definition of suitable coupling matrices $U_{\mathcal{J}_1}$ and $U_{\mathcal{J}_2}$. This issue becomes especially pronounced when the component graphs $G_1$ and $G_2$ differ in size or internal structure, or when different coin-shift (or shift-coin) decompositions are employed. Even for joins of identical graphs, the dimensions and algebraic forms of the resulting coupling matrices may depend on the
type of join and the choice of shift matrices, as illustrated in
Example~\ref{eq.2.2}.
\begin{ex}\label{eq.2.2}
To illustrate the dimensional constraint in the shunt-decomposition
framework, consider the cycle graph $C_4$. Its adjacency matrix admits
the shunt decomposition
\(
A(C_4)=P_1+P_2,
\)
where
\(
P_1=
\begin{pmatrix}
0&1&0&0\\
0&0&1&0\\
0&0&0&1\\
1&0&0&0
\end{pmatrix},
\quad
P_2=
\begin{pmatrix}
0&0&0&1\\
1&0&0&0\\
0&1&0&0\\
0&0&1&0
\end{pmatrix}.
\)
Thus, the shift operator is
\(
S_{C_4}=P_1\otimes E_{11}+P_2\otimes E_{22}.
\)
Using the Hadamard coin
\(
H=\frac{1}{\sqrt2}
\begin{pmatrix}
1&1\\
1&-1
\end{pmatrix},
\quad
C=I_4\otimes H,
\)
the transition operator
\(
U_{G_1}=U_{G_2}=S_{C_4}C
\)
has dimension $8\times8$.
Now consider the join $C_4\vee C_4$, where the
underlying coupling graph is the complete bipartite graph
$K_{4,4}$. Thus, $\mathcal{J}=\mathcal{J}_1\cup\mathcal{J}_2$ is
$4$-regular directed graph, with
\(
A(\mathcal{J}_1)=J_4,
\quad
A(\mathcal{J}_2)=J_4.
\)
For signed coupling, the corresponding signed coupling matrix is
\(
A_{\mathcal{J}_{\pm}}
=
\begin{pmatrix}
0&J_4\\
-J_4&0
\end{pmatrix}.
\)
Since $J_4$ is $4$-regular, it admits a decomposition into four
permutation matrices,
\(
J_4=P_1+P_2+P_3+P_4.
\)
Therefore, the coupling shift operator is
\(
S_{\mathcal{J}_1}
=
\sum_{i=1}^4P_i\otimes E_{ii}.
\)
The coupling degree is $4$, so the Hadamard coin, which is a
degree-$2$ coin for $C_4$, is no longer appropriate. Instead, we use
the degree-$4$ Grover coin
\(
G_4=\frac{2}{4}J_4-I_4
=
\frac12
\begin{pmatrix}
-1&1&1&1\\
1&-1&1&1\\
1&1&-1&1\\
1&1&1&-1
\end{pmatrix}.
\)
With
\(
C_{\mathcal{J}_1}=I_4\otimes G_4,
\)
the coupling transition operator
\(
U_{\mathcal{J}_1}
=
S_{\mathcal{J}_1}C_{\mathcal{J}_1}
\)
has dimension $16\times16$. Similarly,
\(
U_{\mathcal{J}_2}=U_{\mathcal{J}_1}.
\)
Hence, the internal transition operators $U_{G_1}$ and $U_{G_2}$ have
dimension $8\times8$, whereas the coupling transition operators
$U_{\mathcal{J}_1}$ and $U_{\mathcal{J}_2}$ have dimension $16\times16$.
Consequently, the block transition operator~\eqref{transition_operator_join _Graphs}
cannot be formed because its blocks have incompatible dimensions.
Thus, the construction does not define a unitary transition operator
for this choice of internal graph and coupling graph, in accordance
with Proposition~\ref{prop4}.
\end{ex}
Thus, we focus on partial joins in which the degree of the coupling graphs $\mathcal{J}_1$ and $\mathcal{J}_2$ is matched to the degree of the underlying graphs $G_1$ and $G_2$, as discussed in Sections~\ref{sec3} and~\ref{sec4}. This choice is designed to ensure the dimensional compatibility of the associated shunt-decomposition-based quantum walk operators, so that unitarity of $U_G$ is preserved, as we establish below in Proposition~\ref{prop4}. Figures~\ref{Fig 2}, \ref{fig:C4-partial-join-case-a}, and \ref{fig:C4-partial-join-case-b} illustrate examples where the degree of the underlying graph matches that of the coupling graphs. We now derive conditions under which the transition matrix $U_G$ associated with a directed partial join of two $d$-regular directed graphs is unitary.

\begin{lemma}\label{lemma:6.2}
Let $G_1$ and $G_2$ be vertex-disjoint $d$-regular directed graphs with
$|V(G_1)|=|V(G_2)|=n$, and let
$\phi:V(G_1)\to V(G_2)$ be a fixed bijection with permutation matrix
$P_\phi$. Then:
\begin{enumerate}
    \item There exists a directed partial join
    \(
    G=G_1\overset{(d,d)}{\vec{\vee}}G_2,
    \)
    such that
    \[
    A(\mathcal J_1)=P_\phi^\top A(G_2)P_\phi,
    \qquad
    A(\mathcal J_2)=P_\phi A(G_1)P_\phi^\top.
    \]

    \item If $G_1\cong G_2$, relabel $G_2$ so that
    $A(G_2)=A(G_1)$. Then there exists a directed partial join
    \(
    G=G_1\overset{(d,d)}{\vec{\vee}}G_2
    \)
, such that
    \[
    A(\mathcal J_1)=A(G_1)=A(G_2),
    \qquad
    A(\mathcal J_2)=A(G_1)=A(G_2).
    \]
\end{enumerate}
In both cases, the associated signed-coupling transition operator
$U_G$ is unitary.
\end{lemma}

\begin{proof}
Fix $\phi:V(G_1)\to V(G_2)$ and define
\(
(u,v)\in E(\mathcal J_1)
\iff
(\phi(u),v)\in E(G_2),
\)
and
\(
(v,u)\in E(\mathcal J_2)
\iff
(\phi^{-1}(v),u)\in E(G_1).
\)
Hence,
\[
A(\mathcal J_1)=P_\phi^\top A(G_2)P_\phi,
\qquad
A(\mathcal J_2)=P_\phi A(G_1)P_\phi^\top.
\]
Permutation conjugation preserves row and column sums, so the
underlying coupling graph $\mathcal J=\mathcal J_1\cup\mathcal J_2$
is $d$-regular.
The bijections
\(
(\phi(u),v)\mapsto(u,v),
\quad
(a,u)\mapsto(\phi(a),u)
\)
transfer the shunts of $G_2$ and $G_1$ to those of $\mathcal J_1$ and
$\mathcal J_2$, respectively. Thus, with the same coin operator,
\(
U_{\mathcal J_1}=U_{G_2},
\quad
U_{\mathcal J_2}=U_{G_1}.
\)
If $G_1\cong G_2$, choose the labeling so that
$A(G_1)=A(G_2)$. Then the same shunt decomposition gives
\(
U_{G_1}=U_{G_2},
\quad
U_{\mathcal J_1}=U_{\mathcal J_2}=U_{G_1}=U_{G_2}.
\)
Assigning positive signs to $\mathcal J_1$ and negative signs to
$\mathcal J_2$, the corresponding signed-coupling transition operator
is given by~\eqref{transition_operator_join _Graphs}
Using
\(
U_{\mathcal J_1}=U_{G_2},
\quad
U_{\mathcal J_2}=U_{G_1},
\)
we have
\(
U_{G_1}^{\dagger}U_{\mathcal{J}_1}
=
U_{G_1}^{\dagger}U_{G_2}
=
U_{\mathcal{J}_2}^{\dagger}U_{G_2},\quad
U_{\mathcal{J}_1}^{\dagger}U_{G_1}
=
U_{G_2}^{\dagger}U_{G_1}
=
U_{G_2}^{\dagger}U_{\mathcal{J}_2}.
\)
Thus, both off-diagonal conditions in Proposition~\ref{prop4} are
satisfied. Hence, $U_G$ is unitary in both cases.
\end{proof}

\begin{lemma}\label{lem_shared_pi_coupling}
Let $G_1$ and $G_2$ be vertex-disjoint $d$-regular directed graphs on
$n$ vertices with $G_1\cong G_2$. Then there exists a directed partial
join
\(
G=G_1\overset{(d,d)}{\vec{\vee}}G_2
\)
with signed coupling such that
\[
A(\mathcal J_1)=A(\mathcal J_2)=AP_\pi\neq A,
\qquad
A=A(G_1)=A(G_2),
\]
where the arcs of $\mathcal J_1$ are assigned positive signs and those
of $\mathcal J_2$ are assigned negative signs. The associated
transition matrix $U_G$ is unitary.
\end{lemma}

\begin{proof}
Relabel $V(G_2)$ so that
\(
A(G_1)=A(G_2)=A.
\)
Let $\pi$ be an involution on the common vertex-label set such that
\(
\pi^2=\mathrm{id},
\quad
AP_\pi\neq A.
\)
Define
\(
A(\mathcal J_1)=A(\mathcal J_2)=AP_\pi,
\)
where $\mathcal J_1$ consists of arcs from $G_1$ to $G_2$ and
$\mathcal J_2$ consists of arcs from $G_2$ to $G_1$. Since right
multiplication by $P_\pi$ preserves the row and column sums of $A$,
$AP_\pi$ is a $d$-regular $0$-$1$ matrix. Hence
$\mathcal J=\mathcal J_1\cup\mathcal J_2$ is a $d$-regular directed
bipartite coupling graph.
Assign positive signs to the arcs of $\mathcal J_1$ and negative signs
to the arcs of $\mathcal J_2$. Thus, the signed coupling adjacency
matrix is
\[
A_{\mathcal J_\pm}
=
\begin{pmatrix}
0&AP_\pi\\
-AP_\pi&0
\end{pmatrix}.
\]
Let $A=\sum_{i=1}^{d}P_i$ be a shunt decomposition of $G_1$. Since $A(G_2)=A(G_1)$, choose the same shunt decomposition for $G_2$. Define
\(
Q_i=P_iP_\pi, \quad i=1,\ldots,d.
\)
Each $Q_i$ is a permutation matrix, and
\(
\sum_{i=1}^{d}Q_i = \left(\sum_{i=1}^{d}P_i\right)P_\pi = AP_\pi,
\)
so $\{Q_i\}_{i=1}^{d}$ gives a shunt decomposition of both $\mathcal{J}_1$ and $\mathcal{J}_2$. The corresponding shift matrices are
\(
S_{p}=\sum_{i=1}^{d}P_i\otimes E_{ii},
\quad
S_{q}=\sum_{i=1}^{d}Q_i\otimes E_{ii}.
\)
Using the same unitary coin $C$, we have
\(
U_{G_1}=U_{G_2}=S_{p}C,
\quad
U_{\mathcal{J}_1}=U_{\mathcal{J}_2}=S_{q}C.
\) Hence,
\[
S_p^\dagger S_q
=\sum_{i=1}^{d}P_i^\dagger Q_i\otimes E_{ii}
=\sum_{i=1}^{d}P_i^\dagger P_iP_\pi\otimes E_{ii}
=\sum_{i=1}^{d}P_\pi\otimes E_{ii}
=P_\pi\otimes I_d,
\]
using $P_i^\dagger P_i=I$ for each permutation matrix $P_i$. Since $\pi$ is an involution,
$P_\pi^\dagger=P_\pi$, and hence
$P_\pi\otimes I_d$ is Hermitian. Consequently,
\[
U_{G_1}^\dagger U_{\mathcal{J}_1} = (S_pC)^\dagger(S_qC) = C^\dagger(S_p^\dagger S_q)C = C^\dagger(P_\pi\otimes I_d)C
\]
is Hermitian (conjugation of a Hermitian matrix by any $C$ preserves the Hermitian property), so
\(
U_{G_1}^\dagger U_{\mathcal{J}_1} = U_{\mathcal{J}_1}^\dagger U_{G_1}.
\)
Since $U_{G_2}=U_{G_1}$ and $U_{\mathcal{J}_2}=U_{\mathcal{J}_1}$, this gives
\[
U_{G_1}^\dagger U_{\mathcal{J}_1} = U_{\mathcal{J}_1}^\dagger U_{G_1} = U_{\mathcal{J}_2}^\dagger U_{G_2},
\qquad
U_{\mathcal{J}_1}^\dagger U_{G_1} = U_{G_1}^\dagger U_{\mathcal{J}_1} = U_{G_2}^\dagger U_{\mathcal{J}_2}.
\]
Thus the compatibility conditions of Proposition~\ref{prop4} are satisfied. Since $U_{G_1},U_{G_2},U_{\mathcal{J}_1},U_{\mathcal{J}_2}$ are unitary, Proposition~\ref{prop4} implies that $U_G$ is unitary.
\end{proof}

\begin{cor}\label{cor_complement_loop_coupling}
Let $G_1\cong G_2$ be $d$-regular directed loopless graphs on
$n=2d$ vertices, with a shunt decomposition
\(
A(G_1)=\sum_{i=1}^d P_i.
\)
Let $J_n$ denote the $n\times n$ all-ones matrix, and define
\[
A(\mathcal{J}_1)=J_n-A(G_1),
\]
so that $\mathcal{J}_1$ is obtained from the complement of $G_1$ by
adding a loop at every vertex. Set $\mathcal{J}_2=\mathcal{J}_1$.
Then
\(
A(\mathcal{J}_1)=\sum_{i=1}^d Q_i
\)
is a shunt decomposition of the coupling matrix, and
\(
A(\mathcal{J}_1)\neq A(G_1).
\)
Moreover,
\[
U_G\text{ is unitary}
\iff
P_i^\top Q_i\text{ is symmetric for every }i=1,\ldots,d.
\]
\end{cor}

\begin{proof}
Since $J_n$ is $n$-regular and $A(G_1)$ is $d$-regular, with
$n=2d$, every row and column of
\(
A(\mathcal{J}_1)=J_n-A(G_1)
\)
has sum $d$. Thus, both coupling parts $\mathcal{J}_1$ and
$\mathcal{J}_2$ have coupling degree $d$, and
\(
\mathcal{J}=\mathcal{J}_1\cup\mathcal{J}_2
\)
is a $d$-regular directed bipartite coupling graph. By the shunt
decomposition result in Proposition~\ref{prop:signed-coupling-shunt},
there exist permutation matrices $Q_1,\ldots,Q_d$ such that
\(
A(\mathcal{J}_1)=A(\mathcal{J}_2)
=\sum_{i=1}^d Q_i.
\)
Since $G_1$ is loopless, $A(G_1)$ has zero diagonal, whereas
$A(\mathcal{J}_1)$ has diagonal entries equal to one. Hence
\(
A(\mathcal{J}_1)\neq A(G_1).
\)
Since $G_1\cong G_2$, relabel $V(G_2)$ so that
\(
A(G_2)=A(G_1).
\)
Choose the same shunt decomposition $\{P_i\}_{i=1}^d$ and the same
unitary coin $C$ for $G_1$ and $G_2$. Thus,
\(
U_{G_1}=U_{G_2}.
\)
Similarly, since $\mathcal{J}_1=\mathcal{J}_2$, choose the same shunt
decomposition $\{Q_i\}_{i=1}^d$ and the same coin for both coupling
parts, giving
\(
U_{\mathcal{J}_1}=U_{\mathcal{J}_2}.
\) 
Since $U_{G_1}=U_{G_2}$ and
$U_{\mathcal{J}_1}=U_{\mathcal{J}_2}$, the two cross-conditions in
Proposition~\ref{prop4} reduce to
\(
U_{G_1}^{\dagger}U_{\mathcal{J}_1}
=
U_{\mathcal{J}_1}^{\dagger}U_{G_1}.
\)
Thus, $U_G$ is unitary if and only if
$U_{G_1}^{\dagger}U_{\mathcal{J}_1}$ is Hermitian. Write
\[
U_{G_1}=U_{G_1}=S_pC,
\qquad
U_{\mathcal{J}_1}=U_{\mathcal{J}_2}=S_qC,
\]
where
\(
S_p=\sum_{i=1}^d P_i\otimes E_{ii},
\quad
S_q=\sum_{i=1}^d Q_i\otimes E_{ii},
\)
and $C$ is the common unitary coin. Then
\(
U_{G_1}^\dagger U_{\mathcal{J}_1}
=
C^\dagger S_p^\dagger S_q C.
\)
Since conjugation by a unitary matrix preserves Hermiticity,
\[
U_{G_1}^\dagger U_{\mathcal{J}_1}
\text{ is Hermitian}
\iff
S_p^\dagger S_q
\text{ is Hermitian}.
\]
Furthermore,
\(
S_p^\dagger S_q
=
\sum_{i=1}^d
(P_i^\top Q_i)\otimes E_{ii}.
\)
This matrix is block-diagonal with respect to the coin index.
Therefore,
\(
S_p^\dagger S_q
\text{ is Hermitian}
\iff
P_i^\top Q_i
\text{ is symmetric for every }i=1,\ldots,d.
\)
Hence, by Proposition~\ref{prop4},
\(
U_G\text{ is unitary}
\iff
P_i^\top Q_i\text{ is symmetric for every }i=1,\ldots,d.
\)

\end{proof}

\begin{ex}
Let $G_1\cong C_4$, with vertices labelled $1,2,3,4$ cyclically. Then $A(C_4)=P_1+P_2$, where $P_1$ is the cyclic shift by $-1$ and $P_2=P_1^\top$ is the shift by $+1$. The all-ones diagonal of $A(\mathcal{J}_1)$ forces
\(
Q_1=I, \quad Q_2=A(\mathcal{J}_1)-I=P_\sigma, \quad \sigma=(1\,3)(2\,4).
\)
Then
\(
P_1^\top Q_1 = P_1^\top = P_2,
\)
which is the order-$4$ cyclic shift, and satisfies $P_2^\top=P_1\neq P_2$, so $P_1^\top Q_1$ is not symmetric. The criterion of Corollary~\ref{cor_complement_loop_coupling} therefore fails, so $U_G$ is not unitary.
\end{ex}

\section{Spectral Conditions for Periodicity and PST in Directed Partial Join Graphs}
\label{sec6}

In this section, we consider the case in which the two component
graphs and the two coupling parts have the same adjacency matrix,
namely,
\begin{equation}\label{similar_adjacency}
A(G_1)=A(G_2)=A(\mathcal{J}_1)=A(\mathcal{J}_2)=A.
\end{equation}
Thus, the signed adjacency matrix of the corresponding partial join is
\begin{equation}
A_{\pm}(G)
=
\begin{pmatrix}
A&A\\
-A&A
\end{pmatrix}.
\end{equation}
This structure occurs for several graph families introduced in
Section~\ref{sec4}, including the complete graph with loops at all ends
$K_{2n}^{\circlearrowleft}$, the complete bipartite graph $K_{n,n}$
in Eq.~\eqref{signed adjacency matrices}, the tensor product graph
$(K_{n,n})^{\otimes n}$, $K_2^{\otimes n}$, and the circulant join
graph $C_n(S)$ with $S=\{1,\ldots,n/2\}$, all with signed coupling, as
described in Proposition~\ref{prop:general-even-n}.

We now derive spectral conditions for periodicity and PST for this class of $d$-regular directed partial join
graphs, as also defined in Eq.
\eqref{block_equal}.
To describe the coupling between the two copies more generally, let
\(
0\leq\phi\leq\frac{\pi}{2},
\) so that $\cos\phi$ and $\sin\phi$ are nonnegative weights.
 Assign the weight $\cos\phi$ to the internal parts $G_1$ and
$G_2$, and the weight $\sin\phi$ to the coupling parts
$\mathcal{J}_1$ and $\mathcal{J}_2$. Thus, the weighted signed adjacency matrix is given by
\begin{equation}\label{eq:weighted-signed-adjacency}
A_{\pm}^{(\phi)}(G)
=
\begin{pmatrix}
\cos\phi\,A(G_1)&\sin\phi\,A(\mathcal{J}_1)\\
-\sin\phi\,A(\mathcal{J}_2)&\cos\phi\,A(G_2)
\end{pmatrix}.
\end{equation}
Under the assumption given in~\eqref{similar_adjacency},
Eq.~\eqref{eq:weighted-signed-adjacency} becomes
\begin{equation}\label{eq:weighted-signed-adjacency-factor}
A_{\pm}^{(\phi)}(G)
=
\begin{pmatrix}
\cos\phi\,A&\sin\phi\,A\\
-\sin\phi\,A&\cos\phi\,A
\end{pmatrix}
=
M_\phi\otimes A,
\end{equation}
where
\begin{equation}
    M_\phi
=
\begin{pmatrix}
\cos\phi&\sin\phi\\
-\sin\phi&\cos\phi
\end{pmatrix}.
\end{equation}
Thus, $M_\phi$, also called the \emph{rotation matrix}, describes the
weighted signed coupling structure between the two copies, while $A$
describes the common adjacency matrix within each block. The transition operator inherits the same structure. Suppose
that the corresponding transition operators satisfy
\begin{equation}\label{equal_transition}
U_{G_1}=U_{G_2}
=U_{\mathcal{J}_1}=U_{\mathcal{J}_2}=U.
\end{equation}
Then the weighted signed transition operator is
\begin{equation}\label{eq:coupled-transition}
U_G^{(\phi)}
=
\begin{pmatrix}
\cos\phi\,U&\sin\phi\,U\\
-\sin\phi\,U&\cos\phi\,U
\end{pmatrix}
=
M_\phi\otimes U,
\end{equation}
where $U$ describes the common quantum-walk evolution
associated with the underlying graph. Since
\(
M_\phi^{\dagger}M_\phi=I_2
\)
for every $\phi\in\mathbb{R}$, $M_\phi$ is unitary. Therefore, if
$U$ is unitary, then
\[
U_G^{(\phi)}=M_\phi\otimes U
\]
is unitary.
For
\(
\phi=\frac{\pi}{4},
\)
we obtain
\(
M_{\pi/4}
=
\frac{1}{\sqrt{2}}
\begin{pmatrix}
1&1\\
-1&1
\end{pmatrix},
\)
and hence
\begin{equation}\label{eq:equal-transition}
U_G^{(\pi/4)}
=
\frac{1}{\sqrt{2}}
\begin{pmatrix}
U&U\\
-U&U
\end{pmatrix}.
\end{equation}
Thus,
\(
U_G^{(\pi/4)}=U_G,
\)
where $U_G$ is the transition operator in
Eq.~\eqref{transition_operator_join _Graphs} when all four block
transition operators are equal, as in
Eq.~\eqref{equal_transition}. We refer to $U_G^{(\phi)}$ as a \emph{coupling} of the quantum walks
determined by $U_{G_1}$ and $U_{G_2}$, following the coupling
construction discussed in \cite[Sec.~8.2]{Godsil_Zhan_2023}. This
representation separates the internal quantum evolution, described by
$U$, from the two-component coupling described by the rotation matrix
$M_\phi$. Consequently, the eigenvalues of $U_G^{(\phi)}$ can be
determined directly from the eigenvalues of $U$.

Every unitary matrix $U$ admits a spectral decomposition
\begin{equation}
U = \sum_{r} e^{i\theta_r} E_r,
\end{equation}
where each $e^{i\theta_r}$ is an eigenvalue of $U$ and $E_r$ is the Hermitian projection onto the corresponding eigenspace~\cite[Sec.~7]{GodsilZhan2019}. Given a vertex state $e_a$, we define its \emph{eigenvalue support}
\begin{equation}
\Theta_{e_a} = \{\, r : E_r e_a \neq 0 \,\}
\end{equation}
as the set of indices $r$ for which $e_a$ has a nonzero component in the $r$-th eigenspace~\cite[Sec.~4]{ChanZhan2023}.

\begin{prop}\label{prop:general-angle}
Fix any $\phi\in\mathbb{R}$, and let $U_G^{(\phi)}$ be the coupled walk determined by $U$. Let $U=\sum_r e^{i\theta_r}E_r$ be the spectral decomposition of $U$. Then the spectrum of $U_G^{(\phi)}$ is
\(
\operatorname{Spec}(U_G^{(\phi)}) = \{e^{i(\theta_r+\phi)},\ e^{i(\theta_r-\phi)}\}.
\)
Moreover, for $v\in\operatorname{im}(E_r)$, the vectors $\begin{pmatrix}v\\ \pm iv\end{pmatrix}$ are eigenvectors of $U_G^{(\phi)}$ with eigenvalues $e^{i(\theta_r\pm\phi)}$.
\end{prop}

\begin{proof}
We have by~\eqref{eq:coupled-transition}, $U_G^{(\phi)} = M_\phi\otimes U$, where $M_\phi=\begin{pmatrix}\cos\phi&\sin\phi\\
-\sin\phi&\cos\phi\end{pmatrix}$ is a real rotation matrix. Since $M_\phi$ is orthogonal
($M_\phi^\top M_\phi=I$) and $U$ is unitary, $U_G^{(\phi)}$ is unitary. The eigenvalues of
$M_\phi$ are $e^{\pm i\phi}$ (standard fact for a rotation of $2\times2$ at angle $\phi$), with
eigenvectors $(1,\pm i)$, indeed $M_\phi(1,i)^T = (\cos\phi+i\sin\phi,\,-\sin\phi+i\cos\phi)^T
= e^{i\phi}(1,i)^T$, and similarly for $(1,-i)$ with eigenvalue $e^{-i\phi}$. Hence, for
$v\in\operatorname{im}(E_r)$,
\[
U_G^{(\phi)}\begin{pmatrix}v\\ \pm iv\end{pmatrix}
= (M_\phi\otimes U)\begin{pmatrix}v\\ \pm iv\end{pmatrix}
= \begin{pmatrix}\cos\phi\,Uv+\sin\phi\,U(\pm iv)\\ -\sin\phi\,Uv+\cos\phi\,U(\pm iv)\end{pmatrix}
= e^{i\theta_r}\begin{pmatrix}\cos\phi\pm i\sin\phi\\ -\sin\phi\pm i\cos\phi\end{pmatrix}v,
\]
using $Uv=e^{i\theta_r}v$. Since $\cos\phi\pm i\sin\phi=e^{\pm i\phi}$ and $-\sin\phi\pm
i\cos\phi=\pm i(\cos\phi\pm i\sin\phi)=\pm ie^{\pm i\phi}$, this equals
$e^{i(\theta_r\pm\phi)}\begin{pmatrix}v\\\pm iv\end{pmatrix}$, as claimed.
\end{proof}
\begin{cor}
\label{cor:Knn-general-angle}
Let $U=\sum_r e^{i\theta_r}E_r$ be the unitary transition matrix associated with the signed coupling of
$K_{n,n}$ . At $\phi=\pi/2$, the transition operator in
Proposition~\ref{prop:general-angle} is
\[
U_{K_{n,n}}
=
\begin{pmatrix}
0&U\\
-U&0
\end{pmatrix},
\]  since $\cos\phi=0$ and $\sin\phi=1$. Then $U_{K_{n,n}}$ is unitary, with spectrum
\(
\operatorname{Spec}\big(U_{K_{n,n}}\big) = \{e^{i(\theta_r+\pi/2)},\ e^{i(\theta_r-\pi/2)}\},
\)
and for $v\in\operatorname{im}(E_r)$, the vectors $\begin{pmatrix}v\\ \pm iv\end{pmatrix}$ are
eigenvectors of $U_{K_{n,n}}$ with eigenvalues $e^{i(\theta_r\pm\pi/2)}$, respectively.
\end{cor}
\begin{proof}
Immediate specialization of Proposition~\ref{prop:general-angle} at $\phi=\pi/2$: the diagonal
blocks $\cos\phi\,U$ vanish and the off-diagonal blocks $\mp\sin\phi\,U$ become $\mp U$,
matching the block structure $U_{K_{n,n}}=\begin{pmatrix}0&U\\-U&0\end{pmatrix}$.
\end{proof}

\begin{lemma}\label{lem:periodicity-coupled-walk}
Let 
\(
U=\sum_{r}e^{i\theta_r}E_r,
\) be the spectral decomposition of unitary matrix \(U\). Let \(a_i\in V(G_1)\) with $e_{a_i}=
\begin{pmatrix}
e_{a_i}^{(1)}\\
0
\end{pmatrix}$ be the corresponding basis state. Define the eigenvalue support of
$e_{a_i}^{(1)}$ by
\(
\Theta_{e_{a_i}^{(1)}}.
\)
Consider the coupled transition matrix
\(
U_G^{(\phi)}\).
Then $G$ is periodic at \(a_i\) with period \(\tau\in \mathbb{Z}_{>0}\)
if and only if the following
conditions hold:

\begin{enumerate}
\item[(i)] 
\(
\tau\phi\in\pi\mathbb{Z};
\)
\item[(ii)]
\(
\tau(\theta_r-\theta_s)\in2\pi\mathbb{Z}
\quad
\text{for all }r,s\in\Theta_{e_a}.
\)
\end{enumerate}
In particular, if
\(
\phi=\frac{p\pi}{q},
\quad \gcd(p,q)=1,
\)
then condition~(i) is equivalent to
\(
q\mid\tau.
\)
Moreover, suppose there exists a positive real number $g$ such that
\(\
\theta_r-\theta_0=k_r g,
\quad
k_r\in\mathbb{Z},
\quad
r\in\Theta_{e_{a_i}^{(1)}},
\)
with
\(
\gcd\{k_r:r\in\Theta_{e_{a_i}^{(1)}}\}=1.
\)
If
\(
g=l\psi,
\quad
\psi=\frac{u\pi}{v},
\)
for positive integers $l,u,v$, then the minimum positive integer
satisfying condition~(ii) is
\(
\tau_{\min}^{(ii)}
=
\frac{2v}{\gcd(lu,2v)}.
\)
Consequently, the minimum period of $G$ at $a_i$ is
\[
{
\tau_{\min}
=
\operatorname{lcm}
\left(
q,\,
\frac{2v}{\gcd(lu,2v)}
\right).
}
\]
\end{lemma}
\begin{proof}
We have by~\eqref{eq:coupled-transition}
\(
U_G^{(\phi)}
=
M_\phi\otimes U,
\qquad
M_\phi=
\begin{pmatrix}
\cos\phi&\sin\phi\\
-\sin\phi&\cos\phi
\end{pmatrix},
\)
we get
\begin{equation}
\left(U_G^{(\phi)}\right)^\tau
=
M_\phi^\tau\otimes U^\tau
=
\begin{pmatrix}
\cos(\tau\phi)\,U^\tau&
\sin(\tau\phi)\,U^\tau\\
-\sin(\tau\phi)\,U^\tau&
\cos(\tau\phi)\,U^\tau
\end{pmatrix}.
\end{equation}
Therefore,
\begin{equation}\label{eq:power-coupled}
\left(U_G^{(\phi)}\right)^\tau
\begin{pmatrix}
e_{a_i}^{(1)}\\
0
\end{pmatrix}
=
\begin{pmatrix}
\cos(\tau\phi)\,U^\tau e_{a_i}^{(1)}\\
-\sin(\tau\phi)\,U^\tau e_{a_i}^{(1)}
\end{pmatrix}.
\end{equation}

\noindent\textbf{($\Rightarrow$).}
Suppose that $G$ is periodic at $a_i$ at time
$\tau\in\mathbb{Z}_{>0}$. Then there exists a unimodular constant
$\gamma$, $|\gamma|=1$, such that
\[
\left(U_G^{(\phi)}\right)^\tau
\begin{pmatrix}
e_{a_i}^{(1)}\\
0
\end{pmatrix}
=
\gamma
\begin{pmatrix}
e_{a_i}^{(1)}\\
0
\end{pmatrix}.
\]
Comparing this equation with \eqref{eq:power-coupled}, the lower block
gives
\(
\sin(\tau\phi)\,U^\tau e_{a_i}^{(1)}=0.
\)
Since $U$ is unitary and $e_{a_i}^{(1)}\neq0$,
\(
U^\tau e_{a_i}^{(1)}\neq0.
\)
Hence
\(
\sin(\tau\phi)=0,
\)
and therefore
\(
\tau\phi\in\pi\mathbb{Z}.
\)
This proves condition~(i). Now the upper block of \eqref{eq:power-coupled} gives
\(
\cos(\tau\phi)\,U^\tau e_{a_i}^{(1)}
=
\gamma e_{a_i}^{(1)}.
\)
Since $\sin(\tau\phi)=0$, we have
\(
\cos(\tau\phi)\in\{1,-1\}.
\)
Thus there exists a unimodular constant $\gamma'$ such that
\(
U^\tau e_{a_i}^{(1)}=\gamma' e_{a_i}^{(1)}.
\)
Using the spectral decomposition
\(
U=\sum_r e^{i\theta_r}E_r,
\)
we obtain
\begin{equation}
U^\tau e_{a_i}^{(1)}
=
\sum_{r\in\Theta_{e_{a_i}^{(1)}}}
e^{i\tau\theta_r}E_r e_{a_i}^{(1)}.
\end{equation}
Since
\(
e_{a_i}^{(1)}
=
\sum_{r\in\Theta_{e_{a_i}^{(1)}}}E_r e_{a_i}^{(1)},
\)
the equality
\(
U^\tau e_{a_i}^{(1)}=\gamma' e_{a_i}^{(1)}
\)
holds if and only if
\(
e^{i\tau\theta_r}=\gamma'
\quad
\text{for all }r\in\Theta_{e_{a_i}^{(1)}}.
\)
Equivalently,
\(
e^{i\tau(\theta_r-\theta_s)}=1
\quad
\text{for all }r,s\in\Theta_{e_{a_i}^{(1)}},
\)
which is equivalent to
\(
\tau(\theta_r-\theta_s)\in2\pi\mathbb{Z}
\quad
\text{for all }r,s\in\Theta_{e_{a_i}^{(1)}}.
\)
Thus condition~(ii) holds.
If
\(
\phi=\frac{p\pi}{q},
\quad
\gcd(p,q)=1,
\)
then condition~(i) becomes
\(
\frac{\tau p}{q}\in\mathbb{Z}.
\)
Since $\gcd(p,q)=1$, this is equivalent to
\(
q\mid\tau.
\)

\noindent\textbf{($\Leftarrow$).}
Conversely, suppose that $\tau\in\mathbb{Z}_{>0}$ satisfies
conditions~(i) and~(ii). From condition~(i),
\(
\sin(\tau\phi)=0.
\)
Hence the lower block in \eqref{eq:power-coupled} vanishes.
By condition~(ii), for all $r,s\in\Theta_{e_{a_i}^{(1)}}$,
\(
e^{i\tau(\theta_r-\theta_s)}=1.
\)
Fix $r_0\in\Theta_{e_{a_i}^{(1)}}$. Then
\(
e^{i\tau\theta_r}
=
e^{i\tau\theta_{r_0}}
e^{i\tau(\theta_r-\theta_{r_0})}
=
e^{i\tau\theta_{r_0}}
=\lambda
\)
for every $r\in\Theta_{{e_{a_i}^{(1)}}}$. Since $|\lambda|=1$,
\begin{equation}
U^\tau e_{a_i}^{(1)}
=
\sum_{r\in\Theta_{e_{a_i}^{(1)}}}
e^{i\tau\theta_r}E_r e_{a_i}^{(1)}
=
\lambda
\sum_{r\in\Theta_{e_{a_i}^{(1)}}}E_r e_{a_i}^{(1)}
=
\lambda e_{a_i}^{(1)}.
\end{equation}
Moreover, since $\tau\phi\in\pi\mathbb{Z}$,
\(
\cos(\tau\phi)\in\{1,-1\}.
\)
Therefore, setting
\(
\gamma
=
\cos(\tau\phi)\lambda,
\)
we have $|\gamma|=1$, and
\begin{equation}
\left(U_G^{(\phi)}\right)^\tau
\begin{pmatrix}
e_{a_i}^{(1)}\\
0
\end{pmatrix}
=
\begin{pmatrix}
\gamma e_{a_i}^{(1)}\\
0
\end{pmatrix}
=
\gamma
\begin{pmatrix}
e_{a_i}^{(1)}\\
0
\end{pmatrix}.
\end{equation}
Thus, $\tau$ is a period of $G$ at $a_i$.
Now we will show that the minimum value of $\tau$ satisfying
condition~(ii). Suppose that
\(
\theta_r-\theta_0=k_rg,
\quad
k_r\in\mathbb{Z},
\)
where
\(
\gcd\{k_r:r\in\Theta_{e_{a_i}^{(1)}}\}=1.
\)
Then
\[
\theta_r-\theta_s
=
(k_r-k_s)g
=
k_{rs}g,
\qquad
k_{rs}\in\mathbb{Z}.
\]
If
\(
g=l\psi,
\quad
\psi=\frac{u\pi}{v},
\)
then
\(
g=\frac{lu\pi}{v},
\)
and condition~(ii) becomes
\(
\tau k_{rs}\frac{lu\pi}{v}
\in2\pi\mathbb{Z}.
\)
Equivalently,
\(
2v\mid\tau k_{rs}lu
\)
for all $r,s\in\Theta_{e_{a_i}^{(1)}}$.
Since
\(
\gcd\{k_r:r\in\Theta_{e_{a_i}^{(1)}}\}=1,
\)
we also have
\(
\gcd\{k_{rs}:r,s\in\Theta_{e_{a_i}^{(1)}}\}=1.
\)
Indeed, taking $r=0$ gives
\(
k_{0s}=k_0-k_s=-k_s,
\)
where $k_0=0$. Thus, the set of differences $\{k_{rs}\}$ contains
$\{-k_s\}$, and hence has greatest common divisor~$1$.
Consequently,
\(
2v\mid\tau k_{rs}lu
\quad\text{for all }r,s
\)
is equivalent to
\(
2v\mid\tau lu.
\)
Therefore, the minimum positive integer satisfying condition~(ii) is
\(
\tau_{\min}^{(ii)}
=
\frac{2v}{\gcd(lu,2v)}.
\)
Finally, condition~(i) requires
\(
q\mid\tau,
\)
while condition~(ii) requires
\(
\tau_{\min}^{(ii)}\mid\tau.
\)
Hence the minimum positive integer satisfying both conditions is
\begin{equation}
{
\tau_{\min}
=
\operatorname{lcm}
\left(
q,\,
\frac{2v}{\gcd(lu,2v)}
\right).
}
\end{equation}
By the equivalence proved above, this value is indeed a period, and no
smaller positive integer can be a period. This completes the proof.
\end{proof}
\begin{remark}
Assume that $\Theta_{e_a^{(1)}}$ contains at least two distinct eigenvalue
angles. Then periodicity of $G$ at $a_i$ requires both
$\phi/\pi$ and $g/\pi$ to be rational. If either $\phi/\pi$ or
$g/\pi$ is irrational, then $G$ is not periodic at $a_i$.
\end{remark}

\begin{cor}
\label{cor:periodicity-phi-pi4}
Under the assumptions of Lemma~\ref{lem:periodicity-coupled-walk},
suppose, in addition, that
\(
\psi=\phi=\frac{p\pi}{q}
\quad\text{and}\quad
g=l\psi.
\)
Then
\[
\tau_{\min}
=
\operatorname{lcm}\left(
q,\,
\frac{2q}{\gcd(l p,2q)}
\right).
\]
In particular, for
\(
\phi=\psi=\frac{\pi}{4},
\quad p=1,\quad q=4,
\)
we have
\(
g=l\psi=\frac{l\pi}{4}.
\) $G$ is periodic at \(a_i\) with period \(\tau\in \mathbb{Z}_{>0}\) if and only if
\(
4\mid\tau
\)
and
\(
\tau(\theta_r-\theta_s)\in2\pi\mathbb{Z}
\quad
\text{for all }r,s\in\Theta_{e_a^{(1)}}.
\)
Moreover,
\[
{
\tau_{\min}
=
\operatorname{lcm}\left(
4,\,
\frac{8}{\gcd(l,8)}
\right)
\in\{4,8\}.
}
\]
\end{cor}

\begin{ex}
\label{ex:periodicity-Knn}
Consider the coupled transition operator at
\(
\phi=\frac{\pi}{2},
\)
given by
\begin{equation}\label{transition_K_n_n}
U_G^{(\phi)}
=
\begin{pmatrix}
0&U\\
-U&0
\end{pmatrix},
\qquad
U=\sum_r e^{i\theta_r}E_r.
\end{equation}
This case applies to the graph families discussed in Section~\ref{sec4},
including the complete bipartite graph $K_{n,n}$ in Eq.~\eqref{signed adjacency matrices}, the tensor product
graph $(K_{n,n})^{\otimes n}$, and $K_2^{\otimes n}$, all with signed coupling. For $a_i\in V(G_1)$, let
\(
e_{a_i}=
\begin{pmatrix}
e_{a_i}^{(1)}\\
0
\end{pmatrix},
\quad
g=\gcd_{2\pi}\{\theta_0-\theta_r:r\in\Theta_{e_{a_i}^{(1)}}\},
\)
where
\(
\theta_0=\max_{r\in\Theta_{e_{a_i}^{(1)}}}\theta_r.
\)
Writing
\(
\frac{g}{\pi}=\frac{u}{v}
\)
in lowest terms, periodicity of $G$ at $a_i$ requires
\(
2\mid\tau
\)
and
\(
\tau(\theta_r-\theta_s)\in2\pi\mathbb{Z}
\quad
\text{for all }r,s\in\Theta_{e_{a_i}^{(1)}}.
\)
Hence, the minimum period is
\begin{equation}
\tau_{\min}
=
\operatorname{lcm}
\left(
2,\frac{2v}{\gcd(u,2v)}
\right).
\end{equation}
In particular, if
\(
g=\frac{l\pi}{4}
\)
for a unique positive integer $l$, then
\(
\frac{g}{\pi}=\frac{l}{4},
\)
and therefore
\(
\tau_{\min}
=
\operatorname{lcm}\left(
2,\frac{8}{\gcd(l,8)}
\right)
\in\{2,4,8\}.
\)
More explicitly,
\begin{center}
{\small
\begin{tabular}{cccc}
\hline
\(\gcd(l,8)\) & \(8/\gcd(l,8)\) &
\(\operatorname{lcm}(2,8/\gcd(l,8))\) & \(\tau_{\min}\)\\
\hline
1 & 8 & 8 & 8\\
2 & 4 & 4 & 4\\
4 & 2 & 2 & 2\\
8 & 1 & 2 & 2\\
\hline
\end{tabular}}
\end{center}
Thus, in this case, the minimum period can only be
\(
{\tau_{\min}\in\{2,4,8\}},
\)
with its value determined by the eigenvalue support of ${e_{a_i}^{(1)}}$ under
$U$.
\end{ex}
Using the preceding condition of periodicity, we now derive
necessary and sufficient conditions for PST in the following theorems and illustrate them with explicit
example.
\begin{Theorem}
\label{thm:pst-coupling}
Let
\(
G=G_1\overset{(d,d)}{\vec{\vee}}G_2
\)
be a $d$-regular directed partial join with signed coupling and couple
transition operator $U_G^{(\phi)}$,  $U_G^{(\phi)}$, where
\(
\phi=\frac{p\pi}{q},
\quad
\gcd(p,q)=1.
\)
Fix $a_i,a_j\in V(G_1)$, and let $b_i,b_j\in V(G_2)$ denote the
vertices corresponding to $a_i,a_j$, respectively. Define
\[
e_{a_i}=
\begin{pmatrix}
e_{a_i}^{(1)}\\
0
\end{pmatrix},
\qquad
e_{b_j}=
\begin{pmatrix}
0\\
e_{a_j}^{(1)}
\end{pmatrix}.
\]
Then coupling PST from $a_i$ to $b_j$ under
$U_G^{(\phi)}$ occurs at a positive integer time $k$ if and only if
\[
k\equiv\frac{q}{2}\pmod q
\]
and
\(
U^k e_{a_i}^{(1)}
=
\mu e_{a_j}^{(1)}
\)
for some $\mu\in\mathbb{C}$ with $|\mu|=1$.
In particular, coupling PST is possible only when $q$ is even.
If $G$ is periodic at $a_i$ with minimum period
\(
\tau_{\min}=sq,
\quad
s\in\mathbb{Z}_{>0},
\)
then the possible coupling-PST times within one period are
\[
\left\{
\frac{q}{2},
\frac{3q}{2},
\ldots,
\frac{(2s-1)q}{2}
\right\}.
\]
Consequently, if coupling PST occurs, its minimum time is the
smallest element of this set for which
\[
U^k e_{a_i}^{(1)}
=
\mu e_{a_j}^{(1)},
\qquad
|\mu|=1.
\]
\end{Theorem}

\begin{proof}
For every $k\geq0$, by~\eqref{eq:coupled-transition}, we have
\begin{equation}
\bigl(U_G^{(\phi)}\bigr)^k e_{a_i}
=
\begin{pmatrix}
\cos(k\phi)\,U^k e_{a_i}^{(1)}\\
-\sin(k\phi)\,U^k e_{a_i}^{(1)}
\end{pmatrix}.
\end{equation}
For coupling PST from $a_i$ to $b_j$, the upper block must vanish.
Since $U$ is unitary,
\(
U^k e_{a_i}^{(1)}\neq0,
\)
and hence
\(
\cos(k\phi)=0.
\)
Equivalently,
\(
\frac{2kp}{q}\in2\mathbb{Z}+1.
\)
Since $\gcd(p,q)=1$,
\(
q\mid2kp
\iff
q\mid2k.
\)
If $q$ is odd, then $q\mid2k$ implies $q\mid k$, and therefore
\(
\frac{2kp}{q}
\)
is even. Hence, it cannot belong to $2\mathbb{Z}+1$, so no such
$k$ exists. Thus, coupling PST is impossible when $q$ is odd.
Now suppose that $q=2q'$ is even. Since $\gcd(p,q)=1$, $p$ is odd.
Moreover,
\(
q\mid2k
\iff
q'\mid k.
\)
Writing $k=q't$, we obtain
\(
\frac{2kp}{q}=tp,
\)
which is odd if and only if $t$ is odd. Hence,
\begin{equation}\label{eq_16-phi_}
\cos(k\phi)=0
\iff
k=q'(2m+1)
\iff
k\equiv\frac{q}{2}\pmod q.
\end{equation}
Therefore, coupling PST is possible only for even $q$, and at such
times
\(
|\sin(k\phi)|=1.
\)

Conversely, suppose that
\(
k\equiv\frac{q}{2}\pmod q
\)
and
\(
U^k e_{a_i}^{(1)}
=
\mu e_{a_j}^{(1)},
\quad
|\mu|=1.
\)
Then
\(
\cos(k\phi)=0
\)
and
\(
\sin(k\phi)=\pm1.
\)
Therefore,
\begin{equation}
\bigl(U_G^{(\phi)}\bigr)^k e_{a_i}
=
\begin{pmatrix}
0\\
-\sin(k\phi)\,\mu e_{a_j}^{(1)}
\end{pmatrix}.
\end{equation}
By the definition of the corresponding state $e_{b_j}$,
and hence
\(
\bigl(U_G^{(\phi)}\bigr)^k e_{a_i}
=
\gamma e_{b_j},
\)
where
\(
\gamma=-\sin(k\phi)\,\mu.
\)
Since
\(
|\sin(k\phi)|=1
\quad\text{and}\quad
|\mu|=1,
\)
we have
\(
|\gamma|=1.
\)
Thus, coupling PST from $a_i$ to $b_j$ occurs at time $k$.

Now assume that $G$ is periodic at $a_i$ with minimum
period
\(
\tau_{\min}=sq.
\)
Then
\begin{equation}
\bigl(U_G^{(\phi)}\bigr)^{k+\tau_{\min}} e_{a_i}
=
\delta\,
\bigl(U_G^{(\phi)}\bigr)^k e_{a_i},
\qquad
|\delta|=1,
\end{equation}
for all $k\geq1$. Hence, it suffices to search for the minimum
coupling-PST time in the range
\(
1\leq k<\tau_{\min}.
\)
Within this range, the times satisfying \eqref{eq_16-phi_} are exactly
\(
k_0^{(l)}
=
\frac{q}{2}+lq,
\quad
l=0,1,\ldots,s-1.
\)
Thus, within one period, coupling PST can occur only at the times
\begin{equation}
\left\{
\frac{q}{2},
\frac{3q}{2},
\ldots,
\frac{(2s-1)q}{2}
\right\}.
\end{equation}
Therefore, the minimum coupling-PST is the smallest time in this set for
which
\(
U^k  e_{a_i}^{(1)}
=
\mu  e_{a_j}^{(1)},
\quad
|\mu|=1.
\)
In particular, if $s=1$, then $k_{\min}=q/2$. If $s=2$, then
\(
k_{\min}\in\left\{\frac{q}{2},\frac{3q}{2}\right\}.
\)
In general,
\[
k_{\min}\in
\left\{
\frac{q}{2},\frac{3q}{2},\ldots,
\frac{(2s-1)q}{2}
\right\}.
\]
Finally, at $k=k_{\min}$,
\(
\cos(k_{\min}\phi)=0,
\quad
\sin(k_{\min}\phi)=\pm1,
\)
and hence
\begin{equation}
\bigl(U_G^{(\phi)}\bigr)^{k_{\min}}e_{a_i}
=
\gamma e_{b_j},
\quad
\gamma=-\sin(k_{\min}\phi)\mu,
\quad
|\gamma|=1.
\end{equation}
Thus, coupling PST from $a_i$ to $b_j$ occurs at time $k_{\min}$.
\end{proof}
\begin{Theorem}
\label{thm:pst-internal}
Let
\(
G=G_1\overset{(d,d)}{\vec{\vee}}G_2
\)
be a $d$-regular directed partial join with signed coupling and coupled
transition operator $U_G^{(\phi)}$, where
\(
\phi=\frac{p\pi}{q},
\quad
\gcd(p,q)=1.
\) Fix distinct vertices $a_i,a_j\in V(G_1)$, and let
\[
e_{a_i}=
\begin{pmatrix}
e_{a_i}^{(1)}\\
0
\end{pmatrix},
\qquad
e_{a_j}=
\begin{pmatrix}
e_{a_j}^{(1)}\\
0
\end{pmatrix}.
\]
Then, for a positive integer time $k$,
\(
\bigl(U_G^{(\phi)}\bigr)^k e_{a_i}
=
\gamma e_{a_j},
\quad
|\gamma|=1,
\)
if and only if
\(
k\equiv0\pmod q
\)
and
\(
U^k e_{a_i}^{(1)}
=
\mu e_{a_j}^{(1)},
\quad
|\mu|=1.
\)
If $G$ is periodic at $a$ with minimum period
\(
\tau_{\min}=sq,
\quad
s\in\mathbb{Z}_{>0},
\)
then the possible internal-PST times within one period are
\[
\{q,2q,\ldots,(s-1)q\}.
\]

\end{Theorem}

\begin{proof}
For every $k\geq0$, by~\eqref{eq:coupled-transition}, we have
\begin{equation}
\bigl(U_G^{(\phi)}\bigr)^k e_{a_i}
=
\begin{pmatrix}
\cos(k\phi)\,U^k e_{a_i}^{(1)}\\
-\sin(k\phi)\,U^k e_{a_i}^{(1)}
\end{pmatrix}.
\end{equation}
For internal PST from ${a_i}$ to ${a_j}$, the lower block must vanish. Since
$U$ is unitary,
\(
U^k e_{a_i}^{(1)}\neq0,
\)
and hence
\(
\sin(k\phi)=0.
\)
Thus,
\(
\frac{kp}{q}\in\mathbb{Z}.
\)
Since $\gcd(p,q)=1$, this is equivalent to
\begin{equation}\label{eq_17-phi__}
\sin(k\phi)=0
\iff
k\equiv0\pmod q.
\end{equation}
For $k=jq$, we have
\(
\cos(k\phi)
=
\cos(jp\pi)
=
(-1)^{jp}.
\)
Hence, whenever \eqref{eq_17-phi__} holds,
\begin{equation}
\bigl(U_G^{(\phi)}\bigr)^k e_{a_i}
=
\begin{pmatrix}
(-1)^{jp}U^k e_{a_i}^{(1)}\\
0
\end{pmatrix}.
\end{equation}
Therefore,
\(
\bigl(U_G^{(\phi)}\bigr)^k e_{a_i}
=
\gamma e_{a_j},
\quad |\gamma|=1,
\)
if and only if
\(
U^k e_{a_i}^{(1)}
=
\mu e_{a_j}^{(1)},
\quad |\mu|=1,
\)
where
\(
\gamma=(-1)^{jp}\mu.
\)
Conversely, if $U^k e_{a_i}^{(1)}=\mu e_{a_j}^{(1)}$ with $|\mu|=1$, then
$\gamma=(-1)^{jp}\mu$ satisfies $|\gamma|=1$ and gives
\begin{equation}
\bigl(U_G^{(\phi)}\bigr)^k e_{a_i}=\gamma e_{a_j}.
\end{equation}
This proves the stated equivalence.
Now suppose that $G$ is periodic at $a_i$ with minimum
period
\(
\tau_{\min}=sq.
\)
Then
\begin{equation}
\bigl(U_G^{(\phi)}\bigr)^{k+\tau_{\min}}e_{a_i}
=
\delta\,
\bigl(U_G^{(\phi)}\bigr)^k e_{a_i},
\qquad
|\delta|=1,
\end{equation}
for all $k\geq1$. Hence, it suffices to consider
\(
1\leq k<\tau_{\min}.
\)
Since $k\equiv0\pmod q$, the possible times in this range are
\begin{equation}
q,2q,\ldots,(s-1)q.
\end{equation}
This set is empty when $s=1$, so internal PST is impossible in that
case. For $s\geq2$, the minimum internal-PST time is the smallest
time in this set for which
\(
U^k e_{a_i}^{(1)}
=
\mu e_{a_j}^{(1)},
\quad
|\mu|=1.
\)
Thus, at $k=k_{\min}$,
\(
\bigl(U_G^{(\phi)}\bigr)^{k_{\min}}e_{a_i}
=
\gamma e_{a_j},
\quad
|\gamma|=1,
\)
and internal PST from $a_i$ to $a_j$ occurs at time $k_{\min}$.
\end{proof}
\begin{cor}
\label{cor:pst-coupling-phi-pi4}
Under the assumptions of Theorem~\ref{thm:pst-coupling}, let
\(
\phi=\frac{\pi}{4},
\quad
g=l\phi,
\quad
\tau_{\min}\in\{4,8\}.
\)
For $a_i, a_j\in V(G_1)$ and $b_i, b_j\in V(G_2)$, coupling PST from $a_i$ to $b_j$
can occur only at
\(
k\in\{2,6\}.
\)
If $\tau_{\min}=4$, then $k_{\min}=2$. If $\tau_{\min}=8$, then
\(
k_{\min}\in\{2,6\},
\)
where $k_{\min}$ is the smallest $k$ for which
\[
U^k e_{a_i}^{(1)}=\mu e_{a_j}^{(1)},
\qquad |\mu|=1.
\]
\end{cor}

\begin{proof}
By Theorem~\ref{thm:pst-coupling}, coupling PST can occur only at
times satisfying
\(
k\equiv\frac{q}{2}\pmod q.
\)
For $q=4$, this becomes
\(
k\equiv2\pmod4.
\)
If $\tau_{\min}=4$, the only such time with
$1\leq k<\tau_{\min}$ is $k=2$. If $\tau_{\min}=8$, the possible times are
$k=2$ and $k=6$. The stated conclusion follows directly from the
PST condition for $U$.
\end{proof}

\begin{cor}
\label{cor:pst-internal-phi-pi4}
Under the assumptions of Theorem~\ref{thm:pst-internal}, let
\(
\phi=\frac{\pi}{4},
\quad
g=l\phi,
\quad
\tau_{\min}\in\{4,8\}.
\)
For distinct $a_i,a_j\in V(G_1)$, internal PST from $a_i$ to $a_j$ is
possible only when
\(
\tau_{\min}=8,
\)
in which case
\(
k_{\min}=4,
\)
provided that
\[
U^4e_{a_i}^{(1)}
=
\mu e_{a_j}^{(1)},
\qquad
|\mu|=1.
\]
\end{cor}
\begin{proof}
By Theorem~\ref{thm:pst-internal}, internal PST can occur only at
times satisfying
\(
k\equiv0\pmod q.
\)
For $q=4$, this gives
\(
k\equiv0\pmod4.
\)
If $\tau_{\min}=4$, there is no positive multiple of $4$ satisfying
$k<\tau_{\min}$. Hence internal PST is impossible. If $\tau_{\min}=8$,
the only possible time is $k=4$. Therefore, internal PST occurs at
$k_{\min}=4$ precisely when
\(
U^4e_{a_i}^{(1)}
=
\mu e_{a_j}^{(1)},
\quad
|\mu|=1.
\)
\end{proof}

\begin{ex}
Consider two copies of the cycle graph $C_4$,
\(
G_1=G_2=C_4,
\)
with vertex sets
\(
V(G_1)=\{a_1,a_2,a_3,a_4\},
\quad
V(G_2)=\{b_1,b_2,b_3,b_4\},
\)
and vertex ordering
\(
(a_1,a_2,a_3,a_4,b_1,b_2,b_3,b_4).
\)
We use the $2$-regular bipartite coupling
\(
a_1\sim b_2,b_4,\quad
a_2\sim b_1,b_3,\quad
a_3\sim b_2,b_4,\quad
a_4\sim b_1,b_3
\) (see Fig.~\ref{fig:C4-partial-join-case-a}).
The corresponding coupling matrix is
\[
\mathcal{J}_1=
\begin{pmatrix}
0&1&0&1\\
1&0&1&0\\
0&1&0&1\\
1&0&1&0
\end{pmatrix}
=A(C_4).
\]
Hence, the signed adjacency matrix of the partial join is
\[
A_{\pm}(G)=
\begin{pmatrix}
A(C_4)&\mathcal{J}_1\\
-\mathcal{J}_1^{\top}&A(C_4)
\end{pmatrix}.
\]
For the coupled walk, we use the states
\(
e_{a_i}=
\begin{pmatrix}
e_{a_i}^{(1)}\\
0
\end{pmatrix},
\quad
e_{b_j}=
\begin{pmatrix}
0\\
e_{a_j}^{(1)}
\end{pmatrix}.
\)
The adjacency matrix of $C_4$ admits the shunt decomposition
\(
A(C_4)=P_1+P_2,
\)
are the permutation matrices given in
Example~\ref{eq.2.2}.
Hence, the shift operator is
\(
S=P_1\otimes E_{11}+P_2\otimes E_{22}.
\)
First, consider the Hadamard coin
\(
H=
\frac{1}{\sqrt2}
\begin{pmatrix}
1&1\\
1&-1
\end{pmatrix},
\quad
C_H=I_4\otimes H,
\)
and let
\(
U_H=SC_H.
\)
A direct calculation gives
\(
U_H^4=
\begin{pmatrix}
0&I_4\\
I_4&0
\end{pmatrix}.
\)
Therefore,
\(
U_H^4e_{a_1}^{(1)}=e_{a_3}^{(1)},
\quad
U_H^4e_{a_2}^{(1)}=e_{a_4}^{(1)}.
\)
Thus, the base walk exhibits nontrivial internal PST at time $k=4$
between antipodal vertices of $C_4$, with transfer phase $\mu=1$.
Since
\(
U_{G_1}=U_{G_2}
=U_{\mathcal J_1}=U_{\mathcal J_2}
=U_H,
\)
the coupled transition operator for $\phi=\pi/4$ is
\[
U_G^{(\pi/4)}
=
\frac{1}{\sqrt2}
\begin{pmatrix}
U_H&U_H\\
-U_H&U_H
\end{pmatrix}
=
M_{\pi/4}\otimes U_H,
\]
where
\(
M_{\pi/4}
=
\frac{1}{\sqrt2}
\begin{pmatrix}
1&1\\
-1&1
\end{pmatrix}.
\)
Its spectrum is
\(
\operatorname{Spec}\bigl(U_G^{(\pi/4)}\bigr)
=
\left\{
1,-1,\pm i,
e^{\pm i\pi/4},
e^{\pm i3\pi/4}
\right\},
\)
with each eigenvalue having multiplicity $2$. Hence,
\(
\left(U_G^{(\pi/4)}\right)^8=I,
\)
and, since $e^{i\pi/4}$ has order $8$, the walk is periodic with
period $8$.
Moreover,
\(
M_{\pi/4}^4=-I_2,
\)
and therefore
\(
\left(U_G^{(\pi/4)}\right)^4
=
-I_2\otimes U_H^4.
\)
Thus,
\(
\left(U_G^{(\pi/4)}\right)^4e_{a_1}
=
-e_{a_3},
\quad
\left(U_G^{(\pi/4)}\right)^4e_{a_2}
=
-e_{a_4},
\)
and similarly for the states initially localized in $G_2$. Hence,
the full coupled walk exhibits internal PST at time $k=4$ between the
antipodal pairs
\(
a_1\leftrightarrow a_3,\quad
a_2\leftrightarrow a_4; \quad
b_1\leftrightarrow b_3,\quad
b_2\leftrightarrow b_4.
\)
Next, consider the  coin
\(
Z=
\begin{pmatrix}
1&0\\
0&-1
\end{pmatrix},
\quad
C_Z=I_4\otimes Z,
\)
and let
\(
U_Z=SC_Z.
\)
Since
\(
Z^2=I_2
\quad\text{and}\quad
P_1P_2=P_2P_1=I_4,
\)
we obtain
\(
U_Z^2=
\begin{pmatrix}
P_1^2&0\\
0&P_2^2
\end{pmatrix}.
\)
In particular,
\(
U_Ze_{a_1}^{(1)}=e_{a_4}^{(1)},
\quad
U_Z^2e_{a_1}^{(1)}=e_{a_3}^{(1)},
\quad
U_Z^3e_{a_1}^{(1)}=e_{a_2}^{(1)}.
\)
Now fix
\(
\phi=\frac{\pi}{4}.
\)
Then $q=4$, so Theorem~\ref{thm:pst-coupling} requires
\(
k\equiv\frac{q}{2}=2\pmod4.
\)
Since
\(
U_Z^2e_{a_1}^{(1)}=e_{a_3}^{(1)},
\)
the theorem gives coupling PST at $k=2$. Indeed,
\[
\left(U_G^{(Z)}\right)^2e_{a_1}
=
\begin{pmatrix}
0\\
-e_{a_3}^{(1)}
\end{pmatrix}
=
-e_{b_3}.
\]
Thus, the signed partial join exhibits coupling PST
\(
a_1\longrightarrow b_3
\)
at time $k=2$, with transfer phase
\(
\gamma=-1.
\)
Similarly,
\[
\left(U_G^{(Z)}\right)^2e_{a_2}=-e_{b_4},
\qquad
\left(U_G^{(Z)}\right)^2e_{a_3}=-e_{b_1},
\qquad
\left(U_G^{(Z)}\right)^2e_{a_4}=-e_{b_2}.
\]

This example explicitly demonstrates the coupling-PST condition in
Theorem~\ref{thm:pst-coupling} and Theorem~\ref{thm:pst-internal}. The operator $U_Z$ exhibits PST
between distinct vertices within one copy, while the factor
$M_{\pi/4}^2$ transfers the resulting state to the corresponding
vertex in the second copy.
\end{ex}

From~\eqref{transition_K_n_n}, we obtain
\(
\bigl(U_G^{(\phi)}\bigr)^2
=
-\begin{pmatrix}
U^2&0\\
0&U^2
\end{pmatrix},
\)
and hence
\begin{equation}
\bigl(U_G^{(\phi)}\bigr)^{2k}
=
(-1)^k
\begin{pmatrix}
U^{2k}&0\\
0&U^{2k}
\end{pmatrix},
\qquad
\bigl(U_G^{(\phi)}\bigr)^{2k+1}
=
(-1)^k
\begin{pmatrix}
0&U^{2k+1}\\
-U^{2k+1}&0
\end{pmatrix}.
\end{equation}
Therefore, even powers preserve the two components, while odd powers
transfer states between $G_1$ and $G_2$.
\begin{cor}
\label{thm:pst-coupling-phi-pi2} Let
\(
G=G_1\overset{(d,d)}{\vec{\vee}}G_2
\)
be a $d$-regular directed partial join with signed coupling and coupled
transition 
operator 
\[
U_G^{(\phi)}
=
\begin{pmatrix}
0&U\\
-U&0
\end{pmatrix},
\qquad
\phi=\frac{\pi}{2},
\]
where $U$ is unitary. Fix $a_i,a_j\in V(G_1)$ and $b_i,b_j\in V(G_2)$, and let
\(
e_{a_i}=
\begin{pmatrix}
e_{a_i}^{(1)}\\
0
\end{pmatrix},
\quad
e_{b_j}=
\begin{pmatrix}
0\\
e_{a_j}^{(1)}
\end{pmatrix}.
\)
Then coupling PST from $a_i$ to $b_j$ under $U_G^{(\phi)}$ occurs at
time $k$ if and only if $k$ is odd and
\(
U^k e_{a_i}^{(1)}
=
\mu e_{a_j}^{(2)},
\quad
|\mu|=1.
\)

\end{cor}

\begin{proof}
For every $k\geq0$,
\(
\bigl(U_G^{(\phi)}\bigr)^k e_{a_i}
=
\begin{pmatrix}
\cos(k\pi/2)\,U^k e_{a_i}^{(1)}\\
-\sin(k\pi/2)\,U^k e_{a_i}^{(1)}
\end{pmatrix}.
\)
For coupling PST, the upper block must vanish. Since $U$ is unitary,
$U^k e_{a_i}^{(1)}\neq0$, and therefore
\(
\cos(k\pi/2)=0
\iff
k\ \text{is odd}.
\)
For odd $k$,
\(
|\sin(k\pi/2)|=1.
\)
Thus, coupling PST occurs precisely when
\(
U^k e_{a_i}^{(1)}
=
\mu e_{a_j}^{(1)},
\quad
|\mu|=1.
\)
By Example~\ref{ex:periodicity-Knn}, the minimum period
of $G$ at $a_i$ satisfies
\(
\tau_{\min}\in\{2,4,8\}.
\)
Therefore, within one period, the possible coupling-PST times are
\begin{equation}
\begin{cases}
\{1\}, & \tau_{\min}=2,\\
\{1,3\}, & \tau_{\min}=4,\\
\{1,3,5,7\}, & \tau_{\min}=8.
\end{cases}
\end{equation}
Hence, the minimum coupling-PST time is the smallest time in the
corresponding set for which
\(
U^k e_{a_i}^{(1)}
=
\mu e_{a_j}^{(2)},
\quad
|\mu|=1.
\)

\end{proof}

\begin{cor}
\label{thm:pst-internal-phi-pi2}
Let
\(
G=G_1\overset{(d,d)}{\vec{\vee}}G_2
\)
be a $d$-regular directed partial join with signed coupling and coupled
transition operator 
\[
U_G^{(\phi)}
=
\begin{pmatrix}
0&U\\
-U&0
\end{pmatrix},
\qquad
\phi=\frac{\pi}{2},
\]
where $U$ is unitary. Fix distinct $a_i,a_j\in V(G_1)$, and let
\(
e_{a_i}=
\begin{pmatrix}
e_{a_i}^{(1)}\\
0
\end{pmatrix},
\quad
e_{a_j}=
\begin{pmatrix}
e_{a_j}^{(1)}\\
0
\end{pmatrix}.
\)
Then internal PST from $a_i$ to $a_j$ under $U_G^{(\phi)}$ occurs at
time $k$ if and only if $k$ is even and
\(
U^k e_{a_i}^{(1)}
=
\mu e_{a_j}^{(1)},
\quad
|\mu|=1.
\)

\end{cor}

\begin{proof}
For every $k\geq0$,
\(
\bigl(U_G^{(\phi)}\bigr)^k e_{a_i}
=
\begin{pmatrix}
\cos(k\pi/2)\,U^k e_{a_i}^{(1)}\\
-\sin(k\pi/2)\,U^k e_{a_i}^{(1)}
\end{pmatrix}.
\)
For internal PST, the lower block must vanish. Since $U$ is unitary,
$U^k e_{a_i}^{(1)}\neq0$, and hence
\(
\sin(k\pi/2)=0
\iff
k\ \text{is even}.
\)
For even $k$,
\(
\cos(k\pi/2)=\pm1.
\)
Therefore, internal PST occurs precisely when
\(
U^ke_{a_i}^{(1)}
=
\mu e_{a_j}^{(1)},
\quad
|\mu|=1.
\)
By Example~\ref{ex:periodicity-Knn}, the minimum period
of $G$ at $a_i$ satisfies
\(
\tau_{\min}\in\{2,4,8\}.
\)
Thus, the possible internal-PST times within one period are
\begin{equation}
\begin{cases}
\varnothing, & \tau_{\min}=2,\\
\{2\}, & \tau_{\min}=4,\\
\{2,4,6\}, & \tau_{\min}=8.
\end{cases}
\end{equation}
Hence, the minimum internal PST time is the smallest time in the
corresponding set for which
\(
U^k e_{a_i}^{(1)}
=
\mu e_{a_j}^{(1)},
\quad
|\mu|=1.
\)
\end{proof}

\section{Double Cover Conditions for Periodicity and PST in Directed Partial Joins }
\label{sec7}
Continuous-time quantum walks on double covers, and the conditions under which they admit PST, were analyzed by Coutinho and Godsil~\cite{coutinho2016perfect}. Here we consider the discrete-time analogue. The shunt-decomposition walk studied here, however, differs fundamentally from the continuous-time setting, since its transition operator acts on a space of a different dimension. Consequently, their results cannot be applied directly. Instead, we construct the
double cover explicitly in a manner that preserves PST. This construction
enables us to analyze graphs for which the transition operators $U_{G_1}$ and $U_{\mathcal{J}_1}$ do not necessarily
commute.

Let \(G_1=G_2\) be a \(d\)-regular directed graph, and let
\(\mathcal{J}_1\) and \(\mathcal{J}_2\) denote the two directed coupling
parts, where \(\mathcal{J}_1\) consists of \(d_1\) arcs from \(G_1\) to
\(G_2\), while \(\mathcal{J}_2\) consists of \(d_1\) arcs from \(G_2\) to
\(G_1\). Thus, each vertex has \(d_1\) outgoing coupling arcs in each
direction, and
\(
\mathcal{J}=\mathcal{J}_1\cup\mathcal{J}_2
\)
is a \(d_1\)-regular directed bipartite coupling graph. When
\(\mathcal{J}_1=\mathcal{J}_2\), we consider the directed partial join
\(
G=G_1\overset{(d_1,d_1)}{\vec{\vee}}G_1,
\)
whose adjacency matrix is
\begin{equation}\label{eq_46_}
A(G)
=
\begin{pmatrix}
A(G_1) & A(\mathcal{J}_1)\\
A(\mathcal{J}_1) & A(G_1)
\end{pmatrix}.
\end{equation}
Here, \(d_1\) is not required to equal the degree \(d\) of \(G_1\).
An example satisfying \(d_1=d\) is given in
Proposition~\ref{cor:identical-circulant-coupling}. Following the
notation of~\cite{coutinho2016perfect}, we may also write
\(
G=G_1\ltimes \mathcal{J}_1.
\) Suppose  that the nonzero entries of $A(G_1)$ and $A(\mathcal{J}_1)$ are disjoint, that is,
\(
A(G_1)\circ A(\mathcal{J}_1) = 0,
\)
where $\circ$ denotes the Hadamard (entrywise) product. Then $G_1\ltimes \mathcal{J}_1$ is the double
cover of the graph whose adjacency matrix is
\(
A(G_1) + A(\mathcal{J}_1).
\)
Two particular choices of $\mathcal{J}_1$ are useful in what follows.
\begin{enumerate}
\item \textit{Loopless complementary coupling}:
First, suppose that $G_1$ is loopless and take $\mathcal{J}_1=\overline{G_1}$,
where the complement is also taken without loops. Then
\(
A(\mathcal{J}_1)
=
J_n-I-A(G_1),
\)
and hence
\begin{equation}
A(G_1)+A(\mathcal{J}_1)=J_n-I=A(K_n).
\end{equation}
Therefore,
\(
G_1\ltimes\overline{G_1}
\)
is a double cover of the loopless complete graph $K_n$. This graph is
also the switching graph associated with $G_1$
\cite[Sec.~11.5]{GodsilRoyle2001}.
Moreover,
\begin{equation}
A(\mathcal{J}_1)-A(G_1)
=
J_n-I-2A(G_1)
=
\mathcal{S}(G_1),
\end{equation}
where
\(
\mathcal{S}(G_1)=J-I-2A(G_1)
\)
is the Seidel matrix of $G_1$~\cite{brouwer2012spectra}. Thus, the sum and difference of the two
adjacency matrices are
\(
A(G_1)+A(\mathcal{J}_1)=A(K_n),
\quad
A(\mathcal{J}_1)-A(G_1)=\mathcal{S}(G_1).
\)

\item \textit{Complementary coupling with loops}:
Alternatively, define $\mathcal{J}_1$ by
\(
A(\mathcal{J}_1)=J_n-A(G_1).
\)
If $G_1$ is loopless, then $\mathcal{J}_1$ contains a loop at every vertex, and
\begin{equation}\label{coupling_loops_complementary}
A(G_1)+A(\mathcal{J}_1)=J_n.
\end{equation}
Hence, in this case, $G_1\ltimes \mathcal{J}_1$ is a double cover of the complete
graph with a loop at every vertex. The corresponding difference is
\begin{equation}
A(\mathcal{J}_1)-A(G_1)
=
J_n-2A(G_1).
\end{equation}
\end{enumerate}

The adjacency matrix in~\eqref{eq_46_} can be written as
\(
A(G)
=
I_2\otimes A(G_1)
+
A(K_2)\otimes A(\mathcal{J}_1),
\)
where
\(
A(K_2)=
\begin{pmatrix}
0&1\\
1&0
\end{pmatrix}.
\)
Since $I_2$ and $A(K_2)$ commute, they are simultaneously diagonalized
by the Hadamard matrix
\(
H
=
\frac{1}{\sqrt{2}}
\begin{pmatrix}
1&1\\
1&-1
\end{pmatrix}.
\)
Consequently,
\begin{equation}
(H\otimes I)^\dagger A(G)(H\otimes I)
=
\begin{pmatrix}
A(G_1)+A(\mathcal{J}_1)&0\\
0&A(G_1)-A(\mathcal{J}_1)
\end{pmatrix}.
\end{equation}
Thus, the double-cover index separates into two Hadamard eigenspaces,
on which the sum and difference operators act independently.
Let $U_{G_1}$ and $U_{\mathcal{J}_1}$ denote the discrete-time transition
operators associated with $G_1$ and $\mathcal{J}_1$, respectively. In general,
\(
U_{G_1}U_{\mathcal{J}_1}\neq U_{\mathcal{J}_1}U_{G_1},
\)
so the spectral arguments used for commuting transition matrices do
not apply directly. The double-cover decomposition provides an
alternative way to organize the two components and motivates the
construction below for studying periodicity and PST in the non-commuting case.
\begin{prop}\label{cor_shift_block_explicit_general}
Let \(G_1=G_2\) be a \(d\)-regular directed graph on \(n\) vertices, and let
\(\mathcal J_1=\mathcal J_2\) be a \(d_1\)-regular directed coupling part.
Then
\(
G=G_1\overset{(d_1,d_1)}{\vec{\vee}}G_1
\)
has total degree \(D=d+d_1\). If
\(
S_G
=
I_2\otimes S_{G_1}
+
\sigma_x\otimes S_{\mathcal J_1},
\quad
\sigma_x=
\begin{pmatrix}0&1\\1&0\end{pmatrix},
\) where $S_{G_1}$ is the shift operator of $G_1$ and
$S_{\mathcal{J}_1}$ is the shift operator associated with the coupling
part $\mathcal{J}_1$.
Then
\[
(H\otimes I_n\otimes I_D)\,
S_G\,
(H\otimes I_n\otimes I_D)
=
\operatorname{diag}
\left(
S_{G_1}+S_{\mathcal J_1},
S_{G_1}-S_{\mathcal J_1}
\right),
\]
where \(H\) is the Hadamard matrix.

\end{prop}

\begin{proof}
 Since $H\sigma_xH=\sigma_z$ and $HI_2H=I_2$, we have 
\[
(H\otimes I\otimes I)S_G(H\otimes I\otimes I)
=I_2\otimes S_{G_1}+(H\sigma_xH)\otimes S_{\mathcal{J}_1}
=I_2\otimes S_{G_1}+\sigma_z\otimes S_{\mathcal{J}_1}
=\operatorname{diag}(S_{G_1}+S_{\mathcal{J}_1},\,S_{G_1}-S_{\mathcal{J}_1}),
\]
using $\sigma_z=\operatorname{diag}(1,-1)$ blockwise. 
 Since \(G_1\) is a \(d\)-regular directed graph and
\(\mathcal{J}=\mathcal{J}_1\cup\mathcal{J}_2\) is a \(d_1\)-regular
directed bipartite coupling graph, the graphs \(G_1\) and
\(\mathcal{J}_1\) admit shunt decompositions
\(
A(G_1)=\sum_{i=1}^{d}P_i,
\quad
A(\mathcal{J}_1)=\sum_{i=1}^{d_1}Q_i,
\)
by Lemma~\ref{lemma 2.1} and Proposition~\ref{prop:signed-coupling-shunt}, with corresponding shift operators
$S_{G_1}=\sum_{i=1}^d P_i\otimes E_{ii}$, $S_{\mathcal{J}_1}=\sum_{i=1}^{d_1} Q_i\otimes E_{ii}$
on $\mathbb C^n\otimes\mathbb C^D$. Explicitly,
\begin{equation}
S_{G_1}+S_{\mathcal{J}_1}=\sum_{i=1}^d P_i\otimes E_{ii}+\sum_{i=1}^{d_1} Q_i\otimes E_{ii},
\qquad
S_{G_1}-S_{\mathcal{J}_1}=\sum_{i=1}^d P_i\otimes E_{ii}-\sum_{i=1}^{d_1} Q_i\otimes E_{ii}.
\end{equation}
These are  the shift operators $\sum_{r=1}^{D}R_r\otimes E_{rr}$
and $\sum_{r=1}^{D}\varepsilon_rR_r\otimes E_{rr}$ from the combined shunt decomposition
$\{R_1,\dots,R_{D}\}=\{P_1,\dots,P_d,Q_1,\dots,Q_{d_1}\}$ of
$A(G_1)+A(\mathcal{J}_1)=\sum_{r=1}^{D}R_r$. With sign \[
\varepsilon_r=
\begin{cases}
+1,&1\leq r\leq d,\\
-1,&d+1\leq r\leq D.
\end{cases}
\]
Thus,
\(
\sum_{r=1}^{D}\varepsilon_rR_r
=
A(G_1)-A(\mathcal{J}_1).
\) Hence
$S_{G_1}+S_{\mathcal{J}_1}$ and $S_{G_1}-S_{\mathcal{J}_1}$ are precisely the shifts for
$A(G_1)+A(\mathcal{J}_1)$ and $A(G_1)-A(\mathcal{J}_1)$ respectively, and the claimed
identity follows.
\end{proof}
Following the analogous argument in Lemma~5.1 of~\cite{coutinho2016perfect},
we derive the relation between \(U_G\) and
\(U_{G_1+\mathcal{J}_1}\), as well as their \(k\)-th powers.
\begin{Theorem}\label{lem_double_Cover}
Let $G_1$ and \(G_2\) be $d$-regular directed graph on the vertex set
\(
V(G_1)=\{1,\ldots,n\}.
\) Let
\(
\mathcal{J}=\mathcal{J}_1\cup\mathcal{J}_2
\)
be a $d_1$-regular directed bipartite coupling graph on
\(
V(G_1)\cup V(G_2),
\)
where $\mathcal{J}_1=\mathcal{J}_2, G_1=G_2$.
Let $U_G$ denote the transition operator of $G$, and
let $U_{G_1+\mathcal{J}_1}$ and $U_{G_1-\mathcal{J}_1}$ denote the transition operators
associated with the  matrices
\(
A(G_1)+A(\mathcal{J}_1)
\quad\text{and}\quad
A(G_1)-A(\mathcal{J}_1),
\)
respectively. Then the following statements hold.
\begin{enumerate}
\item For all $k\ge0$,
\[
(H\otimes I_n)\,U_G^{\,k}\,(H\otimes I_n)
=\begin{pmatrix}U_{G_1+\mathcal{J}_1}^{\,k}&0\\0&U_{G_1-\mathcal{J}_1}^{\,k}\end{pmatrix},
\qquad
H=\frac{1}{\sqrt2}\begin{pmatrix}1&1\\1&-1\end{pmatrix}.
\]
\item If there exist vertices $a_i,a_j\in V(G_1)$, a time $k\ge0$, and a
unit-modulus scalar $\gamma$ such that
\[
U_{G_1+\mathcal{J}_1}^{\,k}e_{a_i}^{(1)}=\gamma e_{a_j}^{(1)}
\qquad\text{and}\qquad
U_{G_1-\mathcal{J}_1}^{\,k}e_{a_i}^{(1)}=\gamma e_{a_j}^{(1)}.
\]
Then $U_G^ke_{a_i}=\gamma e_{a_j} $, where
    \(
    e_{a_i}=\begin{pmatrix}e_{a_i}^{(1)}\\0\end{pmatrix}
    \quad
    e_{a_j}=\begin{pmatrix}e_{a_j}^{(1)}\\0\end{pmatrix}
    \), that is, G exhibits internal PST from \(a_i\) to \(a_j\)
    at time $k$, with transfer phase $\gamma$.
\end{enumerate}
\end{Theorem}

\begin{proof}
\textbf{(1)}
We have
\begin{equation}
A(G)=\begin{pmatrix}A(G_1)&A(\mathcal{J}_1)\\A(\mathcal{J}_1)&A(G_1)\end{pmatrix}
=I_2\otimes A(G_1)+\begin{pmatrix}0&1\\1&0\end{pmatrix}\otimes A(\mathcal{J}_1).
\end{equation}
Since $H^2=I_2$ diagonalizes $\begin{pmatrix}0&1\\1&0\end{pmatrix}$ as
$H\begin{pmatrix}0&1\\1&0\end{pmatrix}H=\begin{pmatrix}1&0\\0&-1\end{pmatrix}$,
and $HI_2H=I_2$, conjugation gives
\begin{equation}
(H\otimes I_n)\,A(G)\,(H\otimes I_n)
=\begin{pmatrix}A(G_1)+A(\mathcal{J}_1)&0\\0&A(G_1)-A(\mathcal{J}_1)\end{pmatrix}.
\end{equation}
By Proposition~\ref{cor_shift_block_explicit_general}, the shift $S_G$ itself satisfies $S_G=I_2\otimes S_{G_1}+\sigma_x\otimes S_{\mathcal{J}_1}$,
so the same conjugation gives
$(H\otimes I_n)S_G(H\otimes I_n)=\operatorname{diag}(S_{G_1}+S_{\mathcal{J}_1},\,S_{G_1}-S_{\mathcal{J}_1})$,
where $S_{G_1}+S_{\mathcal{J}_1}$ and $S_{G_1}-S_{\mathcal{J}_1}$ are precisely the shifts
of $U_{G_1+\mathcal{J}_1}$ and $U_{G_1-\mathcal{J}_1}$. Hence conjugating $U_G$ itself by
$H\otimes I_n$ gives
\begin{equation}\label{eq.39}
(H\otimes I_n)\,U_G\,(H\otimes I_n)
=\begin{pmatrix}U_{G_1+\mathcal{J}_1}&0\\0&U_{G_1-\mathcal{J}_1}\end{pmatrix}.
\end{equation}
Since $H\otimes I_n$ is an involution, $(H\otimes I_n)^2=I_{2n}$, so
raising both sides of~\eqref{eq.39} to the $k$-th power gives
\begin{equation}\label{eq.40}
(H\otimes I_n)\,U_G^{\,k}\,(H\otimes I_n)
=\begin{pmatrix}U_{G_1+\mathcal{J}_1}^{\,k}&0\\0&U_{G_1-\mathcal{J}_1}^{\,k}\end{pmatrix}.
\end{equation}
Conjugating~\eqref{eq.40} by $H\otimes I_n$ and using $(H\otimes I_n)^2=I_{2n}$
yields
\begin{equation}\label{eq.41}
U_G^{\,k}
=\frac12\begin{pmatrix}
U_{G_1+\mathcal{J}_1}^{\,k}+U_{G_1-\mathcal{J}_1}^{\,k} & U_{G_1+\mathcal{J}_1}^{\,k}-U_{G_1-\mathcal{J}_1}^{\,k}\\[4pt]
U_{G_1+\mathcal{J}_1}^{\,k}-U_{G_1-\mathcal{J}_1}^{\,k} & U_{G_1+\mathcal{J}_1}^{\,k}+U_{G_1-\mathcal{J}_1}^{\,k}
\end{pmatrix}.
\end{equation}
\textbf{(2)}
Applying the operator $U_G^k$ in~\eqref{eq.41} to the initial state
$(e_{a_i}^{(1)},0)$, we obtain
\begin{equation}
U_G^{\,k}\begin{pmatrix}e_{a_i}^{(1)}\\0\end{pmatrix}
=\frac12\begin{pmatrix}
U_{G_1+\mathcal{J}_1}^{\,k}e_{a_i}^{(1)}+U_{G_1-\mathcal{J}_1}^{\,k}e_{a_i}^{(1)}\\[4pt]
U_{G_1+\mathcal{J}_1}^{\,k}e_{a_i}^{(1)}-U_{G_1-\mathcal{J}_1}^{\,k}e_{a_i}^{(1)}
\end{pmatrix}.
\end{equation}
By hypothesis $U_{G_1+\mathcal{J}_1}^{\,k} e_{a_i}^{(1)}=U_{G_1-\mathcal{J}_1}^{\,k}e_{a_i}^{(1)}=\gamma e_{a_j}^{(1)},$ so
\begin{equation}
U_G^{\,k}\begin{pmatrix}e_{a_i}^{(1)}\\0\end{pmatrix}
=\frac{1}{2}\begin{pmatrix}\gamma e_{a_j}^{(1)}+\gamma e_{a_j}^{(1)}\\ \gamma e_{a_j}^{(1)}-\gamma e_{a_j}^{(1)}\end{pmatrix}
=\begin{pmatrix}\gamma e_{a_j}^{(1)}\\0\end{pmatrix}.
\end{equation}
Since $|\gamma|=1$, this shows $U_G^{\,k}$ maps $(e_{a_i}^{(1)},0)$ to
$\gamma\,(e_{a_j}^{(1)},0)$. Thus, PST from $a_i$ to
$a_j$ occurs at time $k$.
\end{proof}

We now establish the relationship between the transition operators
$U_{G_1+\mathcal{J}_1}$ and $U_{G_1-\mathcal{J}_1}$,
 showing that $U_{G_1-\mathcal{J}_1}$ can be obtained from
$U_{G_1+\mathcal{J}_1}$ by conjugation with a fixed diagonal sign matrix
acting on the coin space. This conjugation preserves both PST and periodicity.

\begin{lemma}\label{PST_Double_Cover}
Let \(G_1=G_2\) be a \(d\)-regular directed graph on
\(V(G_1)=\{1,\ldots,n\}\), and let
\(\mathcal{J}=\mathcal{J}_1\cup\mathcal{J}_2\) be a
\(d_1\)-regular directed bipartite coupling graph, with
\(
\mathcal{J}_1=\mathcal{J}_2.
\)
Let
\(
G=G_1\overset{(d_1,d_1)}{\vec{\vee}}G_1.
\)
Then \(G_1+\mathcal{J}_1\) is \((D=d+d_1)\)-regular.
Let \(U_{G_1+\mathcal{J}_1}\) denote the transition operator associated
with the \(D\)-shunt decomposition of \(G_1+\mathcal{J}_1\). Then the
signed matrix
\(
A(G_1-\mathcal{J}_1)
=
A(G_1)-A(\mathcal{J}_1)
\)
admits the transition operator
\(
{
U_{G_1-\mathcal{J}_1}
=
(I_n\otimes\mathcal E)\,
U_{G_1+\mathcal{J}_1}\,
(I_n\otimes\mathcal E),
}
\) where \(
\mathcal E=
\operatorname{diag}
(\underbrace{1,\ldots,1}_{d},
\underbrace{-1,\ldots,-1}_{d_1})
\in\mathbb C^{D\times D}.
\)
Moreover, the following statements hold.

\begin{enumerate}
\item
Suppose that
\(
U_{G_1+\mathcal{J}_1}^{\,k}
\left(e_{a_i}\otimes\phi_0\right)
=
\gamma
\left(e_{a_j}\otimes\phi_1\right),
\quad
|\gamma|=1,
\)
for unit coin states
\(\phi_0,\phi_1\in\mathbb C^{D}\). Then
\[
U_{G_1-\mathcal{J}_1}^{\,k}
\left(e_{a_i}\otimes\mathcal E\phi_0\right)
=
\gamma
\left(e_{a_j}\otimes\mathcal E\phi_1\right).
\]

\item
If \(G_1+\mathcal{J}_1\) is periodic at time \(\tau\), namely
\(
U_{G_1+\mathcal{J}_1}^{\,\tau}
=
\eta I_n\otimes I_D,
\quad |\eta|=1,
\)
then \(U_{G_1-\mathcal{J}_1}\) is also periodic at the same time
\(\tau\) with the same phase \(\eta\).
\end{enumerate}
\end{lemma}

\begin{proof}
Suppose that $G_1$ is a $d$-regular directed graph and
$\mathcal{J}=\mathcal{J}_1\cup\mathcal{J}_2$ is a $d_1$-regular directed
coupling graph. Then by
Lemma~\ref{lemma 2.1} and Proposition~\ref{prop:signed-coupling-shunt}, $G_1$ and the coupling part $\mathcal{J}_1$ admit
shunt decompositions
\(
A(G_1)=\sum_{i=1}^d P_i,
\quad
A(\mathcal{J}_1)=\sum_{i=1}^{d_1} Q_i,
\)
where $P_i$ and $Q_i$ are permutation matrices.
Hence,
\[
A(G_1)+A(\mathcal{J}_1)
=
\sum_{i=1}^d P_i+\sum_{i=1}^{d_1} Q_i
\]
is the adjacency matrix of a $D$-regular directed graph \(G_1+J_1\). Therefore,
the $D$ permutation matrices
\(
P_1,\ldots,P_d,Q_1,\ldots,Q_{d_1}
\)
form a shunt decomposition of $G_1+\mathcal{J}_1$. On the other hand,
\[
A(G_1)-A(\mathcal{J}_1)
=
\sum_{i=1}^d P_i-\sum_{i=1}^{d_1} Q_i
\]
is the corresponding signed matrix, obtained by assigning positive
signs to the $P_i$-shunts and negative signs to the $Q_i$-shunts.
 Let
\(
S_{G_1+\mathcal{J}_1}
=
\sum_{i=1}^dP_i\otimes E_{ii}
+
\sum_{i=1}^{d_1}Q_i\otimes E_{d_1+i,d_1+i}
\)
be the shift operator associated with the combined $D$-shunt
decomposition. With
\(
\mathcal E
=
\operatorname{diag}
(\underbrace{1,\ldots,1}_{d},
\underbrace{-1,\ldots,-1}_{d_1}),
\)
we have
\(
\mathcal E E_{ii}=E_{ii},
\quad
\mathcal E E_{d_1+i,d_1+i}=-E_{d_1+i,d_1+i}.
\)
Therefore,
\begin{equation}\label{63}
(I_n\otimes\mathcal E)S_{G_1+\mathcal{J}_1}
=
\sum_{i=1}^dP_i\otimes E_{ii}
-
\sum_{i=1}^{d_1}Q_i\otimes E_{d+i,d+i},
=S_{G_1-\mathcal{J}_1}\end{equation}
which is the signed shift corresponding to
$A(G_1)-A(\mathcal{J}_1)$. Let the transition operator for $G_1+\mathcal{J}_1$ be
\begin{equation}\label{64}
U_{G_1+\mathcal{J}_1}
=
S_{G_1+\mathcal{J}_1}(I_n\otimes C),
\end{equation}
where $C$ is a unitary coin operator. Using
\eqref{63} and \eqref{64}, we obtain
\[
\begin{aligned}
U_{G_1-\mathcal{J}_1}=S_{G_1-\mathcal{J}_1}(I_n\otimes C\mathcal E)
&=
(I_n\otimes\mathcal E)
S_{G_1+\mathcal{J}_1}(I_n\otimes C)
(I_n\otimes\mathcal E)
=(I_n\otimes\mathcal E)U_{G_1+\mathcal{J}_1}
\bigl(I_n\otimes \mathcal E\bigr).
\end{aligned}
\]
Since $\mathcal E$ and $C$ are unitary, $C\mathcal E$ is unitary.
Thus $U_{G_1-\mathcal{J}_1}$ is a valid transition operator associated with the
signed matrix $A(G_1)-A(\mathcal{J}_1)$.
Since
\(
\mathcal E^\dagger=\mathcal E,
\quad
\mathcal E^2=I_{D},
\)
we have, for every $t\geq0$,
\begin{equation}\label{eq.74}
U_{G_1-\mathcal{J}_1}^{\,t}
=
(I_n\otimes\mathcal E)
U_{G_1+\mathcal{J}_1}^{\,t}
(I_n\otimes\mathcal E).
\end{equation}
Suppose
\(
U_{G_1+\mathcal{J}_1}^{\,k}
(e_{a_i}^{(1)}\otimes\phi_0)
=
\gamma(e_{a_j}^{(1)}\otimes\phi_1).
\)
Then by~\eqref{eq.74}, we have
\begin{equation}
U_{G_1-\mathcal{J}_1}^{\,k}
(e_{a_i}^{(1)}\otimes\mathcal E\phi_0)
=
(I_n\otimes\mathcal E)
U_{G_1+\mathcal{J}_1}^{\,k}
(e_{a_i}^{(1)}\otimes\phi_0)
=
\gamma(I_n\otimes\mathcal E)
(e_{a_j}^{(1)}\otimes\phi_1)
=
\gamma (e_{a_j}^{(1)}\otimes\mathcal E\phi_1).
\end{equation}
Since $\mathcal E$ is unitary, $\mathcal E\phi_0$ and
$\mathcal E\phi_1$ are unit coin states. Hence PST is preserved at the
same time $k$.
Finally, if
\(
U_{G_1+\mathcal{J}_1}^{\,\tau}
=
\eta I_n\otimes I_{D},
\quad |\eta|=1,
\)
then
\begin{equation}
U_{G_1-\mathcal{J}_1}^{\,\tau}
=
(I_n\otimes\mathcal E)
U_{G_1+\mathcal{J}_1}^{\,\tau}
(I_n\otimes\mathcal E)
=
\eta(I_n\otimes\mathcal E^2)
=
\eta I_n\otimes I_{D}.
\end{equation}
Thus, $U_{G_1-\mathcal{J}_1}$ is periodic at the same time $\tau$ with the same
phase $\eta$.
\end{proof}
We now derive the PST condition for the double cover of the complete graph with a loop at every vertex, arising from the complementary coupling with loops defined in~\eqref{coupling_loops_complementary}, in the case \(d_1=d\).

\begin{cor}\label{cor_U_Uprime_relation} 
Let 
\(
U=U_{G_1'+\mathcal{J}_1'}, 
\quad 
U'=U_{G_1'-\mathcal{J}_1'} 
= 
(I_n\otimes\mathcal E)\,
U\,
(I_n\otimes\mathcal E),
\)
be as in Lemma~\ref{PST_Double_Cover} at the $n$-vertex level. Suppose that $G_1'$ and $\mathcal{J}_1'$ admit
$d$-shunt decompositions with disjoint arc and
\(
A(G_1')+A(\mathcal{J}_1')=A(K_n^{\circlearrowleft}).
\)
Let $U_{J_n}$ and $U_{G_1-\mathcal{J}_1}$ be the corresponding $2n$-level
transition operators
\[
U_{J_n}
=
\frac{1}{\sqrt{2}}
\begin{pmatrix}
U&U\\
-U&U
\end{pmatrix},
\qquad
U_{G_1-\mathcal{J}_1}
=
\frac{1}{\sqrt{2}}
\begin{pmatrix}
U'&U'\\
-U'&U'
\end{pmatrix}.
\]
Then the following conditions hold.
\begin{enumerate}
    \item[(a)] The transition operators $U$ and $U'$ are related by
    \(
    U'
    =
    (I_n\otimes\mathcal{E})\,
    U\,
    (I_n\otimes\mathcal{E}).
    \)

    \item[(b)] The corresponding $2n$-vertex transition operators
    satisfy
    \(
    U_{G_1-\mathcal{J}_1}
    =
    (I_2\otimes I_n\otimes\mathcal{E})\,
    U_{J_n}\,
    (I_2\otimes I_n\otimes\mathcal{E}).
    \)

    \item[(c)]
    If $U^k(e_{a_i}^{(1)}\otimes\phi_0)=\gamma(e_{a_j}^{(1)}\otimes\phi_1) \quad |\gamma|=1$. Then \(U'^k(e_{a_i}^{(1)}\otimes\mathcal E\phi_0)=\gamma (e_{a_j}^{(1)}\otimes\mathcal E\phi_1) \). Moreover, if $U$ is periodic at time $\tau$,
    then $U'$ is periodic at the same time $\tau$.

    \item[(d)]
    Suppose that $\phi_0$ and $\phi_1$ are eigenvectors of $\mathcal E$,
    with
    \(
    \mathcal E\phi_0=\varepsilon_0\phi_0,
    \quad
    \mathcal E\phi_1=\varepsilon_1\phi_1,
    \quad
    \varepsilon_0,\varepsilon_1\in\{+1,-1\}.
    \)
    If $U_{J_n}^{\,k}(e_{a_i}^{(1)}\otimes\phi_0,0)^\top=\gamma\,(e_{a_j}^{(1)}\otimes\phi_1,0)^\top, \quad |\gamma|=1,$ then $U_{G_1-\mathcal{J}_1}^{\,k}(e_{a_i}^{(1)}\otimes\phi_0,0)^\top=\varepsilon_0\varepsilon_1\gamma\,(e_{a_j}^{(1)}\otimes\phi_1,0)^\top$. Thus, the PST phase for $U_{G_1-\mathcal{J}_1}$ is
    $\varepsilon_0\varepsilon_1\gamma$. In particular, the phase
    remains $\gamma$ when $\varepsilon_0=\varepsilon_1$, and changes to
    $-\gamma$ when $\varepsilon_0\neq\varepsilon_1$ with phase
    \(
    \varepsilon_0\varepsilon_1\gamma.
    \)
    In particular, the phase is $\gamma$ when
    $\varepsilon_0=\varepsilon_1$, and $-\gamma$ when
    $\varepsilon_0\neq\varepsilon_1$.
\end{enumerate}
\end{cor}
\begin{proof}
Since
\(
A(K_n^{\circlearrowleft})=J_n
\)
is $n$-regular, while $G_1'$ and $\mathcal{J}_1'$ are each $d$-regular, their
disjoint union is $2d$-regular. Hence
\(
n=2d,
\)
and, in particular, $n$ is even.

\textbf{(a)}
This follows directly from Lemma~\ref{PST_Double_Cover} applied to
the $n$-vertex operators
\(
U=U_{G_1'+\mathcal{J}_1'},
\quad
U'=U_{G_1'-\mathcal{J}_1'},
\)
using the fact that $G_1'$ and $\mathcal{J}_1'$ are $d$-regular with disjoint
arc sets and
\(
A(G_1')+A(\mathcal{J}_1')=A(K_n^{\circlearrowleft}).
\)

\textbf{(b)}
Since $\mathcal E$ acts only on the coin space, the same conjugation $I_2\otimes I_n\otimes\mathcal E$ is applied to each block of 
of $U_J$ . Hence,
\begin{equation}
(I_2\otimes I_n\otimes\mathcal E)\,
U_{J_n}\,
(I_2\otimes I_n\otimes\mathcal E)=
\frac{1}{\sqrt{2}}
\begin{pmatrix}
(I_n\otimes\mathcal E)U(I_n\otimes\mathcal E)
&
(I_n\otimes\mathcal E)U(I_n\otimes\mathcal E)
\\[2mm]
-(I_n\otimes\mathcal E)U(I_n\otimes\mathcal E)
&
(I_n\otimes\mathcal E)U(I_n\otimes\mathcal E)
\end{pmatrix}.
\end{equation}
By part~(a), each block is equal to $U'$. Therefore,
\begin{equation}
(I_2\otimes I_n\otimes\mathcal E)\,
U_{J_n}\,
(I_2\otimes I_n\otimes\mathcal E)
=
U_{G_1-\mathcal{J}_1}.
\end{equation}

\textbf{(c)}
This follows directly from Lemma~\ref{PST_Double_Cover}(2), applied to
the transition operators
\(
U=U_{G_1'+\mathcal{J}_1'}
\quad\text{and}\quad
U'=U_{G_1'-\mathcal{J}_1'}.
\)

\textbf{(d)}
Since $(I_2\otimes I_P\otimes\mathcal E)^2=I_2\otimes I_P\otimes\mathcal E^2=I$
(as $\mathcal E^2=I_C$), induction on $k$ using part (b) gives
\begin{equation}\label{eq.100}
U_{G_1-\mathcal{J}_1}^{\,k}=(I_2\otimes I_P\otimes\mathcal E)\,U_{J_n}^{\,k}\,(I_2\otimes I_P\otimes\mathcal E)
\quad\text{for all }k\ge0.
\end{equation}
Suppose $U_{J_n}^{\,k}(e_{a_i}^{(1)}\otimes\phi_0,0)^\top=\gamma\,(e_{a_j}^{(1)}\otimes\phi_1,0)^\top$.
Using $(I_2\otimes I_P\otimes\mathcal E)(x,0)^\top=((I_P\otimes\mathcal E)x,0)^\top$, we have by~\eqref{eq.100}
\begin{equation}
U_{G_1-\mathcal{J}_1}^{\,k}\big((I_2\otimes I_P\otimes\mathcal E)(e_{a_i}^{(1)}\otimes\phi_0,0)^\top\big)
=(I_2\otimes I_P\otimes\mathcal E)\,\gamma\,(e_{a_j}^{(1)}\otimes\phi_1,0)^\top.
\end{equation} Thus,
\( U_{G_1-\mathcal{J}_1}^{\,k}(e_{a_i}^{(1)}\otimes\mathcal E\phi_0,0)^\top=\gamma\,(e_{a_j}^{(1)}\otimes\mathcal E\phi_1,0)^\top.
\)
Substituting $\mathcal E\phi_0=\varepsilon_0\phi_0$, $\mathcal E\phi_1=\varepsilon_1\phi_1$, and using linearity, we have
\begin{equation}
\varepsilon_0\,U_{G_1-\mathcal{J}_1}^{\,k}(e_{a_i}^{(1)}\otimes\phi_0,0)^\top=\varepsilon_1\gamma\,(e_{a_j}^{(1)}\otimes\phi_1,0)^\top,
\end{equation}
so, since $\varepsilon_0\in\{\pm1\}$ is invertible with $\varepsilon_0^{-1}=\varepsilon_0$, we obtain
\begin{equation}
U_{G_1-\mathcal{J}_1}^{\,k}(e_{a_i}^{(1)}\otimes\phi_0,0)^\top=\varepsilon_0\varepsilon_1\gamma\,(e_{a_j}^{(1)}\otimes\phi_1,0)^\top.
\end{equation}
Hence, PST occurs at the same time $k$, with phase
\(
\varepsilon_0\varepsilon_1\gamma.
\)
If $\varepsilon_0=\varepsilon_1$, then
$\varepsilon_0\varepsilon_1=1$, so the phase remains $\gamma$.
If $\varepsilon_0\neq\varepsilon_1$, then
$\varepsilon_0\varepsilon_1=-1$, so the phase becomes $-\gamma$.
\end{proof}
\begin{cor}\label{cor_combine_thm_and_switching_phi}
Fix $\phi\in\mathbb R$ with
\(
\frac{\phi}{\pi}=\frac{p}{q},
\quad
\gcd(p,q)=1,
\)
and let $U,U',U_{J_n},U_{G_1-\mathcal{J}_1}$ be the $\phi$-analogues of the
operators in Corollary~\ref{cor_U_Uprime_relation}. Let the minimum period of $G$ be
\[
\tau_{\min}
=
\operatorname{lcm}\!\left(
q,\frac{2v}{\gcd(lu,2v)}
\right),
\qquad
s=\frac{\tau_{\min}}{q}\in\mathbb Z_{>0},
\]
as in Lemma~\ref{lem:periodicity-coupled-walk}.
Suppose   $U^k(e_{a_i}^{(1)}\otimes\phi_0)=\gamma(e_{a_j}^{(1)}\otimes\phi_1) \quad |\gamma|=1$,
where $\phi_0,\phi_1$ are eigenvectors of $\mathcal E$ satisfying
\(
\mathcal E\phi_0=\varepsilon_0\phi_0,
\quad
\mathcal E\phi_1=\varepsilon_1\phi_1,
\quad
\varepsilon_0,\varepsilon_1\in\{+1,-1\}.
\)
Let $U_{G_1\ltimes \mathcal{J}_1}$ denote the corresponding $4n$-level
double-cover transition operator.  Then $U_{G_1\ltimes \mathcal{J}_1}$ exhibits PST at time $k_{\min}$, where the
transfer between the two graphs \(G_1\) and \(\mathcal{J}_1\) is determined by the eigenvalues
$\varepsilon_0$ and $\varepsilon_1$ of $\mathcal E$ as follows:

\begin{enumerate}
\item[(1)]
If $\varepsilon_0=\varepsilon_1$, then
\(
U_{G_1\ltimes \mathcal{J}_1}^{\,k_{\min}}
\begin{pmatrix}
(e_{a_i}^{(1)}\otimes\phi_0,0)^\top\\
0
\end{pmatrix}
=
\begin{pmatrix}
\gamma(e_{a_j}^{(1)}\otimes\phi_1,0)^\top\\
0
\end{pmatrix},
\quad
|\gamma|=1.
\)

\item[(2)]
If $\varepsilon_0\neq\varepsilon_1$, then
\(
U_{G_1\ltimes \mathcal{J}_1}^{\,k_{\min}}
\begin{pmatrix}
(e_{a_i}^{(1)}\otimes\phi_0,0)^\top\\
0
\end{pmatrix}
=
\begin{pmatrix}
0\\
\gamma(e_{a_j}^{(1)}\otimes\phi_1,0)^\top
\end{pmatrix},
\quad
|\gamma|=1.
\)
\end{enumerate}

\end{cor}

\begin{proof}
Let
\(
e_{a_i}^{(1)'}=e_{a_i}^{(1)}\otimes\phi_0,
\quad
e_{a_j}^{(1)'}=e_{a_j}^{(1)}\otimes\phi_1.
\) Since $U$ exhibits PST at time $k$,
by Theorem~\ref{thm:pst-internal},  the corresponding
\(\phi\)-coupled operator satisfies
\begin{equation}\label{dash_J_}
U_{J_n}^{k_{\min}}
\begin{pmatrix}
e_{a_i}^{(1)'}\\
0
\end{pmatrix}
=
\gamma
\begin{pmatrix}
e_{a_j}^{(1)'}\\
0
\end{pmatrix},
\qquad
|\gamma|=1
\end{equation} where
\(
k\equiv0\pmod q
\)
is possible only when $s\geq2$, and the possible times within one
period are
\(
q,2q,\ldots,(s-1)q.
\)
Hence the minimum PST time satisfies
\(
k_{\min}\in\{q,2q,\ldots,(s-1)q\}.
\)
By the $\phi$-analogue of
Corollary~\ref{cor_U_Uprime_relation}(d),
\begin{equation}\label{dash_2_J}
U_{G_1-\mathcal{J}_1}^{\,k_{\min}}
\begin{pmatrix}
e_{a_i}^{(1)'}\\
0
\end{pmatrix}
=
\varepsilon_0\varepsilon_1\gamma
\begin{pmatrix}
e_{a_j}^{(1)'}\\
0
\end{pmatrix}.
\end{equation}
By the double-cover construction in
Lemma~\ref{lem_double_Cover}, the $4n$-level transition operator
satisfies
\begin{equation}
U_{G_1\ltimes \mathcal{J}_1}^{\,k_{}}
\begin{pmatrix}
(e_{a_i}^{(1)'},0)^\top\\
0
\end{pmatrix}
=
\frac12
\begin{pmatrix}
U_{J_n}^k(e_{a_i}^{(1)'},0)^\top
+
U_{G_1-\mathcal{J}_1}^k(e_{a_i}^{(1)'},0)^\top
\\[2mm]
U_{J_n}^k(e_{a_i}^{(1)'},0)^\top
-
U_{G_1-\mathcal{J}_1}^k(e_{a_i}^{(1)'},0)^\top
\end{pmatrix}.
\end{equation}
Setting \(k=k_{\min}\) and using~\eqref{dash_J_} and~\eqref{dash_2_J}, we obtain
\begin{equation}
U_{G_1\ltimes \mathcal{J}_1}^{\,k_{\min}}
\begin{pmatrix}
(e_{a_i}^{(1)'},0)^\top\\
0
\end{pmatrix}
=
\frac12
\begin{pmatrix}
(1+\varepsilon_0\varepsilon_1)
\gamma(e_{a_j}^{(1)'},0)^\top
\\[2mm]
(1-\varepsilon_0\varepsilon_1)
\gamma(e_{a_j}^{(1)'},0)^\top
\end{pmatrix}.
\end{equation}
If $\varepsilon_0=\varepsilon_1$, then
\(
\varepsilon_0\varepsilon_1=1,
\)
and therefore
\begin{equation}
U_{G_1\ltimes \mathcal{J}_1}^{\,k_{\min}}
\begin{pmatrix}
(e_{a_i}^{(1)'},0)^\top\\
0
\end{pmatrix}
=
\begin{pmatrix}
\gamma(e_{a_j}^{(1)'},0)^\top\\
0
\end{pmatrix}.
\end{equation}
This gives case~(1).
If $\varepsilon_0\neq\varepsilon_1$, then
\(
\varepsilon_0\varepsilon_1=-1,
\)
and hence
\begin{equation}
U_{G_1\ltimes \mathcal{J}_1}^{\,k_{\min}}
\begin{pmatrix}
(e_{a_i}^{(1)'},0)^\top\\
0
\end{pmatrix}
=
\begin{pmatrix}
0\\
\gamma(e_{a_j}^{(1)'},0)^\top
\end{pmatrix}.
\end{equation}
This gives case~(2). Thus, in both cases, the $4n$-level operator
exhibits PST at time $k_{\min}$.

It remains to determine a period of the $4n$-level operator. By the
definition of $\tau_{\min}$, the $2n$-level operator $U_{J_n}$ satisfies
\(
U_{J_n}^{\tau_{\min}}=\eta'I_{2n}
\)
for some $|\eta'|=1$. Indeed, the eigenvalues of $U_{J_n}$ are of the
form
\(
e^{i(\theta_r\pm\phi)},
\)
where $e^{i\theta_r}$ are the eigenvalues of $U$ by Proposition~\ref{prop:general-angle}. Since
$\tau_{\min}$ is a period of the underlying coupled walk,
\(
e^{i\tau_{\min}\theta_r}
\)
is constant across $r$, and since
\(
\tau_{\min}=qs,
\)
we have
\(
\tau_{\min}\phi=sp\pi\in\pi\mathbb Z.
\)
Therefore,
\(
e^{\pm i\tau_{\min}\phi}
=
(-1)^{sp},
\)
so all eigenvalues of $U_J^{\tau_{\min}}$ are equal to the same
unimodular constant. Hence,
\[
e^{i\tau_{\min}(\theta_r\pm\phi)} = e^{i\tau_{\min}\theta_r}\,e^{\pm i\tau_{\min}\phi} = \eta\,(-1)^{sp} = \eta'
\]
is a single unimodular constant, independent of both $r$ and the choice of sign, giving
$U_{J_n}^{\,\tau_{\min}}=\eta' I_{2n}$.
By the $\phi$-analogue of Corollary~\ref{cor_U_Uprime_relation}(c), the
same argument applied to $U'$ gives $U_{G_1-\mathcal{J}_1}^{\,\tau_{\min}}=\eta' I_{2n}$ as well,
independent of the sign match between $\varepsilon_0,\varepsilon_1$. Substituting these into
Lemma~\ref{lem_double_Cover}(1) at $k=\tau_{\min}$ gives $U_{G_1\ltimes \mathcal{J}_1}^{\,\tau_{\min}}=\eta' I_{4n}$,
so ${G_1\ltimes \mathcal{J}_1}$ is periodic at $\tau_{\min}$ in both cases. 

\end{proof}

\begin{remark}[Special case $\phi=\pi/4$]
At $\phi=\pi/4$ ($p=1,q=4$, so $\tau_{\min}=8=qs$ with $s=2$), Corollary~\ref{cor_combine_thm_and_switching_phi}
specializes to; if $U^k(e_{a_i}^{(1)}\otimes\phi_0)=\gamma(e_{a_j}^{(1)}\otimes\phi_1), |\gamma|=1$  at some time
$k\equiv0\pmod4$, then ${G_1\ltimes \mathcal{J}_1}$ exhibits PST at time $k=4$ with minimal period
$\tau_{\min}=8$, within \(G_1\) or \(\mathcal{J}_1\) according to whether $\varepsilon_0=\varepsilon_1$. Here $sp=2$ is even, so the sign factor $(-1)^{sp}=1$ and $\eta'=\eta$ by
Proposition~\ref{prop:general-angle}, $U_J$'s eigenvalues are $\theta_r\pm\pi/4$ for $U$'s eigenvalues
$\theta_r$, so
\(
e^{i\cdot8(\theta_r\pm\pi/4)} = e^{i8\theta_r}\,e^{\pm i2\pi} = e^{i8\theta_r} = \eta
\)
for every $r$, giving $U_{J_n}^8=\eta I_{2n}$. By Corollary~\ref{cor_U_Uprime_relation}(c), the
same argument applied to $U'$ gives $U_{G_1-\mathcal{J}_1}^8=\eta I_{2n}$, independent of the sign match
between $\varepsilon_0,\varepsilon_1$. Substituting into Lemma~\ref{lem_double_Cover}(1) at
$k=8$ gives $U_{G_1\ltimes \mathcal{J}_1}^8=\eta I_{4n}$, confirming ${G_1\ltimes \mathcal{J}_1}$ is periodic at
$\tau_{min}=8$ in both cases.
\end{remark}
\begin{remark}
The two cases of Corollary~\ref{cor_combine_thm_and_switching_phi} give
qualitatively different types of PST, depending on whether
$\varepsilon_0$ and $\varepsilon_1$ are equal. In \textit{Case (1)},
where $\varepsilon_0=\varepsilon_1$, internal PST occurs. In
\textit{Case (2)}, where $\varepsilon_0\neq\varepsilon_1$, coupling PST
occurs. In both cases, the transfer is perfect and occurs at the same
time $k_{\min}$, with the same minimum period $\tau_{\min}$.
\end{remark}
\section{PST in the Partial Join of the Complete Graph \(K_n\), \(n=2^m\), \(m\geq2\)}\label{sec8}

 In this section, we show that although \(K_n\), \(n\geq4\), does not
exhibit PST for the shunt walk considered below, a
suitable partial join can exhibit PST through the double-cover
construction.

Let \(n\geq4\) be even. Then \(K_n\) is \((n-1)\)-regular. By the
\(1\)-factorization of \(K_n\), its adjacency matrix admits the
decomposition
\(
A(K_n)=P_1+\cdots+P_{n-1},
\)
where each \(P_r\) is the permutation matrix of a perfect matching.
Consequently,
\(
P_r^T=P_r,\quad P_r^{-1}=P_r,\quad P_r^2=I_n,
\)
and the corresponding shift operator
\(
S=\sum_{r=1}^{n-1}P_r\otimes E_{rr}
\)
satisfies
\(
S^2=I.
\)
Thus, the shift is an involution, as in the Grover-walk framework
\cite{KubotaSegawa2022}. The corresponding transition operator is
\(
U_{K_n}
=
S\left(\frac{2}{n-1}J_{n-1}-I_{n-1}\right).
\)
It is known that, in the Grover-walk setting, PST
among complete graphs occurs only for \(K_2\) and \(K_3\)
\cite[Sec.~4.3]{KubotaSegawa2022}. Hence \(K_n\) does not exhibit PST
for \(n\geq4\). In Proposition~\ref{prop:noPST_Kn_cyclic}, we further
establish this non-transfer property directly for the walk considered
here.
\begin{lemma}\label{prop:noPST_Kn_cyclic}
Let $K_n$, $n\ge4$, be the complete graph with the cyclic shunt decomposition
$A(K_n)=P+P^2+\cdots+P^{n-1}$, where $P$ is the cyclic shift matrix ($P^n=I_n$).
Let $d=n-1$ and
\[
S_{K_n}=\sum_{r=1}^{d}P^r\otimes E_{rr},\qquad
C_d=\tfrac2dJ_d-I_d,\qquad
U_{K_n}=S_{K_n}(I_n\otimes C_d).
\]
Let $u=\tfrac1{\sqrt d}\mathbf 1_d$ and $\psi_{a_i}=e_{a_i}\otimes u$. Then for distinct vertices ${a_i},a_j$,
\(
U_{K_n}^{\,k}\psi_{a_i}\neq\gamma\,\psi_{a_j}\qquad(k\ge1,\ |\gamma|=1).
\)
Hence, $K_n$ has no PST between distinct vertices $a_i$ and $a_j$.
\end{lemma}

\begin{proof}
Let
\(
\omega=e^{2\pi i/n},\quad
z_r=\omega^r\quad(r=1,\ldots,d),\quad
D=\operatorname{diag}(z_1,\ldots,z_d).
\)
Let
\(
f_\ell=\frac1{\sqrt n}\,(\omega^{\ell j})_{j=1}^{n},\quad \ell=0,\ldots,n-1,
\)
so that \(Pf_\ell=\omega^{\ell}f_\ell\), and let \(F\) be the unitary matrix with columns
\(f_0,\ldots,f_{n-1}\). Suppose, for contradiction, that
\(
U_{K_n}^{\,k}\psi_{a_i}=\gamma\psi_{a_j},
\quad a_i\neq a_j,\quad k\ge1,\quad |\gamma|=1.
\)
The Fourier transform \(F^\dagger\otimes I_d\) on the vertex factor diagonalizes each power of \(P\)
and leaves \(I_n\otimes C_d\) invariant. Hence
\[
U_{K_n}\sim\bigoplus_{\ell=0}^{n-1}U_\ell,
\qquad
U_\ell=D_\ell C_d,
\qquad
D_\ell=\operatorname{diag}(\omega^{\ell},\omega^{2\ell},\ldots,\omega^{d\ell}).
\]
Since
\(
(F^\dagger\otimes I_d)\psi_{a_i}=\frac1{\sqrt n}\sum_{\ell=0}^{n-1}\omega^{-a\ell}\,(f_\ell\otimes u),
\)
the assumed relation \(U_{K_n}^{k}\psi_{a_i}=\gamma\psi_{a_j}\) gives,
\[
U_\ell^{\,k}u=\gamma\,\omega^{(a_i-a_j)\ell}u=\gamma_\ell u,\qquad|\gamma_\ell|=1.
\]
In particular \(U_1^{\,k}u=\gamma_1u\), \(D_1=\operatorname{diag}(\omega^{1},\omega^{2},\ldots,\omega^{d})\) and since
\(C_d=\frac2d\mathbf 1\mathbf 1^T-I_d\),
\(
U_1=-D_1+\frac2d\,D_1\mathbf 1\mathbf 1^{T}.
\)
Suppose \(U_1v=\lambda v\) with \(v\neq0\) and \(\mathbf 1^Tv=0\). Then \(U_1v=-D_1v\), so
\(
(\lambda+z_r)v_r=0,\quad r=1,\ldots,d.
\)
Since \(z_1,\ldots,z_d\) are distinct, \(\lambda=-z_r\) for at most one \(r\), so \(v\) has at most one
nonzero entry \(v_r\). Then \(\mathbf 1^Tv=v_r=0\), so \(v=0\), a contradiction.
Thus every eigenvector \(v\) of \(U_1\) satisfies \(\mathbf 1^Tv\neq0\), i.e.\ has nonzero overlap with
\(u\) (since \(u^\dagger v=\tfrac1{\sqrt d}\mathbf 1^Tv\)).
Since \(U_1\) is unitary, its eigenspaces are mutually orthogonal. If
an eigenspace had dimension greater than one, it would contain a
nonzero vector orthogonal to \(u\), which is impossible. Hence every
eigenspace is one-dimensional. Let
\(
v_1,\ldots,v_d
\)
be an orthonormal eigenbasis with
\(
U_1v_j=\lambda_jv_j.
\)
Since \(u\) has nonzero projection onto every eigenspace,
\(
u=\sum_{j=1}^d c_jv_j,
\quad c_j\neq0.
\)
Thus
\[
U_1^{\,k}u
=
\sum_{j=1}^d c_j\lambda_j^kv_j
=
\gamma_1\sum_{j=1}^dc_jv_j,
\]
and hence
\(
\lambda_j^k=\gamma_1
\quad
(j=1,\ldots,d).
\)
Therefore
\(
U_1^{\,k}=\gamma_1 I_d.
\)
Define
\(
\widetilde v=\Bigl(\frac{z_1}{z_1-1},\ldots,\frac{z_d}{z_d-1}\Bigr)^{T}.
\)
Using \(\sum_{r=1}^{n-1}\frac1{1-z_r}=\frac{n-1}2\) and \(\frac z{z-1}=1-\frac1{1-z}\), we get
\(\mathbf 1^T\widetilde v=\frac d2\). Hence
\[
(U_1\widetilde v)_r=-z_r\widetilde v_r+\frac2d\,z_r\,\mathbf 1^T\widetilde v
=-\frac{z_r^2}{z_r-1}+z_r=-\frac{z_r}{z_r-1}=-\widetilde v_r,
\]
so \(U_1\widetilde v=-\widetilde v\). Thus, \(-1\) is an eigenvalue of \(U_1\), and \(U_1^{\,k}=\gamma_1I_d\)
gives \(\gamma_1=(-1)^k\). Therefore
\(
(-\lambda_j)^k=(-1)^k\lambda_j^{k}=1\quad(j=1,\ldots,d),
\)
so every \(\lambda_j\) is a root of unity, hence an algebraic integer.
Since \(\sum_{r=1}^{n-1}\omega^r=-1\),
\[
\operatorname{tr}(U_1)=\Bigl(\frac2d-1\Bigr)\sum_{r=1}^{n-1}\omega^r=\frac{n-3}{n-1}=1-\frac2{n-1}.
\]
Since \(\operatorname{tr}(U_1)=\sum_j\lambda_j\), it must be an
algebraic integer. But
\(
\frac{n-3}{n-1}
=
1-\frac2{n-1}
\)
is rational and satisfies
\(
0<\frac{n-3}{n-1}<1
\quad (n\ge4),
\)
so it is not an integer and therefore cannot be an algebraic integer.
This contradiction proves that PST cannot occur. Hence
\(
{U_{K_n}^{\,k}\psi_{a_i}\neq\gamma\psi_{a_j}\quad(a_i\neq a_j,\ k\ge1,\ |\gamma|=1).}
\)
Thus, there is no PST from \(\psi_{a_i}=e_{a_i}\otimes\frac1{\sqrt d}\mathbf 1_d\) to \(\psi_{a_j}\) for \(a_i\neq a_j\).
\end{proof}
We next show that the complete graph with a loop at every vertex,
whose adjacency matrix is \(J_n\), exhibits PST when
\(n=2^m\). The proof follows the same argument as in
Lemma 8.2 of~\cite{Katuwal2026pst}.
\begin{lemma}\label{lem:PST_Jn_Grover}
Let $n=2^m, m\geq 2$ and let $P$ be the cyclic shift matrix $n\times n$, so that $J_n=\sum_{r=0}^{n-1}P^r$ is the adjacency
matrix of the complete graph with a loop at every vertex. With the shunt decomposition $P_r=P^{r-1}$
($r=1,\dots,n$), let
\[
S_{J_n}=\sum_{r=0}^{n-1}P^r\otimes E_{r+1,r+1},\qquad
C_n=\tfrac2nJ_n-I_n,\qquad
U_{J_n}=S_{J_n}(I_n\otimes C_n).
\]
Then
\(
U_{J_n}^{\,n}=P^{n/2}\otimes I_n .
\)
Consequently, $U_{J_n}^{\,n}(e_{a_i}\otimes\phi)=e_{a_{i+n/2}}\otimes\phi$ for every coin state $\phi\in\mathbb C^n$
(indices mod $n$), so the walk has perfect state transfer at time $k=n$ between the vertices $a_i$ and $a_{i+n/2}$.
\end{lemma}

\begin{proof}
Let
\(
\omega=e^{2\pi i/n},\quad
f_\ell=\frac1{\sqrt n}(\omega^{\ell j})_{j=0}^{n-1},\quad \ell=0,\ldots,n-1,
\)
so that $Pf_\ell=\omega^{\ell}f_\ell$. The Fourier transform in the vertex factor diagonalizes every
$P^r$ and leaves $I_n\otimes C_n$ unchanged. Hence, it blocks diagonalizes $U_{J_n}$ as
\[
U_{J_n}\sim\bigoplus_{\ell=0}^{n-1}U_\ell,\qquad
U_\ell=D_\ell C_n,\qquad
D_\ell=\operatorname{diag}(1,\omega^{\ell},\omega^{2\ell},\ldots,\omega^{(n-1)\ell}).
\]
We prove that for $n=2^m$, $\;U_\ell^{\,n}=(-1)^{\ell}I_n$ for $\ell=0,\ldots,n-1$.
Each $U_\ell$ is a product of a diagonal unitary and a reflection, hence unitary and therefore diagonalizable.
So $U_\ell^{\,n}$ is determined by the $n$-th powers of its eigenvalues.
\emph{Case $\ell=0$.} $U_0=C_n$ and $C_n^2=I_n$. Since $n$ is even, $U_0^{\,n}=I_n$.
\emph{Case $1\le\ell\le n-1$.} Let $g=\gcd(\ell,n)$ and $M=n/g$. Since $n=2^m$ and $1\le\ell\le n-1$,
$g$ is a power of $2$ with $g<n$. Hence $M\ge2$ is a power of $2$, in particular even. Moreover $g=1$ if $\ell$ is odd,
and $g$ is even if $\ell$ is even. The diagonal entries $d_j=\omega^{\ell j}$ of $D_\ell$ are the $M$-th roots of unity,
each repeated $g$ times.
We have $U_\ell=-D_\ell+\frac2nD_\ell\mathbf 1\mathbf 1^{T}$. By the matrix determinant lemma, for
\(\lambda\notin\{-d_1,\ldots,-d_n\}\),
\[
\det(\lambda I-U_\ell)
=
\det(\lambda I+D_\ell)
\left(
1-\frac2n\sum_{j=0}^{n-1}
\frac{d_j}{\lambda+d_j}
\right).
\]
Since $M$ is even, $\prod_{\zeta^M=1}(\lambda+\zeta)=\lambda^M-1$, so $\det(\lambda I+D_\ell)=(\lambda^M-1)^g$. Also
\[
\sum_{\zeta^M=1}\frac{\zeta}{\lambda+\zeta}=M-\lambda\frac{M\lambda^{M-1}}{\lambda^M-1}=-\frac{M}{\lambda^M-1},
\]
and each root occurs $g$ times, so $\sum_j\frac{d_j}{\lambda+d_j}=-\frac{n}{\lambda^M-1}$. Therefore
\[
\det(\lambda I-U_\ell)=(\lambda^M-1)^g\Bigl(1+\frac2{\lambda^M-1}\Bigr)=(\lambda^M-1)^{g-1}(\lambda^M+1).
\]
 If $\ell$ is odd, then $g=1$, $M=n$, and every eigenvalue satisfies $\lambda^n=-1$. Hence $U_\ell^{\,n}=-I_n$.
 If $\ell$ is even, then $g$ is even and every eigenvalue satisfies $\lambda^M=\pm1$, so
$\lambda^n=(\lambda^M)^g=1$. Hence $U_\ell^{\,n}=I_n$.
 Since $\omega^{\ell n/2}=(-1)^\ell$, we have $P^{n/2}f_\ell=(-1)^\ell f_\ell$. Thus
$U_{J_n}^{\,n}$ and $P^{n/2}\otimes I_n$ both act as the scalar $(-1)^\ell$ on $f_\ell\otimes\mathbb C^n$ for every $\ell$, so
\[
U_{J_n}^{\,n}=P^{n/2}\otimes I_n.
\]
Finally, $Pe_j=e_{j-1}$, so $P^{n/2}e_{a_i}=e_{a_{i-n/2}}=e_{a_{i+n/2}}$ (mod $n$), and
\(
U_{J_n}^{\,n}(e_{a_i}\otimes\phi)=e_{a_{i+n/2}}\otimes\phi\quad\text{for every }\phi\in\mathbb C^n .
\)
Thus, PST occurs at time $k=n$ between $a_i$ and $a_{i+n/2}$.
\end{proof}
\begin{Theorem}\label{thm:PST_partial_join_Kn}
Let \(n=2^m\), \(m\ge2\), and let
\(
G=K_n\overset{(1,1)}{\vec{\vee}}K_n,
\)
equivalently,
\(
G=K_n\ltimes\mathcal J_1,
\)
where
\(
\mathcal J_1=\mathcal J_2
\)
and each of \(\mathcal J_1,\mathcal J_2\) has one outgoing arc from
each vertex. Then \(G\) exhibits PST at time
\(k=n\).
\end{Theorem}
\begin{proof}
Consider the complementary coupling with loops defined by
\(
A(\mathcal J_1)=J_n-A(K_n).
\)
Since
\(
A(K_n)=J_n-I_n,
\)
we have
\(
A(\mathcal J_1)=I_n.
\)
Thus, each of \(\mathcal J_1\) and \(\mathcal J_2\) is \(1\)-regular,
while the coupling graph
\(\mathcal J=\mathcal J_1\cup\mathcal J_2\) is \(1\)-regular directed graph. In
particular,
\(
2\neq n-1=\deg(K_n),
\)
so the coupling degree need not equal the degree of the component
graphs.
Moreover,
\(
A(K_n)+A(\mathcal J_1)=J_n,
\)
 hence, the corresponding double-cover decomposition in
\eqref{eq.41} gives
\[
U_G^{\,n}
=
\frac12
\begin{pmatrix}
U_{J_n}^{\,n}+U_{K_n-\mathcal J_1}^{\,n}
&
U_{J_n}^{\,n}-U_{K_n-\mathcal J_1}^{\,n}
\\[4pt]
U_{J_n}^{\,n}-U_{K_n-\mathcal J_1}^{\,n}
&
U_{J_n}^{\,n}+U_{K_n-\mathcal J_1}^{\,n}
\end{pmatrix}.
\]
By Lemma~\ref{lem:PST_Jn_Grover}, the Grover walk on \(J_n\),
for \(n=2^m\), satisfies
\(
U_{J_n}^{\,n}(e_{a_i}\otimes\phi)
=
e_{a_{i+n/2}}\otimes\phi
\)
for every \(\phi\in\mathbb C^n\). By
Lemma~\ref{PST_Double_Cover},
the signed component satisfies the same transfer relation,
\[
U_{K_n-\mathcal J_1}^{\,n}(e_{a_i}\otimes\phi)
=
e_{a_{i+n/2}}\otimes\phi,
\]
because \(\mathcal E^2=I_n\) and the PST relation for \(J_n\) holds for
arbitrary coin states.

Therefore,
\[
U_G^{\,n}
\begin{pmatrix}
e_{a_i}\otimes\phi\\
0
\end{pmatrix}
=
\frac12
\begin{pmatrix}
U_{J_n}^{\,n}(e_{a_i}\otimes\phi)
+
U_{K_n-\mathcal J_1}^{\,n}(e_{a_i}\otimes\phi)
\\
U_{J_n}^{\,n}(e_{a_i}\otimes\phi)
-
U_{K_n-\mathcal J_1}^{\,n}(e_{a_i}\otimes\phi)
\end{pmatrix}
\\
=
\frac12
\begin{pmatrix}
2(e_{a_{i+n/2}}\otimes\phi)
\\
0
\end{pmatrix}
\\
=
\begin{pmatrix}
e_{a_{i+n/2}}\otimes\phi\\
0
\end{pmatrix}.
\]
Hence
\(
{
U_G^{\,n}
\begin{pmatrix}
e_{a_i}\otimes\phi\\
0
\end{pmatrix}
=
\begin{pmatrix}
e_{a_{i+n/2}}\otimes\phi\\
0
\end{pmatrix}.
}
\)
Thus, the partial join \(G\) exhibits perfect state transfer at time
\(k=n\) between the antipodal vertices \(a_i\) and
\(a_{i+n/2}\). In contrast, by
Lemma~\ref{prop:noPST_Kn_cyclic}, the individual \(K_n\) does not
exhibit PST for \(n\ge4\).
\end{proof}
Hence, this construction provides an example in which a suitable
partial join exhibits PST even though the individual component
\(K_n\) does not (see Fig.~\ref{fig:K4-partial-join-matching}).
\begin{figure}[t]
\centering
\begin{tikzpicture}[
scale=1.0, transform shape,
every node/.style={circle, draw, fill=white, inner sep=1.2pt, minimum size=15pt, font=\small},
AtoB/.style={blue, semithick, ->, >=stealth, dash pattern=on 2pt off 1.2pt},
BtoA/.style={red,  semithick, ->, >=stealth, dash pattern=on 2pt off 1.2pt},
Gedge/.style={thick, <->, >=stealth}
]
\node (a1) at (0, 1.5)  {$a_1$};
\node (a2) at (1, 0)    {$a_2$};
\node (a3) at (0, -1.5) {$a_3$};
\node (a4) at (-1, 0)   {$a_4$};
\draw[Gedge] (a1)--(a2);
\draw[Gedge] (a2)--(a3);
\draw[Gedge] (a3)--(a4);
\draw[Gedge] (a4)--(a1);
\draw[Gedge] (a1)--(a3);   
\draw[Gedge] (a2)--(a4);   
\node (b1) at (4, 1.5)  {$b_1$};
\node (b2) at (5, 0)    {$b_2$};
\node (b3) at (4, -1.5) {$b_3$};
\node (b4) at (3, 0)    {$b_4$};
\draw[Gedge] (b1)--(b2);
\draw[Gedge] (b2)--(b3);
\draw[Gedge] (b3)--(b4);
\draw[Gedge] (b4)--(b1);
\draw[Gedge] (b1)--(b3);   
\draw[Gedge] (b2)--(b4);   
\draw[AtoB] (a1) to[bend left=15] (b1);
\draw[BtoA] (b1) to[bend left=15] (a1);
\draw[AtoB] (a3) to[bend left=15] (b3);
\draw[BtoA] (b3) to[bend left=15] (a3);
\draw[AtoB] (a2) to[bend left=30] (b2);
\draw[BtoA] (b2) to[bend left=30] (a2);
\draw[AtoB] (a4) to[bend left=30] (b4);
\draw[BtoA] (b4) to[bend left=30] (a4);
\begin{scope}[shift={(6.2,0.5)}]
    \draw[Gedge] (0,0.8) -- (0.6,0.8) node[draw=none, fill=none, right, font=\scriptsize] {edges of $G_i$};
    \draw[blue, semithick, ->, >=stealth] (0,0.4) -- (0.6,0.4) node[draw=none, fill=none, right, font=\scriptsize] {$a_i \to b_i$};
    \draw[red, semithick, ->, >=stealth] (0,0) -- (0.6,0) node[draw=none, fill=none, right, font=\scriptsize] {$b_i \to a_i$};
\end{scope}
\end{tikzpicture}
\caption{The directed partial join $K_4\overset{(1,1)}{\vec\vee}K_4$ with matching
coupling $a_i\to b_i$ and $b_i\to a_i$ ($i=1,\dots,4$). Each vertex of $G_1$ and $G_2$ sends and
receives exactly one arc, so the graph $\mathcal{J}=\mathcal{J}_1\cup\mathcal{J}_2$ is
perfect matchings ($1$-regular bipartite subgraphs of $K_{4,4}$), matching the $(1,1)$ degree label.}
\label{fig:K4-partial-join-matching}
\end{figure}
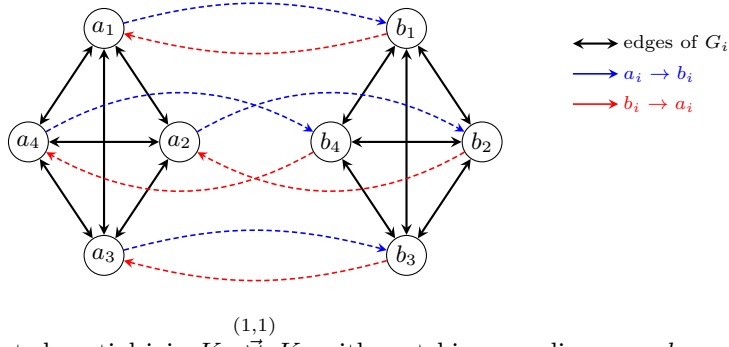
\begin{ex}{The case \(n=4\)}
Let
\(
P=
\begin{pmatrix}
0&1&0&0\\
0&0&1&0\\
0&0&0&1\\
1&0&0&0
\end{pmatrix},
\quad
P^4=I_4.
\)
Then
\(
J_4=I_4+P+P^2+P^3.
\)
Thus, we use the cyclic shunt decomposition
\(
P_1=I_4,\quad P_2=P,\quad P_3=P^2,\quad P_4=P^3,
\)
with shift operator
\(
S_{J_4}
=
I_4\otimes E_{11}
+P\otimes E_{22}
+P^2\otimes E_{33}
+P^3\otimes E_{44}.
\)
Let
\(
C_4=\frac12J_4-I_4
=
\frac12
\begin{pmatrix}
-1&1&1&1\\
1&-1&1&1\\
1&1&-1&1\\
1&1&1&-1
\end{pmatrix},
\)
and define
\(
U_{J_4}
=
S_{J_4}(I_4\otimes C_4).
\)
Equivalently,
\(
U_{J_4}
=
\frac12
\begin{pmatrix}
-I_4&I_4&I_4&I_4\\
P&-P&P&P\\
P^2&P^2&-P^2&P^2\\
P^3&P^3&P^3&-P^3
\end{pmatrix}.
\)
A direct multiplication gives
\[
{
U_{J_4}^{\,4}
=
P^2\otimes I_4
=
\begin{pmatrix}
0&0&I_4&0\\
0&0&0&I_4\\
I_4&0&0&0\\
0&I_4&0&0
\end{pmatrix}.
}
\]
Since
\(
P^2e_{a_i}=e_{a_{i+2}},
\quad i\pmod 4,
\)
we obtain, for every \(\phi\in\mathbb C^4\),
\(
{
U_{J_4}^{\,4}(e_{a_i}\otimes\phi)
=
e_{a_{i+2}}\otimes\phi.
}
\)
Hence the walk on \(J_4\) exhibits PST at time
\(k=4\), with
\(
a_1\leftrightarrow a_3,
\quad
a_2\leftrightarrow a_4.
\)
Now take the complementary-loop coupling
\(
A(\mathcal J_1)=I_4.
\)
Then
\[
A(K_4-\mathcal J_1)
=
A(K_4)-A(\mathcal J_1)
=
J_4-2I_4
=
P+P^2+P^3-I_4.
\]
Let
\(
\mathcal E=
\operatorname{diag}(1,1,1,-1),
\)
and define the signed shift
\[
S_{K_4-\mathcal J_1}
=
P\otimes E_{11}
+P^2\otimes E_{22}
+P^3\otimes E_{33}
-I_4\otimes E_{44}.
\]
Following Lemma~\ref{PST_Double_Cover}, the corresponding
transition operator is
\(
U_{K_4-\mathcal J_1}
=
S_{K_4-\mathcal J_1}
(I_4\otimes C_4\mathcal E).
\)
Since
\(
S_{K_4-\mathcal J_1}
=
(I_4\otimes\mathcal E)S_{J_4},
\)
we have
\[
U_{K_4-\mathcal J_1}
=
(I_4\otimes\mathcal E)
S_{J_4}
(I_4\otimes C_4)
(I_4\otimes\mathcal E)
=
(I_4\otimes\mathcal E)
U_{J_4}
(I_4\otimes\mathcal E).
\]
Therefore
\(
U_{K_4-\mathcal J_1}^{\,4}
=
(I_4\otimes\mathcal E)
U_{J_4}^{\,4}
(I_4\otimes\mathcal E).
\)
Using
\(
U_{J_4}^{\,4}=P^2\otimes I_4
\)
and \(\mathcal E^2=I_4\), we obtain
\[
{
U_{K_4-\mathcal J_1}^{\,4}
=
P^2\otimes I_4
=
\begin{pmatrix}
0&0&I_4&0\\
0&0&0&I_4\\
I_4&0&0&0\\
0&I_4&0&0
\end{pmatrix}.
}
\]
Consequently,
\(
{
U_{K_4-\mathcal J_1}^{\,4}
(e_{a_i}\otimes\mathcal E\phi)
=
e_{a_{i+2}}\otimes\mathcal E\phi
}
\)
for every \(\phi\in\mathbb C^4\). Hence, the signed component also
exhibits PST at time \(k=4\) from
\(
a_1\leftrightarrow a_3,
\quad
a_2\leftrightarrow a_4.
\)
\end{ex}
\section{Conclusion}\label{sec9}
In this paper, we construct transition operators for directed partial
join graphs with signed coupling within the shunt-decomposition
framework. We establish necessary and sufficient conditions for internal
PST, coupling PST, and periodicity. We further develop a double-cover
construction to derive conditions for PST and periodicity in the
resulting graphs, including cases in which the associated transition
operators do not necessarily commute. The framework applies to several
important graph families, including complete graphs with loops,
circulant partial joins, complete bipartite graphs, tensor powers of
complete bipartite graphs, and double covers of complete graphs with
loops. In particular, we provide an example in which the complete graph
\(K_n\), \(n\geq4\), does not exhibit PST, whereas a suitable partial
join of \(K_n\) exhibits PST. Moreover, the partial join operation can induce PST even when neither constituent graph admits PST individually.
 A summary of the graph families
considered in this work, together with their transition-operator
structures and the corresponding PST results is given in
Table~\ref{tab:graph-families-pst}.

\begin{table}[h]
\centering
\small
\begin{tabular}{|p{3.6cm}|p{4.6cm}|p{5.6cm}|}
\hline
\textbf{Transition operator} & \textbf{PST results} &
\textbf{Graph families } \\
\hline
$U_{G_1}=U_{G_2}=U_{\mathcal J_1}=U_{\mathcal J_2}$
&
Theorems~\ref{thm:pst-coupling} and~\ref{thm:pst-internal}
&
$K_n^{\circlearrowleft}$~\eqref{signed adjacency matrices};
circulant partial joins (Prop.~\ref{prop:general-even-n})
\\
\hline
$U_{G_1}=U_{G_2}=0$, \; $U_{\mathcal J_1}=U_{\mathcal J_2}$
&
Theorems~\ref{thm:pst-coupling-phi-pi2} and~\ref{thm:pst-internal-phi-pi2}
&
$K_{n,n}$~\eqref{signed adjacency matrices};
$(K_{n,n})^{\otimes n}$ (Prop.~\ref{prop:offdiag-perm-sum});
$K_2^{\otimes n}$ (Cor.~\ref{cor.4.8})
\\
\hline
$U_{G_1}=U_{G_2}$, \; $U_{\mathcal J_1}=U_{\mathcal J_2}$
\newline (with $A(\mathcal J_1)+A(G_1)=J_n$)
&
If $d_{G_1}=d_{\mathcal J_1}$: Corollary~\ref{cor_combine_thm_and_switching_phi} (signed coupling).
\newline
If $d_{G_1}\neq d_{\mathcal J_1}$: Theorem~\ref{thm:PST_partial_join_Kn} (not signed).
&
Double cover of $K_n^{\circlearrowleft}$ (Eq.~\eqref{coupling_loops_complementary}).
\newline
The directed partial join
\(
G=K_n\overset{(1,1)}{\vec{\vee}}K_n
\)
\\
\hline
\end{tabular}
\caption{Transition-operator structures, corresponding PST results, and graph families
in the directed partial join framework for even $n$. Here $d_{G_1}$ and $d_{\mathcal J_1}$
denote the degrees of $G_1$ and $\mathcal J_1$, and $A(\mathcal J_1)+A(G_1)=J_n$ in the third row.}
\label{tab:graph-families-pst}
\end{table}
\printbibliography

@article{angeles2009perfect,
  title={Perfect state transfer, integral circulants and join of graphs},
  author={R. J. Angeles-Canul and R. Norton and M. Opperman and C. Paribello and M. Russell and C. Tamon},
  journal={arXiv preprint arXiv:0907.2148},
  year={2009}
}

@book{Godsil_Zhan_2023,
  title     = {Discrete Quantum Walks on Graphs and Digraphs},
  author    = {C. Godsil and H. Zhan},
  year      = {2023},
  series    = {London Mathematical Society Lecture Note Series},
  publisher = {Cambridge University Press},
  address   = {Cambridge},
}

@book{GodsilRoyle2001,
  author    = {C. Godsil and G. Royle},
  title     = {Algebraic Graph Theory},
  publisher = {Springer},
  year      = {2001},
}

@article{godsil2021sedentary,
  title={Sedentary quantum walks},
  author={C. Godsil},
  journal={Linear Algebra and its Applications},
  volume={614},
  pages={356--375},
  year={2021},
}

@article{stiebitz1993colouring,
  title={On colouring partial joins of a complete graph and a cycle},
  author={M. Stiebitz and W. Wessel},
  journal={Mathematische Nachrichten},
  volume={163},
  number={1},
  pages={109--116},
  year={1993},
}

@inproceedings{Aharonov2001,
  author    = {D. Aharonov and A. Ambainis and J. Kempe and U. Vazirani},
  title     = {Quantum Walks on Graphs},
  booktitle = {Proceedings of the 33rd Annual ACM Symposium on Theory of Computing (STOC)},
  pages     = {50--59},
  year      = {2001},
  publisher = {ACM},
}

@article{GodsilZhan2019,
  author  = {C. Godsil and H. Zhan},
  title   = {Discrete-time Quantum Walks and Graph Structures},
  journal = {Journal of Combinatorial Theory, Series A},
  volume  = {167},
  pages   = {181--212},
  year    = {2019},
  issn    = {0097-3165},
}

@article{coutinho2016perfect,
  title={Perfect state transfer in products and covers of graphs},
  author={G. Coutinho and C. Godsil},
  journal={Linear and Multilinear Algebra},
  volume={64},
  number={2},
  pages={235--246},
  year={2016},
}

@article{vlachou2018quantum,
  title={Quantum key distribution with quantum walks: C. Vlachou et al.},
  author={C. Vlachou and W. Krawec and P. Mateus and N. Paunkovi{\'c} and A. Souto},
  journal={Quantum Information Processing},
  volume={17},
  number={11},
  pages={288},
  year={2018},
}

@book{portugal2013quantum,
  title={Quantum walks and search algorithms},
  author={R. Portugal},
  volume={19},
  year={2013},
  publisher={Springer}
}

@article{childs2009universal,
  title={Universal computation by quantum walk},
  author={A. M. Childs},
  journal={Physical Review Letters},
  volume={102},
  number={18},
  pages={180501},
  year={2009},
  publisher={APS}
}

@article{lovett2010universal,
  title={Universal quantum computation using the discrete-time quantum walk},
  author={N. B. Lovett and S. Cooper and M. Everitt and M. Trevers and V. Kendon},
  journal={Physical Review A---Atomic, Molecular, and Optical Physics},
  volume={81},
  number={4},
  pages={042330},
  year={2010},
  publisher={APS}
}

@article{childs2013universal,
  title={Universal computation by multiparticle quantum walk},
  author={A. M. Childs and D. Gosset and Z. Webb},
  journal={Science},
  volume={339},
  number={6121},
  pages={791--794},
  year={2013},
  publisher={American Association for the Advancement of Science}
}

@inproceedings{childs2003exponential,
  title={Exponential algorithmic speedup by a quantum walk},
  author={A. M. Childs and R. Cleve and E. Deotto and E. Farhi and S. Gutmann and D. A. Spielman},
  booktitle={Proceedings of the thirty-fifth annual ACM symposium on Theory of computing},
  pages={59--68},
  year={2003}
}

@article{Ambainis2003,
  author  = {A. Ambainis},
  title   = {Quantum Walks and Their Algorithmic Applications},
  journal = {International Journal of Quantum Information},
  volume  = {1},
  number  = {4},
  pages   = {507--518},
  year    = {2003},
}

@article{Vlachou2015,
  author  = {C. Vlachou and J. Rodrigues and P. Mateus and N. Paunkovi{\'c} and A. Souto},
  title   = {Quantum Walk Public-Key Cryptographic System},
  journal = {International Journal of Quantum Information},
  volume  = {13},
  number  = {7},
  pages   = {1550050},
  year    = {2015},
}

@book{CoutinhoGodsil2016,
  author       = {G. Coutinho and C. Godsil},
  title        = {Graph Spectra and Continuous Quantum Walks},
  year         = {2016},
  note         = {Monograph/Book draft},
  institution  = {University of Waterloo},
}

@article{FarhiGutmann1998,
  author  = {E. Farhi and S. Gutmann},
  title   = {Quantum Computation and Decision Trees},
  journal = {Physical Review A},
  volume  = {58},
  number  = {2},
  pages   = {915--928},
  year    = {1998},
}

@article{zhan2021quantum,
  title={Quantum walks on embeddings},
  author={H. Zhan},
  journal={Journal of Algebraic Combinatorics},
  volume={53},
  number={4},
  pages={1187--1213},
  year={2021},
}

@article{WingBocanegra2023,
  author    = {A. Wing-Bocanegra and S. E. Venegas-Andraca},
  title     = {Circuit Implementation of Discrete-Time Quantum Walks via the Shunt Decomposition Method},
  journal   = {Quantum Information Processing},
  volume    = {22},
  number    = {3},
  pages     = {146},
  year      = {2023},
}

@inproceedings{sato2024circuit,
  title={Circuit Implementation of Discrete-Time Quantum Walks on Complex Networks},
  author={R. Sato and K. Saito},
  booktitle={2024 IEEE International Conference on Quantum Computing and Engineering (QCE)},
  volume={2},
  pages={376--377},
  year={2024},
}

@article{WingBocanegra2025,
  author    = {A. Wing-Bocanegra and C. E. Quintero-Narvaez and S. E. Venegas-Andraca},
  title     = {Circuit Implementation and Analysis of a Quantum-Walk Based Search Complement Algorithm},
  journal   = {Scientific Reports},
  volume    = {15},
  number    = {1},
  pages     = {4865},
  year      = {2025},
}

@inproceedings{Katuwal2026,
  author    = {B. Katuwal and S. M. S. Srinath and Y. Lakshmi Naidu},
  title     = {Graph-Decomposition Based Shift Matrices: Applications to Quantum Channels and Circuit Implementations},
  booktitle = {2026 International Conference on Next-Gen Quantum and Advanced Computing (NQComp)},
  year      = {2026},
  pages     = {753--760},
  publisher = {IEEE},
}

@article{panda2023quantum,
  title={Quantum direct communication protocol using recurrence in k-cycle quantum walks},
  author={S. S. Panda and P. A. A. Yasir and C. M. Chandrashekar},
  journal={Physical Review A},
  volume={107},
  number={2},
  pages={022611},
  year={2023},
}

@article{mulken2011continuous,
  title={Continuous-time quantum walks: Models for coherent transport on complex networks},
  author={O. M{\"u}lken and A. Blumen},
  journal={Physics Reports},
  volume={502},
  number={2-3},
  pages={37--87},
  year={2011},
}

@article{christandl2004perfect,
  title={Perfect state transfer in quantum spin networks},
  author={M. Christandl and N. Datta and A. Ekert and A. J. Landahl},
  journal={Physical Review Letters},
  volume={92},
  number={18},
  pages={187902},
  year={2004},
}

@article{godsil2012state,
  title={State transfer on graphs},
  author={C. Godsil},
  journal={Discrete Mathematics},
  volume={312},
  number={1},
  pages={129--147},
  year={2012},
}

@article{Godsil2011,
  author    = {C. Godsil},
  title     = {Periodic Graphs},
  journal   = {The Electronic Journal of Combinatorics},
  volume    = {18},
  number    = {1},
  pages     = {P23},
  year      = {2011},
}

@article{ChanZhan2023,
  author    = {A. Chan and H. Zhan},
  title     = {Pretty Good State Transfer in Discrete-Time Quantum Walks},
  journal   = {Journal of Physics A: Mathematical and Theoretical},
  volume    = {56},
  number    = {16},
  pages     = {165305},
  year      = {2023},
}

@article{Katuwal2026pst,
  author       = {B. Katuwal and S. M. S. Srinath and Y. Lakshmi Naidu and S. Dutta},
  title        = {Perfect State Transfer on Quotient Graphs in Shunt Decomposition-Based Quantum Walks},
  journal      = {arXiv preprint},
  volume       = {arXiv:2606.24440},
  year         = {2026},
}

@article{Bose2009,
  author    = {S. Bose and A. Casaccino and S. Mancini and S. Severini},
  title     = {Communication in XYZ All-to-All Quantum Networks with a Missing Link},
  journal   = {International Journal of Quantum Information},
  volume    = {7},
  number    = {4},
  pages     = {713--723},
  year      = {2009},
}

@article{ge2011perfect,
  title={Perfect state transfer, graph products and equitable partitions},
  author={Y. Ge and B. Greenberg and O. Perez and C. Tamon},
  journal={International Journal of Quantum Information},
  volume={9},
  number={03},
  pages={823--842},
  year={2011},
}

@article{angeles2009quantum,
  title={Quantum perfect state transfer on weighted join graphs},
  author={R. J. Angeles-Canul and R. M. Norton and M. C. Opperman and C. C. Paribello and M. C. Russell and C. Tamon},
  journal={International Journal of Quantum Information},
  volume={7},
  number={08},
  pages={1429--1445},
  year={2009},
}

@article{alvir2016perfect,
  title={Perfect state transfer in Laplacian quantum walk},
  author={R. Alvir and S. Dever and B. Lovitz and J. Myer and C. Tamon and Y. Xu and H. Zhan},
  journal={Journal of Algebraic Combinatorics},
  volume={43},
  number={4},
  pages={801--826},
  year={2016},
}

@article{arezoomand2023perfect,
  title={Perfect state transfer on semi-Cayley graphs over abelian groups},
  author={M. Arezoomand},
  journal={Linear and Multilinear Algebra},
  volume={71},
  number={14},
  pages={2337--2353},
  year={2023},
}

@article{2481614.2481624,
  author = {J. Brown and C. Godsil and D. Mallory and A. Raz and C. Tamon},
  title = {Perfect state transfer on signed graphs},
  year = {2013},
  issue_date = {May 2013},
  volume = {13},
  number = {5--6},
  issn = {1533-7146},
  journal = {Quantum Info. Comput.},
  month = {may},
  pages = {511--530},
  numpages = {20},
}

@article{kirkland2026quantum,
  title={Quantum walks on join graphs},
  author={S. Kirkland and H. Monterde},
  journal={Discrete Mathematics},
  volume={349},
  number={3},
  pages={114832},
  year={2026},
}

@article{godsil2025perfect,
  title={Perfect state transfer between real pure states},
  author={C. Godsil and S. Kirkland and H. Monterde},
  journal={SIAM Journal on Matrix Analysis and Applications},
  volume={46},
  number={3},
  pages={2093--2115},
  year={2025},
}

@book{brouwer2012spectra,
  title={Spectra of Graphs, Universitext, Springer, New York},
  author={A. E. Brouwer and W. H. Haemers},
  year={2012}
}

@article{zhang2022polarity,
  title={Polarity-based graph neural network for sign prediction in signed bipartite graphs},
  author={X. Zhang and H. Wang and J. Yu and C. Chen and X. Wang and W. Zhang},
  journal={World Wide Web},
  volume={25},
  number={2},
  pages={471--487},
  year={2022},
  publisher={Springer}
}

@article{acharya2016lict,
  title={On $\bullet$-lict signed graphs $L_{\bullet}^{c}(S)$ and
         $\bullet$-line signed graphs $L_{\bullet}(S)$},
  author={M. Acharya and R. Jain and S. Kansal},
  journal={Transactions on Combinatorics},
  volume={5},
  number={1},
  pages={37--48},
  year={2016},
  publisher={Transactions on Combinatorics}
}

@article{KubotaSegawa2022,
  author  = {S. Kubota and E. Segawa},
  title   = {Perfect state transfer in Grover walks between states associated to vertices of a graph},
  journal = {Linear Algebra and its Applications},
  volume  = {646},
  pages   = {238--251},
  year    = {2022}
}
\end{document}